\documentclass[letterpaper, 10 pt, conference]{ieeeconf}

\IEEEoverridecommandlockouts

\usepackage[T1]{fontenc}

\usepackage[utf8]{inputenc}

\usepackage{mathtools}
\usepackage{amssymb}

\newcommand*{\Scale}[2][4]{\scalebox{#1}{$#2$}}%
\usepackage{bm}

\usepackage[]{algorithm2e}

\mathtoolsset{showonlyrefs,showmanualtags}

\usepackage{cite} 

\usepackage[]{units}

\usepackage{setspace}

\usepackage{wasysym}

\usepackage{hyperref}

\usepackage{tikz}
\usepackage{tikz-3dplot}
\usepackage{tkz-euclide}
\usetkzobj{all}
\usetikzlibrary{calc,decorations.markings,,intersections,shapes.misc}

\usepackage{graphicx}

\usepackage{multirow}

\usepackage{placeins}

\definecolor{dkgreen}{rgb}{0,0.6,0}

\usepackage{arydshln}

\usepackage{epstopdf}

\newcommand{\sss}[1]{\scriptscriptstyle {#1}}
\newcommand{\s}[1] {{\scriptscriptstyle #1}}
\newcommand{\tp}{^{\sss{T}}}

\newcommand{\innerproduct}[2]{\left\langle#1,#2\right\rangle}

\newcommand{\reals}[0]{\mathbb{R}}

\newcommand{\Rn}[1][n]{\mathbb{R}^{\sss{#1}}}
\newcommand{\Rpositive}{\mathbb{R}_{\sss{\ge 0}}}

\newcommand{\sk}[1]{\mathcal{S}\left({#1}\right)}

\newcommand{\invsk}[1]{\mathcal{S}^{\sss{-1}}\left({#1}\right)}
\newcommand{\OP}[1]{\Pi\left({#1}\right)}

\newtheorem{thm}{Theorem}

\newtheorem{prop}[thm]{Proposition}
\newtheorem{assum}[thm]{Assumption}

\newtheorem{prob}{Problem}
\newtheorem{rem}[thm]{Remark}
\newtheorem{defn}{Definition}

\newcommand{\amb}{\mathbf{a}}
\newcommand{\bmb}{\mathbf{b}}

\newcommand{\emb}{\mathbf{e}}
\newcommand{\fmb}{\mathbf{f}}

\newcommand{\hmb}{\mathbf{h}}

\newcommand{\mmb}{\mathbf{m}}
\newcommand{\nmb}{\mathbf{n}}

\newcommand{\qmb}{\mathbf{q}}

\newcommand{\umb}{\mathbf{u}}
\newcommand{\vmb}{\mathbf{v}}

\newcommand{\xmb}{\mathbf{x}}
\newcommand{\ymb}{\mathbf{y}}
\newcommand{\zmb}{\mathbf{z}}

\newcommand{\Hmb}{\mathbf{H}}

\newcommand{\Idmat}{\mathbf{I}} 							
\newcommand{\Rmat}{\ensuremath{\mathcal{R}}} 	

\newcommand{\zvec}{\mathbf{0}} 				

\newcommand{\SO}[1][2]{\mathbb{SO}(#1)}

\newcommand{\embi}[1]{\emb_{\sss{#1}}}

\newcommand{\xmbi}[1]{\xmb_{\sss{#1}}}
\newcommand{\xmbiDot}[1]{\dot{\xmb}_{\sss{#1}}}

\newcommand{\qmbi}[1]{\qmb_{\sss{#1}}}

\newcommand{\nmbi}[1]{\nmb_{\sss{#1}}}

\newcommand{\numb}{\bm{\nu}}
\newcommand{\numbi}[1]{\bm{\nu}_{\sss{#1}}}

\newcommand{\zmbi}[1]{\zmb_{\sss{#1}}}

\newcommand{\Rmati}[1]{\Rmat_{\sss{#1}}}

\newcommand{\Omi}[1]{\bm{\omega}_{\sss{#1}}}

\newcommand{\Nset}{\mathcal{N}}
\newcommand{\ConeOpen}{\mathcal{C}}
\newcommand{\ConeClosed}{\bar{\mathcal{C}}}

\newcommand{\Sn}[1]{\mathbb{S}^{\sss{#1}}}

\def\paperextended{0}

\usepackage{subcaption}

\newcommand{\mrtoq}{T_{\sss{r}}^{\sss{q}}}
\newcommand{\mrtoqi}{T_{\sss{r}}^{\sss{qi}}}
\newcommand{\mqtor}{T_{\sss{q}}^{\sss{r}}}

\begin{document}

	\title{A Common Framework for Attitude Synchronization of Unit Vectors in Networks with Switching Topology}
	\date{\today}
	\author{
		Pedro~O.~Pereira, Dimitris~Boskos and Dimos~V.~Dimarogonas
		\thanks{ %
			The authors are with the School of Electrical Engineering, KTH Royal Institute of Technology, SE-100 44, Stockholm, Sweden. %
			Email addresses: \texttt{\{ppereira, boskos, dimos\}@kth.se}. 
			This work was supported by the EU H2020 Research and Innovation Programme under GA No.644128 (AEROWORKS), the Swedish Research Council (VR), the Swedish Foundation for Strategic Research (SSF) and the KAW Foundation.
		}
	}
	
	\maketitle	

	\begin{abstract}
		In this paper, we study attitude synchronization for elements in the unit sphere of $\Rn[\sss{3}]$ and for elements in the $3D$ rotation group, for a network with switching topology. 
The agents' angular velocities are assumed to be the control inputs, and a switching control law for each agent is devised that guarantees synchronization, provided that all elements are initially contained in a given region, unknown to the network.
%
%
The control law is decentralized and it does not require a common orientation frame among all agents.
We refer to synchronization of unit vectors in $\Rn[3]$ as incomplete synchronization, and of $3D$ rotation matrices as complete synchronization.
Our main contribution lies in showing that these two problems can be analyzed under a common framework, where all elements' dynamics are transformed into unit vectors dynamics on a sphere of appropriate dimension.

	\end{abstract}
	
	\section{Introduction}
	Decentralized control in a multi-agent environment has been a topic of active research, with applications in large scale robotic systems. 
Attitude synchronization in satellite formations is one of the relevant applications~\cite{lawton2002synchronized,abdessameud2009attitude}, where the control goal is to guarantee that a network of fully actuated rigid bodies acquires a common attitude. 
Coordination of underwater vehicles in ocean exploration missions~\cite{leonard2007collective}, and of unmanned aerial vehicles in aerial exploration missions, may also be casted as attitude synchronization problems.

In the literature, attitude synchronization strategies for elements in the special orthogonal group are found in~\cite{hatanaka2012passivity,1656474, sarlette2009autonomous, dimarogonas2009leader, bai2008rigid, tron2012intrinsic, chung2013phase, bondhus2005leader, krogstad2006coordinated}, which focus on \emph{complete} attitude synchronization; 
and in~\cite{oh2014formation,olfati2006swarms, moshtagh2007distributed, paley2009stabilization, li2014unified, sarlette2008global, sarlette2009synchronization, dorfler2014synchronization, moreau2004stability, moreau2005stability,zhang2011general}, which focus on \emph{incomplete} attitude synchronization.
In this paper, we focus on both \emph{complete} and \emph{incomplete} attitude synchronization. 
We refer to \emph{incomplete} attitude synchronization when the agents are unit vectors in $\Rn[3]$, a space also called $\Sn{2}$; and we refer to \emph{complete} attitude synchronization, when the agents are $3D$ rotation matrices, a space also called $\SO[3]$.
Incomplete synchronization represents a relevant practical problem, when the goal among multiple agents is to share a common direction. 
In flocking, for example, moving along a common direction is a requirement.
Also, in a network of satellites, whose antennas are to point in a common direction, incomplete synchronization may be more important than complete.

In~\cite{1656474,sarlette2009autonomous,dimarogonas2009leader,bai2008rigid,chung2013phase,bondhus2005leader,krogstad2006coordinated}, state dependent control laws for torques are presented which guarantee synchronization for elements in $\SO[3]$, while in~\cite{pereira2015Synchronization,paley2009stabilization,li2014unified} state dependent control laws for torques are presented which guarantee synchronization for elements in $\Sn{2}$. 
In these works, all agents are dynamic and their angular velocities are part of the state of the system, rather than control inputs.
In this paper, however, we consider kinematic agents and we design control laws, for the agents' angular velocities, which are not exclusively state dependent, but are also time dependent, with the time dependency encapsulating the case of a switching network topology.

We note that, regarding synchronization in $\SO[3]$, relevant results are found in~\cite{thunberg2014distributed,igarashi2009passivity,dorfler2014synchronization,sepulchre2011consensus,sarlette2009consensus}.
%
%
%
In~\cite{sepulchre2011consensus,dorfler2014synchronization,sarlette2009consensus,lin2007state}, consensus on non-linear spaces is analyzed with the help of a common weak non-smooth Lyapunov function, i.e., a  Lyapunov function which is non-increasing along solutions.
Also, in~\cite{thunberg2014distributed}, control laws which guarantee synchronization under a switching topology are presented, under the hypothesis of a dwell time between consecutive switches.

In our proposed framework, we relax the assumption of a dwell time by providing conditions for synchronization under average dwell time.
Our approach is based on the construction of a common weak non-smooth Lyapunov function for analyzing synchronization in $\Sn{n}$.
In order to handle the non-smoothness of the proposed Lyapunov function, we present an invariance-like result (see also~\cite{hespanha2004uniform, bacciotti2005invariance, mancilla2006extension, fischer2013lasalle} for invariance like theorems for switched systems).
We propose control laws for angular velocities of unit vectors in $\Rn[3]$ and $3D$ rotation matrices that guarantee synchronization for a network of agents with switching topology.
The control laws devised for unit vectors and rotation matrices achieve different goals, and differ in two aspects worth emphasizing. First, controlling rotation matrices requires more measurements when compared with controlling unit vectors; secondly, while controlling rotation matrices requires full actuation, i.e., all body components of the angular velocity need to be controllable, controlling unit vectors does not.
Our main contribution compared to the aforementioned literature lies in analyzing both problems under a common framework, in order to allow for a unified stability analysis using the same common weak Lyapunov function. 
Particularly both problems are transformed into synchronization problems in $\Sn{m}$ for an appropriate $m \in \mathbb{N}$.
Since rotation matrices can be parametrized by unit quaternions~\cite{field1984spacecraft}, which are unit vectors in $\Rn[\sss{4}]$, these are chosen for the analysis of the proposed control law.
We also note that consensus in $\Rn[\sss{n}]$ can be casted as a synchronization problem in $\Sn{n}$.
%
%
We note that under our framework, we do not require a dwell time between consecutive switches. 
A preliminary version of this work was presented at the 2016 IEEE Conference on Decision and Control~\cite{pereira2016CDCSynchronization}. 
With respect to this preliminary version, this paper provides significant additional results. 
In particular,  more details on the derivation of the main theorems and propositions are presented; additional figures illustrating concepts and results are provided; and supplementary simulations are presented which further illustrate the theoretical results.

The remainder of this paper is structured as follows. 
In Sections~\ref{sec:Notation} and~\ref{sec:Preliminaries}, we present the notation used in this manuscript and conditions on vector fields that guarantee convergence to the consensus set, respectively. 
In Section~\ref{sec:Synchronization}, we describe the common framework for analysis of both synchronization in $\Sn{2}$ and $\SO[3]$.
In Sections~\ref{sec:SynchronizationInS2} and \ref{sec:SynchronizationInSO3}, the control laws for synchronization in $\SO[3]$ and $\Sn{2}$ are presented, respectively, and the agents dynamics are transformed into the common framework.
In Section~\ref{sec:Analysis}, asymptotic synchronization is established for the  common framework vector field.
In Section~\ref{sec:Simulations}, illustrative simulations are presented.
	
	\section{Notation}
	\label{sec:Notation}
%
In what follows, let $n$ and $m$ be two integers.
$\innerproduct{\cdot}{\cdot}: \Rn[n] \times \Rn[n] \mapsto \Rn[]$ denotes the inner product in $\Rn[n]$.
$\Rn[\sss{m \times n}]$ denotes the set of linear mappings from $\Rn[n]$ to $\Rn[m]$, and with $\Idmat_{\sss{n}} \in \Rn[\sss{n \times n}]$ as the identity matrix and $\zvec_{\sss{m \times n}} \in \Rn[\sss{m \times n}]$ as the zero matrix.
%
%
%
We denote by $M_{\sss{3,3}}$ and $\bar{M}_{\sss{3,3}}$ the space of symmetric matrices and antisymmetric matrices in $\Rn[\sss{3 \times 3}]$, respectively. 
Given $\amb,\bmb \in \Rn[3]$, the matrix $\sk{\amb} \in \bar{M}_{\sss{3,3}}$ is the skew-symmetric matrix that satisfies $\sk{\amb} \, \bmb = \amb \times \bmb$;
moreover, given any $A \in \bar{M}_{\sss{3,3}}$, $\mathcal{S}^{\sss{-1}}(A) \in \Rn[3]$ denotes the vector for which $\sk{\invsk{A}} = A$.
We denote by $\Sn{n} = \{\xmb \in \Rn[\sss{n+1}]: \innerproduct{\xmb}{\xmb} =1 \}$ the set of unit vectors in $\Rn[\sss{n+1}]$.
The map $\Pi: \Sn{n} \ni \xmb \mapsto \OP{\xmb} \in \reals^{\sss{(n+1)\times (n+1)}}$, yields a matrix that satisfies $\OP{\xmb} \ymb = \ymb - \innerproduct{\ymb}{\xmb} \xmb$ for any $\ymb \in \Rn[3]$, and it represents the orthogonal projection operator onto the subspace orthogonal to $\xmb \in \Sn{n}$.	
$\embi{1}, \cdots, \embi{n} \in \Sn{n-1} \subset \Rn[n]$ denote the canonical basis vectors in $\Rn[n]$. 
%
%
Given $r \ge 0$, we denote $\mathcal{B}(r) = \{ \xmb \in \Rn[n] : \|\xmb\| < r \}$ and $\bar{\mathcal{B}}(r) = \{ \xmb \in \Rn[n] : \|\xmb\| \le r \}$ as the open and closed balls of radius $r$ and centered around $\zvec$, respectively.
%
%
Given a set $\mathcal{H}$, we use the notation $|\mathcal{H}|$ for the cardinality of $\mathcal{H}$.
Given a function $f: A \ni a \mapsto f(a) \in B$, for some arbitrary domain $A$ and codomain $B$: 
if $f$ is differentiable $df$ denotes its derivative; 
and given any $N \in \mathbb{N}$, we denote
$
	f^{\sss{N}}: 
	A^{\sss{N}} \ni a := (a_{\sss{1}},\cdots,a_{\sss{N}}) 
	\mapsto 
	f^{\sss{N}}(a) := (f(a_{\sss{1}}),\cdots,f(a_{\sss{N}}))\in B^{\sss{N}}
$.
Finally, given a manifold $M$, $T_{\sss{\mmb}} M$ denotes the tangent space to $M$ at $\mmb \in M$.

	\section{Preliminaries}
	\label{sec:Preliminaries}
	We consider a network of $N \in \mathbb{N}$ agents, with $n \in \mathbb{N}$ the dimension of the space which the agents belong to.
We associate bijectively the (finite) set of network graphs to the set $\{1,2,\cdots\} =: \mathcal{P} \subset \mathbb{N}$, and consider an average dwell time switching signal $\sigma: \Rn[{}]_{\sss{\ge 0}} \mapsto \mathcal{P} $~\cite{mancilla2006extension}.
We then consider a trajectory
\begin{align}
	 \Rpositive \ni t \mapsto \xmb(t) := (\xmbi{1}(t), \cdots , \xmbi{N}(t)) \in (\Rn[n])^{\sss{N}}
	\label{eq:Solution}
\end{align}
which satisfies 
\begin{align}
	\dot{\xmb}(t)
	=
	\fmb_{\sss{\sigma(t)}}(\xmb(t)) ,
	\forall t \in \Rpositive,
	\label{eq:DifferentiaEquation}
\end{align}
and where, for any $p \in \mathcal{P}$,
\begin{align}
	\fmb_{\sss{p}} 
	:=
	(
		\fmb_{\sss{1,p}},
		\cdots,
		\fmb_{\sss{N,p}}
	)
	:
	(\Rn[n])^{\sss{N}}
	\mapsto
	(\Rn[n])^{\sss{N}}.
	\label{eq:VectorField}
\end{align}
We now present conditions on the vector fields $\{\fmb_{\sss{p}}\}_{\sss{p \in \mathcal{P}}}$ in~\eqref{eq:VectorField}  which guarantee that $\xmb$ in~\eqref{eq:Solution} converges to the consensus set $\mathcal{C} = \{(\xmbi{1},\cdots,\xmbi{N}) \in (\Rn[n])^{\sss{N}}: \xmbi{1} = \cdots = \xmbi{N} \}$.
In later sections, when studying synchronization in $\Sn{2}$ and $\SO[3]$, given the proposed control laws, we verify that the dynamics of all agents satisfy these conditions, allowing us to refer to the results in this section. 
%

Given $\xmb = (\xmbi{1},\cdots,\xmbi{N}) \in (\Rn[n])^{\sss{N}}$, denote 
\begin{align}
	\mathcal{H}(\xmb) := \arg \max_{j \in \Nset}( \|\xmb_{\sss{j}}\|) \subseteq \Nset.
	\label{eq:SetH}
\end{align}
Then, given $\xmb = (\xmbi{1},\cdots,\xmbi{N}) \in (\Rn[n])^{\sss{N}}$ and $i \in \mathcal{H}(\xmb)$, $\|\xmbi{i}\|$ is larger than or equal to $\|\xmbi{j}\|$ for all $j \in \Nset$.
Given $\xmb = (\xmbi{1},\cdots,\xmbi{N}) \in (\Rn[n])^{\sss{N}}$ assume that
\begin{align}
			\underset{p \in \mathcal{P} }{\max} 
			\,
			\innerproduct{\xmbi{i}}{\fmb_{\sss{i,p}}(\xmb)}
			\le 0,
			\text{ for all }
			i \in \mathcal{H}(\xmb),
			\label{eq:Deltaxmax}
\end{align}
which quantifies an upper bound on the time derivative of $\frac{1}{2}\|\xmbi{i}\|^2$ along a solution that satisfies~\eqref{eq:DifferentiaEquation}, since $\frac{d}{dt} \frac{1}{2}\|\xmbi{i}(t)\|^2 = \innerproduct{\xmbi{i}(t)} {\fmb_{\sss{i,\sigma(t)}}(\xmb(t))}$, for all time instants $t$ where the derivative is well defined.
On the other hand, given $\xmb \not\in \mathcal{C}$, assume that~\eqref{eq:Deltaxmax} holds and that
\begin{align}
	\forall p \in \mathcal{P} \exists i \in \mathcal{H}(\xmb) : 
	\innerproduct{\xmbi{i}}{\fmb_{\sss{i,p}}(\xmb)}  < 0.
	\label{eq:ConditionInvariance}
\end{align}
Condition~\eqref{eq:ConditionInvariance} implies that, for every switching signal $p \in \mathcal{P}$, and within the set $\{\frac{1}{2}\|\xmbi{i}\|^2\}_{\sss{i \in \mathcal{H}(\xmb)}}$, one can always find an element whose time derivative along a solution of~\eqref{eq:DifferentiaEquation}, if well defined, is negative.
Conditions~\eqref{eq:Deltaxmax} and~\eqref{eq:ConditionInvariance} are central for the proof of the following theorem, which establishes asymptotic convergence of~\eqref{eq:Solution} to $\mathcal{C}$ and constitutes the main result of this section. 
\begin{thm}
	\label{thm:MainTheorem}
	Consider~\eqref{eq:Solution} and assume that for certain $r>0$ and all $\xmb  = (\xmbi{1},\cdots,\xmbi{N})\in \mathcal{B}(r)^{\sss{N}}$ the following hold:
	\begin{enumerate}
		\item when $\xmb \not\in \mathcal{C}$,
		\begin{enumerate}
			\item %
			$\underset{p \in \mathcal{P} }{\max}  
			\innerproduct{\xmbi{i}}{\fmb_{\sss{i,p}}(\xmb)} \le 0$, for all $i \in \mathcal{H}(\xmb)$ ,
			\item %
			$\forall p \in \mathcal{P} 
			\exists i \in \mathcal{H}(\xmb) : 
			\innerproduct{\xmbi{i}}{\fmb_{\sss{i,p}}(\xmb)} < 0$,
		\end{enumerate}	
		\item when $\xmb \in \mathcal{C} $, %
		$\innerproduct{\xmbi{i}}{\fmb_{\sss{i,p}}(\xmb)} = \zvec$ for all $i\in \Nset $ and $p \in \mathcal{P} $.
	\end{enumerate}
	Then, the set  $\bar{\mathcal{B}}(r_{\sss{0}})^{\sss{N}}$, with  $r_{\sss{0}} := \max_{\sss{i \in \Nset}} \| \xmbi{i}(0)\|<r$, is positively invariant.
	Moreover, for each initial condition $\xmb(0)\in \mathcal{B}(r)^{\sss{N}}$, and given $V : (\Rn[n])^{\sss{N}} \ni \xmb = (\xmbi{1},\cdots,\xmbi{N}) \mapsto V(\xmb) := \max_{\sss{i \in \Nset}} \frac{1}{2} \|\xmbi{i}\|^{\sss{2}} \in \Rn[]_{\sss{\ge 0}}$, there exists a constant $V^{\sss\infty} \in [0 , V(\xmb(0))]$ such that $\lim_{\sss{t \rightarrow \infty}} V(\xmb(t)) = V^{\sss{\infty}}$.
	Finally,~\eqref{eq:Solution} converges asymptotically to $V^{\sss{-1}}(V^{\sss{\infty}}) \cap \mathcal{C} \subset  \bar{\mathcal{B}}(r_{\sss{0}})^{\sss{N}} \cap \mathcal{C}$.
\end{thm}
The proof is found in Appendix~\ref{app:AppendixProofConvergenceToC}.
	
	\section{Synchronization}
	\label{sec:Synchronization}
	In the next subsections, we study synchronization of agents in $\Sn{2}$ and $\SO[3]$.
More specifically, we first present feedback control laws for the angular velocities of the agents, with which we determine the closed loop dynamics.
Afterwards, by means of appropriate transformations, those dynamics are rewritten in a common form that allows us to study synchronization in $\Sn{2}$ and $\SO[3]$ under a common framework.
Additionally, in Appendix~\ref{app:ConsensusRn}, we also show that consensus in $\Rn[n]$ can be casted as a synchronization problem in $\Sn{n}$, for any $n \in \mathbb{N}$.
\begin{defn}
	\label{def:HyperConeUnitVector2}
	Given $n \in \mathbb{N}$, $\alpha \in [0,\pi]$ and $\bar{\bm{\nu}} \in \Sn{n}$, the open $\alpha$-cone $\ConeOpen(\alpha,\bar{\bm{\nu}})$ is defined as $\ConeOpen(\alpha,\bar{\bm{\nu}}) := \{\bm{\nu} \in \Sn{n}: \innerproduct{\bar{\bm{\nu}}} {\bm{\nu}} > \cos(\alpha) \}$, representing the set of unit vectors that are $\alpha$ close to the unit vector $\bar{\bm{\nu}}$.
	Similarly, we define the closed $\alpha$-cone $\ConeClosed(\alpha,\bar{\bm{\nu}}) := \{\bm{\nu} \in \Sn{n}: \innerproduct{\bar{\bm{\nu}}} {\bm{\nu}} \ge \cos(\alpha) \}$.
\end{defn}
In Fig.~\ref{fig:Cone30}, we illustrate a closed cone, for $n = 3$, where three unit vectors $\numbi{1}$, $\numbi{2}$ and $\numbi{3}$ are contained in the closed $30^{\circ}$-cone formed by the unit vector $\bar{\numb}$.
\begin{figure}
	\centering
	\includegraphics[clip=true,trim=0cm 0cm 0cm 0cm,width=0.45\textwidth]
	{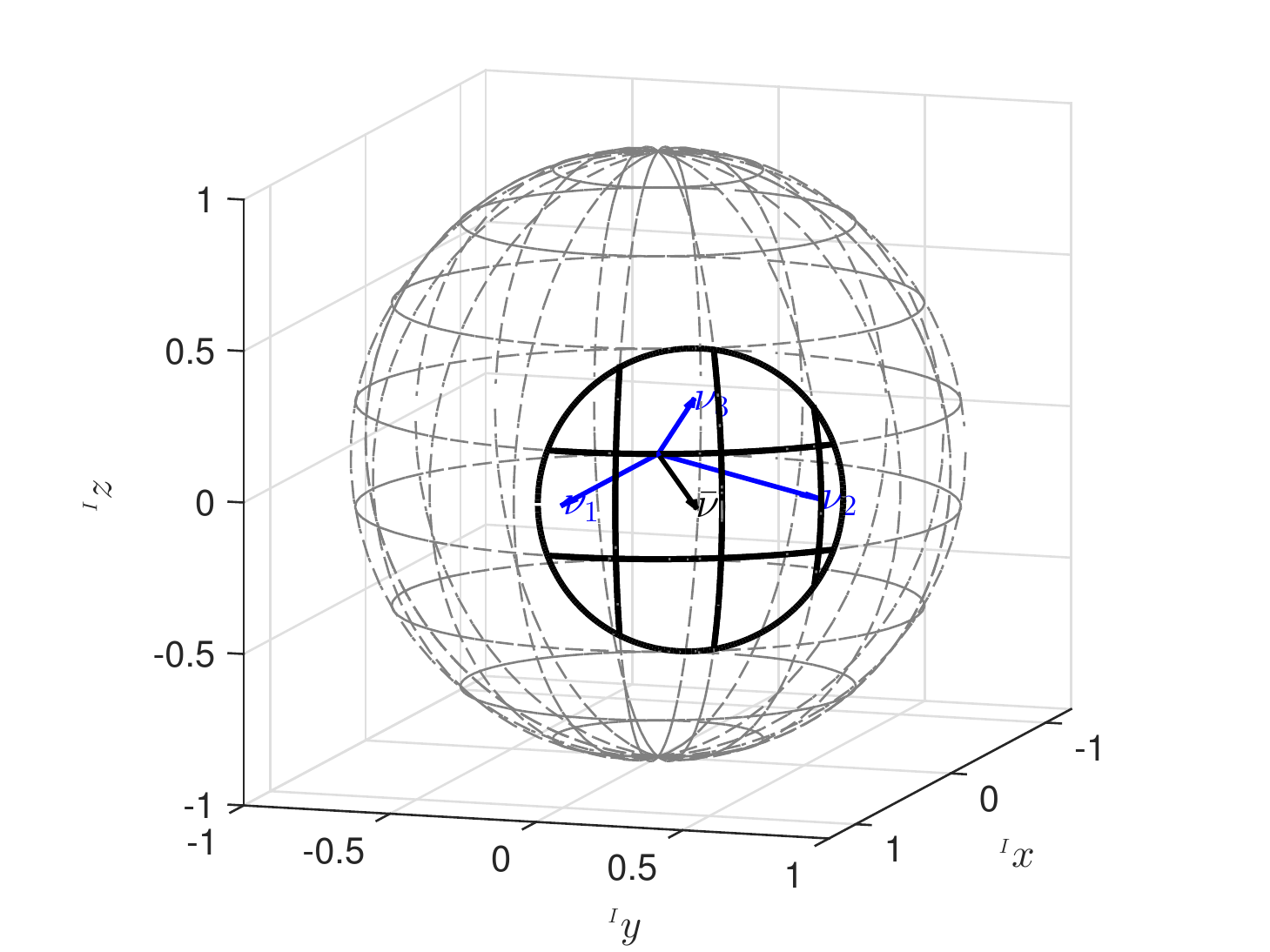}
	\caption{Three unit vectors, $\numbi{1}$, $\numbi{2}$ and $\numbi{3}$, in $\Rn[3]$ contained in closed $30^{\circ}$-cone formed by the unit vector $\bar{\numb}$.}
	\label{fig:Cone30}
\end{figure}

%
Consider a group of unit vectors $\numb = (\bm{\nu}_{\sss{1}} , \cdots , \bm{\nu}_{\sss{N}}) \in (\Sn{n})^{\sss{N}}$, for some $N,n \in \mathbb{N}$. 
We say that $\numb$ belongs to an open (closed) $\alpha$-cone, for some $\alpha \in [0,\pi]$, if $\exists \bar{\bm{\nu}} \in \Sn{n} : \bm{\nu} \in \ConeOpen(\alpha,\bar{\bm{\nu}})^{\sss{N}} (\ConeClosed(\alpha,\bar{\bm{\nu}})^{\sss{N}})$.
We say that $\numb$ is synchronized if $\bm{\nu}_{\sss{1}} = \cdots = \bm{\nu}_{\sss{N}}$. 
%
We show later that, given the proposed control laws, synchronization of a group of unit vectors takes place asymptotically if the group of unit vectors is initially contained in an open $\alpha^{\sss{\star}}$-cone, where $\alpha^{\sss{\star}} = \frac{\pi}{2}$ for synchronization in $\Sn{2}$ and consensus in $\Rn[n]$; and where $\alpha^{\sss{\star}} = \frac{\pi}{4}$ for synchronization in $\SO[3]$.

In the next subsections, we always consider a group of $N$ agents, indexed by the set $\Nset := \{ 1, \cdots, N\}$, operating in either $\Sn{2}$, $\SO[3]$ or $\Rn[n]$.
The agents' network graph is modeled as a time varying digraph, $\Rn[]_{\sss{+}} \ni t \mapsto \mathcal{G}(\sigma(t)) = \{\Nset, \mathcal{E}(\sigma(t))\}$, with $\sigma: \Rn[{}]_{\sss{\ge 0}} \mapsto \mathcal{P} $ as the switching signal, and with $\mathcal{G}(p)$ and $\mathcal{E}(p)$ as the graph and edges' set corresponding to the switching signal $p \in \mathcal{P}$ (where $|\mathcal{P}| \le 2^{\sss{N(N-1)}}$). 
We also denote $\Nset_{\sss{i}}(p) \subset \Nset$ as the neighbor set of agent $i \in \Nset$ corresponding to the switching signal $p \in \mathcal{P}$; 
and, for convenience, we also denote $\{ j_{\sss{1}}, \cdots, j_{\sss{| \Nset_{\sss{i}}(q)|}} \} \equiv \Nset_{\sss{i}}(q)$.

In order to perform analysis under a common framework, we transform all problems' dynamics into a standard form, which we describe next.
Given $n \in \mathbb{N}$, we denote $\bm{\nu} = (\bm{\nu}_{\sss{1}} , \cdots , \bm{\nu}_{\sss{N}}) : \Rn[]_{\sss{\ge 0}} \mapsto (\Sn{n})^{\sss{N}}$ as the state of a group of unit vectors in $\Sn{n}$, which evolves according to the dynamics
\begin{align}
	\dot{\bm{\nu}}(t)
	=
	\tilde{\fmb}_{\sss{\sigma(t)}}(\bm{\nu}(t))
	=
	\begin{bmatrix}
		\tilde{\fmb}_{\sss{1,\sigma(t)}}(\bm{\nu}(t))
		\\
		\vdots
		\\
		\tilde{\fmb}_{\sss{N,\sigma(t)}}(\bm{\nu}(t))
	\end{bmatrix},
	\numb(0) \in (\Sn{n})^{\sss{N}}
	\label{eq:GeneralForm}
\end{align}
where $\tilde{\fmb}_{\sss{i,\sigma(t)}} : (\Sn{n})^{\sss{N}} \ni \bm{\nu}:=(\bm{\nu}_{\sss{1}},\cdots,\bm{\nu}_{\sss{N}}) \mapsto \tilde{\fmb}_{\sss{i,\sigma(t)}}(\bm{\nu}) \in T_{\sss{\bm{\nu}_{\sss{i}}}} \Sn{n} \subset \Rn[n+1]$ is defined as
\begin{align}
	\tilde{\fmb}_{\sss{i,\sigma(t)}}(\bm{\nu})
	:=
	\sum_{j \in \Nset_{\sss{i}}(\sigma(t))} 
	\tilde{w}_{\sss{ij}}(\bm{\nu}_{\sss{i}},\bm{\nu}_{\sss{j}})
	\OP{\numbi{i}}\numbi{j},
	\label{eq:GeneralForm2}
\end{align}
for all $i\in\Nset$; i.e., $\dot{\bm{\nu}}_{\sss{i}}(t) = \tilde{\fmb}_{\sss{i,\sigma(t)}}(\bm{\nu}(t))$.
Notice that indeed $\langle \numbi{i}, \tilde{\fmb}_{\sss{i,p}}(\bm{\nu}) \rangle = 0$ for all $ i \in \Nset$, $p \in \mathcal{P}$ and $\bm{\nu} = (\bm{\nu}_{\sss{1}},\cdots,\bm{\nu}_{\sss{N}}) \in (\Sn{n})^{\sss{N}}$, which implies that the set $(\Sn{n})^{\sss{N}}$ is positively invariant with respect to~\eqref{eq:GeneralForm}.

The system~\eqref{eq:GeneralForm}-\eqref{eq:GeneralForm2} is the standard form all problems are transformed into:
for synchronization in $\Sn{2}$, $\bm{\nu} :\Rn[]_{\sss{\ge 0}} \mapsto  (\Sn{2})^{\sss{N}}$; 
for synchronization in $\SO[3]$, $\bm{\nu} :\Rn[]_{\sss{\ge 0}} \mapsto (\mathcal{S}^{\sss{3}})^{\sss{N}}$; 
and for consensus in $\Rn[n]$, $\bm{\nu} :\Rn[]_{\sss{\ge 0}} \mapsto (\Sn{n})^{\sss{N}}$. 

The functions $\tilde{w}_{\sss{ij}}: \ConeOpen(\alpha,\bar{\bm{\nu}})^{\sss{2}}  \mapsto \Rn[{}]_{\sss{\ge 0}}$ in~\eqref{eq:GeneralForm2} are continuous weight functions, for some $\alpha \in [\frac{\pi}{2},\pi]$ and $\bar{\bm{\nu}} \in \Sn{n}$.
Thus, given $(\bm{\nu}_{\sss{i}},\bm{\nu}_{\sss{j}}) \in \ConeOpen(\alpha,\bar{\bm{\nu}})^{\sss{2}}$, $\tilde{w}_{\sss{ij}}(\bm{\nu}_{\sss{i}},\bm{\nu}_{\sss{j}})$ is the weight agent $i$ assigns to the deviation between itself and its neighbor $j$, for all $i,j \in \Nset$ (and where we emphasize that the agents are within the same cone). 
All functions $\tilde{w}_{\sss{ij}}$ are assumed to satisfy the following condition,
\begin{align}
	\hspace{-0.3cm}
	\tilde{w}_{\sss{ij}}(\bm{\nu}_{\sss{i}},\bm{\nu}_{\sss{j}}) > 0 
	\forall 
	(\bm{\nu}_{\sss{i}},\bm{\nu}_{\sss{j}}) \in \ConeOpen(\alpha,\bar{\bm{\nu}})^{\sss{2}} 
	\text{ with }
	\innerproduct{\bm{\nu}_{\sss{i}}}{\bm{\nu}_{\sss{j}}} \ne 1
	.
	\label{eq:gTildeCase}
\end{align}
Thus, from continuity, it follows that the weight between two neighbors is zero if and only if they are synchronized, though the weight may be arbitrarily small when the neighbors are arbitrarily close to each other or when the neighbors are close the boundaries of the domain of the weight functions.

The dependency of the dynamics~\eqref{eq:GeneralForm}-\eqref{eq:GeneralForm2} on time comes from the time varying network graph, and more specifically, the time varying neighbor set of each agent, as specified in~\eqref{eq:GeneralForm2}.

Although the results of the paper remain valid under other types of switching (see proof of Theorem~\ref{thm:MainTheorem}), in practical cases the switching instants of each agent's control law cannot accumulate to a certain time value.
In order to formulate this observation as an assumption, we adopt Definition 2.3. in~\cite{mancilla2006extension}, and say that the switching signal $\sigma$ has an average dwell-time $\tau_{\sss{D}} > 0$ and a chatter bound $N_{\sss{0}} \in \mathbb{N}$ if the number of switching times of $\sigma$ in any open finite interval $(t_{\sss{1}},t_{\sss{2}}) \subset \Rpositive$ is upper bounded by $N_{\sss{0}} + \frac{t_{\sss{2}} - t_{\sss{1}}}{\tau_{\sss{D}}}$.
\begin{assum}
	\label{assump:SwitchesAgents}
	For each agent $i \in \Nset$, the time-varying neighbor set $\Rn[]_{\sss{+}} \ni t \mapsto \Nset_{\sss{i}}(\sigma(t)) \subset \Nset$ switches with an average dwell-time $\tau_{\sss{D}}^{\sss{i}} >0 $ and a chatter bound $N_{\sss{0}}^{\sss{i}} \in \mathbb{N} $.
\end{assum}
As explained in more detail in the next sections, each agent~$i \in \Nset$ is in charge of providing the input that follows from composing the proposed output feedback control laws with the measurements made at each time instant.
Thus, requiring an agent's neighbor set to switch with an average dwell time guarantees that the agent's input does not experience infinite many discontinuities in any time interval of finite length. 
In fact, if Assumption~\ref{assump:SwitchesAgents} is satisfied, then $\sigma$ has an average dwell time, which allows us to invoke the results from Section~\ref{sec:Preliminaries}.
\begin{prop}
	If each agent's $i\in \Nset$ neighbor set  $\Rn[]_{\sss{+}} \ni t \mapsto \Nset_{\sss{i}}(\sigma(t))$ switches with an average dwell-time $\tau_{\sss{D}}^{\sss{i}}$ and chatter bound $N_{\sss{0}}^{\sss{i}}$, then the network dynamics~\eqref{eq:GeneralForm} has a switching signal with average dwell time $\tau_{\sss{D}} = \frac{1}{N} \min_{\sss{i \in \Nset}} \tau_{\sss{D}}^{\sss{i}}$ and chatter bound $N_{\sss{0}} = N \max_{\sss{i \in \Nset}}N_{\sss{0}}^{\sss{i}}$.
\end{prop}
A proof is found in Appendix~\ref{app:AppendixAverageDwellTime}.

	\subsection{Complete synchronization in $\SO[3]$ casted as synchronization in $\Sn{3}$}
	\label{sec:SynchronizationInSO3}
Consider a group of $N$ agents 
$
	\Rmati{1}, \ldots, \Rmati{N} \in \SO[3] 
	:=  
	\{
		\Rmat \in \Rn[3 \times 3]: 
		\Rmat\tp \Rmat = \Rmat \Rmat\tp =\Idmat, 
		\det(\Rmat) =1 
	\}
$, 
where, for every  $i \in N$, $\Rmati{i}$ represents the orientation frame, w.r.t. an unknown inertial orientation frame, of agent $i$. 
We say that the agents are synchronized if they all share the same \emph{complete} orientation, i.e., if $\Rmati{1} = \cdots = \Rmati{N}$, as illustrated in Fig.~\ref{fig:2RigidBodies} for $N =2 $.
\begin{figure}
	\centering
    \begin{subfigure}[b]{\linewidth}
    	\centering
        \includegraphics[clip=true,trim=2cm 3.5cm 0.5cm 3.2cm, width=0.75\linewidth]{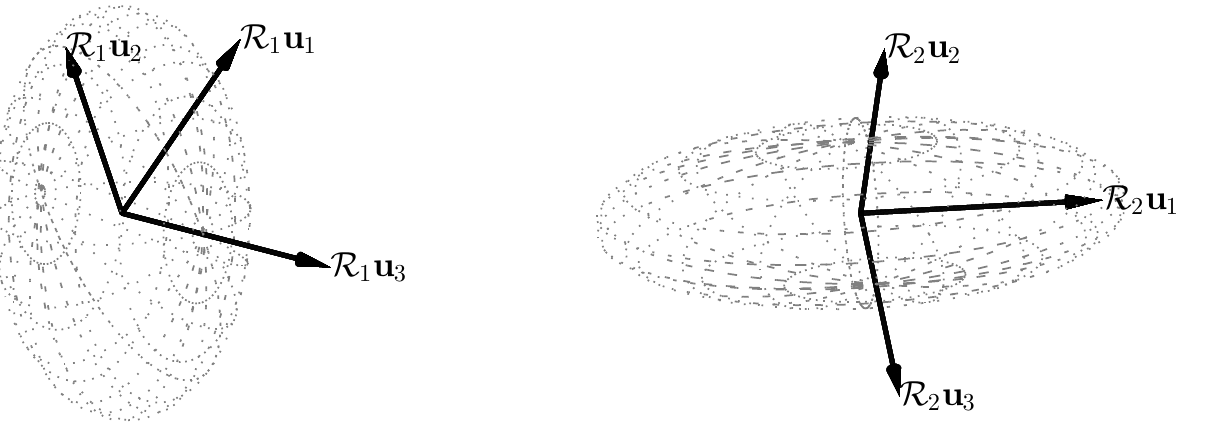}
        \caption{Two agents not synchronized.}
        \label{subfig:Ellipsoids_Synchronized}
    \end{subfigure}	   
    \begin{subfigure}[b]{\linewidth}
    	\centering
        \includegraphics[clip=true,trim=2cm 3.5cm 0.5cm 3.2cm, width=0.75\linewidth]{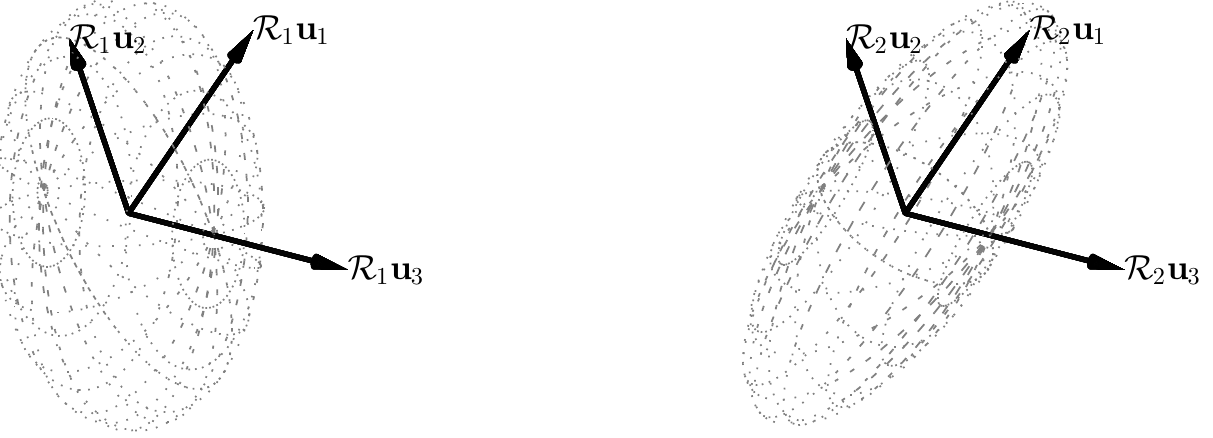}
        \caption{Two agents synchronized.}
        \label{subfig:Ellipsoids_Synchronized_Not}
    \end{subfigure}	
	\caption{In complete synchronization, $N$ agents, indexed by $i = \{1,\cdots,N\}$, synchronize their rotation matrices ($\umb_{\sss{1}}$,$\umb_{\sss{2}}$ and $\umb_{\sss{3}}$ stand for the canonical basis vectors in $\Rn[3]$).}
	\label{fig:2RigidBodies}
\end{figure}
The term \emph{complete} synchronization is used in juxtaposition with \emph{incomplete} synchronization as described in the next subsection.
In incomplete synchronization, rather than synchronizing all three bodies axes, the agents synchronize only one body direction, and, as explained in the next section, complete synchronization does not guarantee incomplete synchronization (and vice-versa).

For each $i \in \Nset$, $\bm{\omega}_{\sss{i}}:  \Rn[]_{\sss{\ge 0}} \mapsto \Rn[3]$ denotes the body-framed angular velocity of agent $i$, which can be actuated.
Each rotation matrix $\Rmat_{\sss{i}}: \Rn[]_{\sss{\ge 0}} \ni t \mapsto \Rmat_{\sss{i}}(t) \in \SO[3]$ evolves according to $\dot{\Rmat}_{\sss{i}}(t) = \fmb_{\sss{\Rmat}}(\Rmat_{\sss{i}}(t),\bm{\omega}_{\sss{i}}(t))$ where
\begin{align}
	&
	\fmb_{\sss{\Rmat}}:
	\SO[3] \times \Rn[3] 
	\ni
	(\Rmat,\bm{\omega})
	\mapsto
	\fmb_{\sss{\Rmat}}(\Rmat,\bm{\omega})
	\in T_{\Rmat} \SO[3] 
	\\
	&
	\fmb_{\sss{\Rmat}}(\Rmat,\bm{\omega})
	:=
	\Rmat \sk{\bm{\omega}}.
	\label{eq:RotationMatrixKinematics}
\end{align}
If, at a time instant $t \in \Rn[]_{\sss{\ge 0}}$, agent $i \in \Nset$ is aware of and can measure the relative attitude between itself and another agent $j$, then $j \in \Nset_{\sss{i}}(\sigma(t))$.
This motivates the definition of the measurement function $\hmb_{\sss{i}}(t,\cdot)$ for each time instant $t \in \Rpositive$, namely 
\begin{align}
	&
	\hmb_{\sss{i}}(t,\cdot): \SO[3]^{\sss{N}} \ni \Rmat \mapsto \hmb_{\sss{i}}(t,\Rmat) \in \SO[3]^{\sss{|\Nset_{\sss{i}}(\sigma(t))|}}
	\\
	&
	\hmb_{\sss{i}}(t,\Rmat)
	:= 
	(\hmb_{\sss{ij_{\sss{1}}}}(\Rmat), \cdots, \hmb_{\sss{ij_{\sss{| \Nset_{\sss{i}}(\sigma(t))|}}}}(\Rmat)),
	\label{eq:RotationMatrixAgentMeasurement}
\end{align}
where $\{ j_{\sss{1}}, \cdots, j_{\sss{| \Nset_{\sss{i}}(\sigma(t))|}} \} \equiv \Nset_{\sss{i}}(\sigma(t))$ and where
\begin{align}
	&\hmb_{\sss{ij}}: 
	\SO[3]^{\sss{N}} \ni \Rmat:=(\Rmati{1},\cdots,\Rmati{N}) \mapsto \hmb_{\sss{ij}}(\Rmat) \in \SO[3]
	\\
	&
	\hmb_{\sss{ij}}(\Rmat) 
	:= 
	\Rmati{i}\tp\Rmati{j},
	\label{eq:RotationMatrixAgentMeasurement2}
\end{align}
for each $j \in \Nset_{\sss{i}}(\sigma(t))$.
Thus, at each time instant $t \in \Rpositive$, agent $i \in \Nset$ measures the $|\Nset_{\sss{i}}(\sigma(t))|$ rotation matrices corresponding to its neighbors orientation with respect to its own orientation.
We emphasize that the measurement in~\eqref{eq:RotationMatrixAgentMeasurement2} does not require an agent to be aware of its own rotation matrix or its neighbors rotation matrices (recall that these are specified in an unknown inertial orientation frame);
rather it requires an agent to measure the projection of each of its neighbors three axes onto its own three axes.

\begin{prob}
	\label{prob:ProblemRotMatrixDynamics}
	For each agent $i \in \Nset$ and time instant $t \ge 0$, given the measurement function~\eqref{eq:RotationMatrixAgentMeasurement}, design time-varying decentralized feedback laws $\Omi{i}^{\sss{h}}(t,\cdot): \SO[3]^{\sss{|\Nset_{\sss{i}}(\sigma(t))|}} \mapsto \Rn[3] $, such that asymptotic synchronization of $\Rmat := (\Rmati{1},\cdots,\Rmati{N}) : \Rn[]_{\sss{\ge 0}} \mapsto \SO[3]^{\sss{N}}$ is attained, where $\dot{\Rmat}_{\sss{i}}(t) = \fmb_{\sss{\Rmat}}(\Rmat_{\sss{i}}(t), \Omi{i}^{\sss{h}}(t,\hmb_{\sss{i}}(t,\Rmat(t))))$ for every $i \in \Nset$.
\end{prob}
Problem~\ref{prob:ProblemRotMatrixDynamics} may be restated as finding a control law for each agent that depends exclusively on the measurement function as defined in~\eqref{eq:RotationMatrixAgentMeasurement}, and which encodes the partial state information available to each agent at a given time instant.
\begin{defn}
	\label{defn:ThetaRotMatrix}
	The angular displacement between two rotation matrices is given by $\theta: \SO[3] \times \SO[3] \ni (\Rmati{1},\Rmati{2}) \mapsto \theta(\Rmati{1},\Rmati{2}) := \arccos\left(\frac{\text{tr}(\Rmati{1}\tp\Rmati{2}) -1}{2} \right) \in  [0, \pi]$. 
\end{defn}
For each agent $i \in \Nset$ and each time instant $t \in \Rpositive$, consider the control law
$
	\Omi{i}^{\sss{h}}(t,\cdot): 
	\SO[3]^{\sss{| \Nset_{\sss{i}}(\sigma(t))|}} 
	\ni \hmb_{\sss{i}} := (\hmb_{\sss{ij_{\sss{1}}}}, \cdots, \hmb_{\sss{ij_{\sss{| \Nset_{\sss{i}}(\sigma(t))|}}}})
	\mapsto \Omi{i}^{\sss{h}}(t,\hmb_{\sss{i}}) \in \Rn[3]
$ 
defined as
\begin{align}
	\hspace{-0.4cm}
	\Omi{i}^{\sss{h}}(t,\hmb_{\sss{i}})
	:=
	\sum\limits_{j \in \Nset_{\sss{i}}(\sigma(t))} 
	w_{\sss{ij}}(\theta(\Idmat,\hmb_{\sss{ij}}))
	\invsk{
		\frac{
			\hmb_{\sss{ij}}
			-
			\hmb_{\sss{ij}}\tp
		}{2}
	},
	\label{eq:ControlLawSO3}
\end{align}
where $w_{\sss{ij}}: [0, \pi] \mapsto \Rn[{}]_{\sss{\ge 0}} $ is continuous and satisfies
\begin{align}
	w_{\sss{ij}}(\theta) > 0  \,\, \forall \theta \in (0, \pi].
	\label{eq:RotationMatrixGFunction}
\end{align}
Notice that $w_{\sss{ij}}$ corresponds to a weight on the feedback law~\eqref{eq:ControlLawSO3} that agent $i$ assigns to the displacement between itself and its neighbor $j$.
Denote $\Omi{i}^{\sss{cl}}$ as the composition of the output feedback control law~\eqref{eq:ControlLawSO3}  with the output function~\eqref{eq:RotationMatrixAgentMeasurement}, i.e., $\Omi{i}^{\sss{cl}}: \Rn[]_{\sss{+}} \times \SO[3]^{\sss{N}} \ni (t, \Rmat) \mapsto \Omi{i}^{\sss{cl}}(t, \Rmat) := \Omi{i}^{\sss{h}}(t, \hmb_{\sss{i}}(t,\Rmat)) \in \Rn[3]$.
It follows that
\begin{align}
	\hspace{-0.2cm}
	\Scale[0.95]{
		\Omi{i}^{\sss{cl}}(t, \Rmat)
		=
		\sum\limits_{j \in \Nset_{\sss{i}}(\sigma(t))} 
		w_{\sss{ij}}(\theta(\Idmat,\Rmati{i}\tp\Rmati{j}))
		\invsk{
			\frac{
				\Rmati{i}\tp \Rmati{j}
				-
				\Rmati{j}\tp \Rmati{i}
			}{2}
		},	
	}
	\label{eq:ControlLawSO3StateFeedback}
\end{align}
where we have made us of the notation $\Rmat = (\Rmat_{\sss{1}},\cdots,\Rmat_{\sss{N}})$. 
We emphasize that if Assumption~\ref{assump:SwitchesAgents} is satisfied, then~\eqref{eq:ControlLawSO3} (and~\eqref{eq:ControlLawSO3StateFeedback}) does not have infinite many discontinuities in any time interval of finite length, which in turn implies that it can be implemented in a practical scenario.  

Recall that we wish to analyze different problems under a common framework where agents are unit vectors. 
Thus, in order to cast complete synchronization in $\SO[3]$ in the form~\eqref{eq:GeneralForm}, we perform a change of coordinates based on unit quaternions.
This change of variables serves only the purpose of analysis, while the implemented control law is still that in~\eqref{eq:ControlLawSO3}.

For that purpose, and for convenience, denote
\begin{align}
	\tilde{\mathbb{SO}}(3)
	:=
	\left\{
		\Rmat \in \SO[3]:
		\theta(\Idmat,\Rmat) < \pi
	\right\}
\end{align}
which is an open subset of $\SO[3]$, and where $\theta(\Idmat,\Rmat) < \pi \Leftrightarrow 1 + \text{tr}(\Rmat) > 0$ (see Definition~\ref{defn:ThetaRotMatrix}). 
Consider then the map 
$
	\mrtoq:
	\tilde{\mathbb{SO}}(3) \ni \Rmat
	\mapsto
	\mrtoq(\Rmat)
	\in\Sn{3}
$
defined as
\begin{align}
	\mrtoq(\Rmat):=
	\frac{1}{2}
	\left(
		\sqrt{1+\text{tr}(\Rmat)},
		\frac{\invsk{\Rmat - \Rmat\tp}}{\sqrt{1+\text{tr}(\Rmat)}}
	\right),
	\label{eq:map rotation to quaternion}
\end{align}
with $\mrtoq$ smooth on $\tilde{\mathbb{SO}}(3)$. 
Intuitively, $\mrtoq$ is a map that transforms a rotation matrix into a unit vector in $\Rn[4]$, also named unit quaternion, whose first component is positive. 
The idea followed later is: \emph{(i)} given a rotation matrix $\Rmat \in \tilde{\mathbb{SO}}(3)$, to consider the quaternion $\qmb = \mrtoq(\Rmat)$; and, \emph{(ii)} given the closed loop dynamics of the rotation matrix, to compute the closed loop dynamics of the quaternion, which are in the common form~\eqref{eq:GeneralForm2}. 
For that purpose, consider also the map 
$
	\mqtor:
	\Sn{3} \ni \qmb := (\bar{q},\hat{\qmb})
	\mapsto
	\mqtor(\qmb)
	\in
	\SO[3]
$
defined as
\begin{align}
	\mqtor(\qmb):=
	\Idmat
	+
	2 \bar{q} \sk{\hat{\qmb}}
	+
	2 \sk{\hat{\qmb}} \sk{\hat{\qmb}},
\end{align}
where $\mqtor$ restricted to the image of $\mrtoq$ yields the inverse map of $\mrtoq$, i.e., $\mqtor \circ \mrtoq = \text{id}_{\sss{\tilde{\mathbb{SO}}(3)}}$.

Let us now compute the kinematics of a quaternion computed from the mapping~\eqref{eq:map rotation to quaternion}. 
Given $\Rmat \in \tilde{\mathbb{SO}}(3)$, consider $\qmb = \mrtoq(\Rmat)$.
Given the kinematics~\eqref{eq:RotationMatrixKinematics}, one can compute the kinematics of the unit quaternion, as done in Appendix~\ref{app:QuaternionsAppendix}. 
It follows that
\begin{align}
	\dot{\qmb}
	\equiv
	\fmb_{\sss{q}}:
	\tilde{\mathbb{S}}^{\sss{3}} \times \Rn[3]
	\ni
	(\qmb,\bm{\omega})
	\mapsto
	\fmb_{\sss{q}}(\qmb,\bm{\omega})
	\in 
	T_{\sss{\qmb}} \tilde{\mathbb{S}}^{\sss{3}} 
\end{align}
is given by
\begin{align}
	\dot{\qmb}(\qmb,\bm{\omega})
	\equiv
	&
	\fmb_{\sss{q}}(\qmb,\bm{\omega})
	=
	\frac{1}{2} Q(\qmb) 
	[\zvec_{\sss{3\times 1}} \, \Idmat_{\sss{3}}]\tp \bm{\omega},
	\label{eq:quaternion kinematics}
\end{align}
where, given any $\qmb \in \Sn{3}$, $Q(\qmb) \in T_{\qmb} \Sn{3} \subset \Rn[4\times4]$ is the linear map that satisfies
\begin{align}
	\hspace{-0.5cm}
	Q(\qmb) \vmb  
	:= 
	&
	\innerproduct{\qmb}{\emb_{\sss{1}}} \vmb - \innerproduct{\qmb}{\vmb} \emb_{\sss{1}}  + \innerproduct{\emb_{\sss{1}}}{\vmb} \qmb +
	\\
	&
 	 + \text{diag}(0,\sk{\hat{\qmb}}) \vmb,
	\label{eq:q matrix quaternion}
\end{align}
for any $\vmb \in \Rn[4]$ (it is easy to verify that $\innerproduct{\qmb}{Q(\qmb) \vmb} = 0$).

The control law~\eqref{eq:ControlLawSO3StateFeedback} may also be rewritten in terms of unit quaternions, which motivates the definition of
$
	 \Omi{i}^{\sss{q}}(t,\cdot): 
	 (\tilde{\mathbb{S}}^{\sss{3}})^{\sss{N}}  
	 \ni 
	 \qmb := (\qmb_{\sss{1}},\cdots,\qmb_{\sss{1}})
	 \mapsto 
	 \Omi{i}^{\sss{q}}(t,\qmb) 
	 \in 
	 \Rn[3]
$
as
\begin{align}
	\hspace{-0.4cm}
	& \Omi{i}^{\sss{q}}(t,\qmb) 
	:=
	\Omi{i}^{\sss{cl}}(t,\Rmat)|_{\Rmat = \mqtor{}^{\sss{N}}(\qmb)}
	\label{eq:ControlLawSO3v2}
	\\
	\hspace{-0.4cm}
	& =
	\sum_{j \in \Nset_{\sss{i}}(\sigma(t))} 
	2 
	\innerproduct{\qmbi{i}}{\qmbi{j}}
	w_{\sss{ij}}(\theta(\qmbi{i},\qmbi{j}))
	[
		\zvec_{\sss{3 \times 1}} \,
		\Idmat
	]
	Q(\qmbi{i})\tp	
	\qmbi{j},
	\label{eq:control law so3 quaternion}
\end{align}
where we have made used of~\eqref{eq:theta quaternion product} and~\eqref{eq:invsk quaternion product} in Appendix~\ref{app:QuaternionsAppendix}, and where, for brevity, we denoted $\theta(\qmbi{i},\qmbi{j}) := \arccos(2 \innerproduct{\qmbi{i}}{\qmbi{j}}^{\sss{2}}  - 1)$.
Denote $\tilde{w}_{\sss{ij}}(\qmbi{i},\qmbi{j}) := \innerproduct{\qmbi{i}}{\qmbi{j}} w_{\sss{ij}}(\arccos(2 \innerproduct{\qmbi{i}}{\qmbi{j}}^{\sss{2}}  - 1))$, and notice that this satisfies~\eqref{eq:gTildeCase} for $\alpha = \frac{\pi}{4}$ and for any $\bar{\numb} \in \mathcal{S}^{\sss{3}}$ (see also~\eqref{eq:RotationMatrixGFunction}).

Then, the unit quaternion kinematics~\eqref{eq:quaternion kinematics}, when composed with the control law~\eqref{eq:control law so3 quaternion}, is in the same form as~\eqref{eq:GeneralForm2}, i.e.
\begin{align}
	&
	\fmb_{\sss{q}}(\qmb_{\sss{i}},\Omi{i}^{\sss{q}}(t,\qmb))|_{\qmb=(\qmb_{\sss{1}},\cdots,\qmb_{\sss{N}})}
	=
	\\
	{\sss{\eqref{eq:quaternion kinematics}}}
	&
	=
	\frac{1}{2} Q(\qmb_{\sss{i}}) 
	[\zvec_{\sss{3\times 1}} \, \Idmat_{\sss{3}}]\tp 
	\Omi{i}^{\sss{q}}(t,\qmb)|_{\qmb=(\qmb_{\sss{1}},\cdots,\qmb_{\sss{N}})}
	\\
	{\sss{\eqref{eq:control law so3 quaternion}}}
	&
	=
	\sum\nolimits_{\sss{j \in \Nset_{\sss{i}}(\sigma(t))}}
	\tilde{w}_{\sss{ij}}(\qmbi{i},\qmbi{j})
	Q(\qmbi{i})
	\text{diag}(0,\Idmat_{\sss{3}})
	Q(\qmbi{i})\tp
	\qmbi{j}
	\\ 
	{\sss{\eqref{eq:q matrix quaternion}}}
	&
	=
	\sum\nolimits_{\sss{j \in \Nset_{\sss{i}}(\sigma(t))}}
	\tilde{w}_{\sss{ij}}(\qmbi{i},\qmbi{j})
	\OP{\qmbi{i}} \qmbi{j}
	\\
	{\sss{\eqref{eq:GeneralForm2}}}
	&
	=:
	\tilde{\fmb}_{\sss{i,\sigma(t)}}(\qmb).
\end{align}
We have then casted this problem in the form~\eqref{eq:GeneralForm} with $\bm{\nu} \equiv \qmb \in (\Sn{3})^{\sss{N}}$.

\begin{rem}
	Consider $ (\qmb_{\sss{1}},\cdots,\qmb_{\sss{N}}) \in \ConeOpen(\frac{\pi}{4},\bar{\qmb})^{\sss{N}}$ for some $\bar{\qmb} \in \Sn{3}$, as required by $\tilde{w}_{\sss{ij}}$ to be positive.
	Without loss of generality, one may assume that $\bar{\qmb} = \emb_{\sss{1}}$, where $\mqtor(\bar{\qmb}) = \Idmat$.
	It then follows that $\theta(\mqtor(\bar{\qmb}),\mqtor(\qmb_{\sss{i}})) = \theta(\Idmat,\mqtor(\qmb_{\sss{i}})) = \arccos(2 (\innerproduct{\bar{\qmb}}{\qmbi{i}})^2  - 1)  \le  \arccos\left(2 \cos^{\sss{2}}\left(\frac{\pi}{4}\right)  - 1\right) = \frac{\pi}{2}$ for all $i \in \Nset$.
	As such, i.e. all rotation matrices are $\frac{\pi}{2}$ \emph{close} to the identity matrix, and therefore if $ (\qmb_{\sss{1}},\cdots,\qmb_{\sss{N}}) \in \ConeOpen(\frac{\pi}{4},\bar{\qmb})^{\sss{N}}$ then $ (\Rmat_{\sss{1}},\cdots,\Rmat_{\sss{N}}) \in \tilde{\mathbb{SO}}(3)^{\sss{N}}$.
\end{rem}

	\subsection{Incomplete synchronization in $\SO[3]$ casted as synchronization in $\Sn{2}$}
	\label{sec:SynchronizationInS2}
	In this section, we consider again a group of $N$ agents operating in $\SO[3]$, but with a different synchronization objective.
%
As in Section~\ref{sec:SynchronizationInSO3}, for each $i \in N$, $\Rmati{i}$ represents the orientation frame of agent $i$. 
Additionally, for each agent $i$ there is a constant body direction $\bar{\nmb}_{\sss{i}} \in \Sn{2}$, known by the agent and its neighbors, which is required to synchronize with all the other agents' body directions.
The goal of incomplete attitude synchronization in $\SO[3]$ is that all agents share the same orientation along the chosen body directions;
i.e., given  $(\Rmati{1},\cdots,\Rmati{N}) \in \SO[3]^{\sss{N}}$ and $(\bar{\nmb}_{\sss{1}},\cdots,\bar{\nmb}_{\sss{N}}) \in (\mathcal{S}^{\sss{2}})^{\sss{N}}$, \emph{incomplete} synchronization exists when $\Rmati{1}\bar{\nmb}_{\sss{1}}  = \cdots = \Rmati{N} \bar{\nmb}_{\sss{N}}$, as illustrated in Fig.~\ref{fig:2RigidBodiesUnitVectors} for $N =2 $.
We note that the requirement for incomplete synchronization is independent from that of complete synchronization: i.e., complete synchronization does not imply incomplete synchronization (consider the case where $\bar{\nmb}_{\sss{1}}\ne\bar{\nmb}_{\sss{2}}$), and vice-versa.
\begin{figure}
	\centering
	\begin{subfigure}[b]{\linewidth}
		\centering
		\includegraphics[clip=true,trim=1cm 1.5cm 0.1cm 1.5cm,width=0.7\textwidth]
		{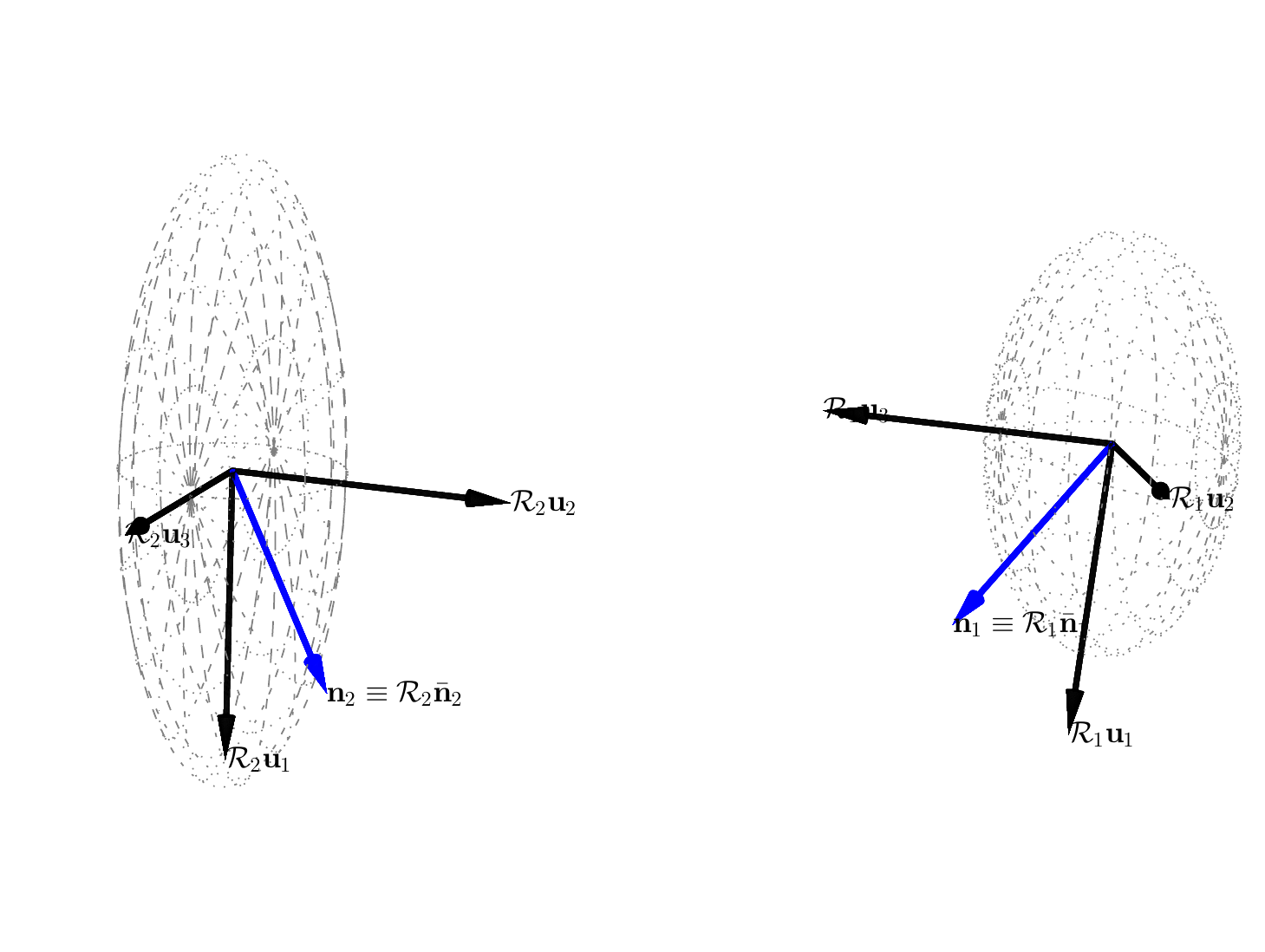}
		\caption{Two agents not synchronized, i.e., $\nmb_{\sss{1}} \ne \nmb_{\sss{2}}$ with $\bar{\nmb}_{\sss{1}} = \bar{\nmb}_{\sss{2}} = 3^{\sss{-\frac{1}{2}}} {[1\, 1\, 1]}\tp$.}
		\label{subfig:Ellipsoids_Not_Synchronized_UnitVectors}
	\end{subfigure}
	\begin{subfigure}[b]{\linewidth}
		\centering
		\includegraphics[clip=true,trim=1cm 1.5cm 0.1cm 1.5cm,width=0.7\textwidth]
		{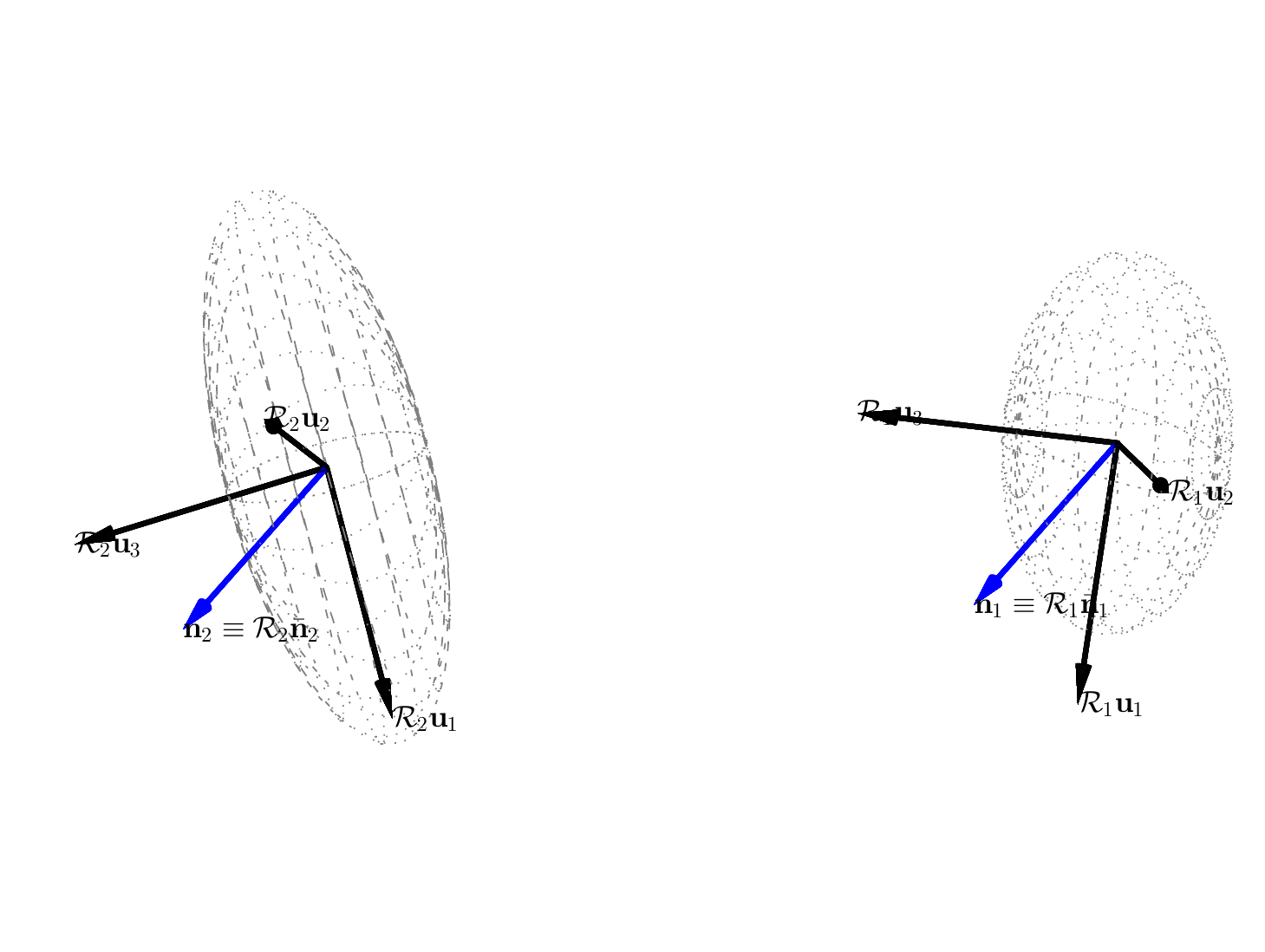}
		\caption{Two agents  synchronized, i.e, $\nmb_{\sss{1}} = \nmb_{\sss{2}}$ with $\bar{\nmb}_{\sss{1}} = \bar{\nmb}_{\sss{2}} = 3^{\sss{-\frac{1}{2}}} {[1\, 1\, 1]}\tp$.}
		\label{subfig:Ellipsoids_Synchronized_UnitVectors}
	\end{subfigure}
	\caption{In incomplete synchronization, all agents $i = \{1,\cdots,N\}$, align the unit vectors $\nmbi{i} \triangleq \Rmati{i}\bar{\nmb}_{\sss{i}}$, where $\bar{\nmb}_{\sss{i}}$ is fixed in rigid body $i$ ($\umb_{\sss{1}}$,$\umb_{\sss{2}}$ and $\umb_{\sss{3}}$ stand for the canonical basis vectors in $\Rn[3]$).}
	\label{fig:2RigidBodiesUnitVectors}
\end{figure}
%
%

Similarly to Section~\ref{sec:SynchronizationInSO3}, for each $i \in \Nset$, $\bm{\omega}_{\sss{i}}:  \Rn[]_{\sss{\ge 0}} \mapsto \Rn[3]$ denotes the body-framed angular velocity of agent $i$, which can be actuated.
Again, each rotation matrix $\Rmat_{\sss{i}}: \Rn[]_{\sss{\ge 0}} \ni t  \mapsto \Rmat_{\sss{i}}(t) \in \SO[3]$ evolves according to $\dot{\Rmat}_{\sss{i}}(t) = \fmb_{\sss{\Rmat}}(\Rmat_{\sss{i}}(t),\bm{\omega}_{\sss{i}}(t))$ with $\fmb_{\sss{\Rmat}}$ as defined in~\eqref{eq:RotationMatrixKinematics}.
If, additionally, we consider some constant $\bar{\nmb}_{\sss{i}} \in \Sn{2}$, then $\nmb_{\sss{i}} := \Rmat_{\sss{i}} \bar{\nmb}_{\sss{i}}: \Rn[]_{\sss{\ge 0}} \ni t  \mapsto \nmb_{\sss{i}}(t) \in \Sn{2}$, evolves according to $\dot{\nmb}_{\sss{i}}(t) = \fmb_{\sss{n}}(\Rmat_{\sss{i}}(t),\bm{\omega}_{\sss{i}}(t),\bar{\nmb}_{\sss{i}})$ where
\begin{align}
	&
	\Scale[0.95]{
		\fmb_{\sss{n}}: \SO[3] \times \Rn[3] \times \Sn{2} \ni (\Rmat,\bm{\omega},\bar{\nmb}) \mapsto \fmb_{\sss{n}}(\Rmat,\bm{\omega},\bar{\nmb}) \in 
		T_{\sss{\Rmat \bar{\nmb}}} \Sn{2} 
	}
	\\
	&
	\Scale[0.95]{
		\fmb_{\sss{n}}(\Rmat,\bm{\omega},\bar{\nmb})
		:=
		\fmb_{\sss{R}}(\Rmat,\bm{\omega}) \bar{\nmb}
		=
		\Rmat \sk{\bm{\omega}} \bar{\nmb}
		=
		-\sk{\Rmat \bar{\nmb}} \Rmat \bm{\omega} .
	}
	\label{eq:UnitVectorKinematics}
\end{align}

In what follows, we make use of the notation below, 
\begin{align}
	&
	\Rmat = (\Rmat_{\sss{1}},\cdots,\Rmat_{\sss{N}}),
	\label{eq:R complete}
	\\
	&
	(\nmbi{1},\cdots,\nmbi{N}) = (\Rmati{1} \bar{\nmb}_{\sss{1}},\cdots,\Rmati{N} \bar{\nmb}_{\sss{N}}).
	\label{eq:n complete}
\end{align}

If, at a time instant $t \in \Rpositive$, agent $i \in \Nset$ is aware of the relative attitude of agent's $j$ unit vector to be synchronized, then $j \in \Nset_{\sss{i}}(\sigma(t))$.
This motivates the definition of the measurement function $\hmb_{\sss{i}}(t,\cdot)$ for each time instant $t \in \Rpositive$, namely 
\begin{align}
	&
	\hmb_{\sss{i}}(t,\cdot): \SO[3]^{\sss{N}} \ni \Rmat \mapsto \hmb_{\sss{i}}(t,\Rmat) \in (\Sn{2})^{\sss{|\Nset_{\sss{i}}(\sigma(t))|}}
	\\
	&
	\hmb_{\sss{i}}(t,\Rmat) 
	:= 
	(
		\hmb_{\sss{ij_{\sss{1}}}}(\Rmat), 
		\cdots, 
		\hmb_{\sss{ij_{\sss{| \Nset_{\sss{i}}(\sigma(t))|}}}}(\Rmat)
	),
	\label{eq:UnitVectorAgentMeasurement}
\end{align}
where $\{ j_{\sss{1}}, \cdots, j_{\sss{| \Nset_{\sss{i}}(\sigma(t))|}} \} \equiv \Nset_{\sss{i}}(\sigma(t))$ and where
\begin{align}
	&\hmb_{\sss{ij}}: 
	\SO[3]^{\sss{N}} \ni \Rmat \mapsto \hmb_{\sss{ij}}(\Rmat) \in \Sn{2}
	\\
	&
	\hmb_{\sss{ij}}(\Rmat) 
	:= 
	\Rmati{i}\tp \Rmati{j} \bar{\nmb}_{\sss{j}}
	=
	\Rmati{i}\tp\nmbi{j},
	\label{eq:UnitVectorAgentMeasurement2}
\end{align}
for each $j \in \Nset_{\sss{i}}(\sigma(t))$, and where we have used of the notation in~\eqref{eq:R complete} and~\eqref{eq:n complete}.
Thus, at each time instant $t \in \Rpositive$, agent $i \in \Nset$ measures the $| \Nset_{\sss{i}}(\sigma(t))|$ unit vectors corresponding to the projection of a neighbor's unit vector onto agent's $i$ orientation frame.
\begin{prob}
	\label{prob:ProblemUnitVectorDynamics}
	For each agent $i \in \Nset$ and time instant $t \ge 0$, given the measurement function~\eqref{eq:UnitVectorAgentMeasurement}, design time-varying decentralized feedback laws $\Omi{i}^{\sss{h}}(t,\cdot): (\Sn{2})^{\sss{|\Nset_{\sss{i}}(\sigma(t))|}} \mapsto \Rn[3] $, such that asymptotic synchronization of $(\Rmati{1} \bar{\nmb}_{\sss{1}},\cdots,\Rmati{N} \bar{\nmb}_{\sss{N}}) : \Rn[]_{\sss{\ge 0}} \mapsto (\Sn{2})^{\sss{N}}$ is accomplished, where $\dot{\Rmat}_{\sss{i}}(t) = \fmb_{\sss{\Rmat}}(\Rmat_{\sss{i}}(t), \Omi{i}^{\sss{h}}(t,\hmb_{\sss{i}}(t,\Rmat(t))))$ for every $i \in \Nset$.
\end{prob}
Problem~\ref{prob:ProblemUnitVectorDynamics} may be restated as finding a control law for each agent that depends exclusively on the measurement function as defined in~\eqref{eq:UnitVectorAgentMeasurement}, and which encodes the partial state information available to each agent at a given time instant.
\begin{defn}
	\label{defn:ThetaUnitVector}
	The angular displacement between two unit vectors is given by $\theta: \Sn{2} \times \Sn{2} \ni (\nmbi{1},\nmbi{2}) \mapsto \theta(\nmbi{1},\nmbi{2}) := \arccos(\innerproduct{\nmbi{2}}{\nmbi{2}}) \in  [0, \pi]$.
\end{defn}
For each $t \in \Rpositive$ and agent $i \in \Nset$, consider the control law $\Omi{i}^{\sss{h}}(t,\cdot):  (\Sn{2})^{\sss{| \Nset_{\sss{i}}(\sigma(t))|}} \ni \hmb_{\sss{i}} := (\hmb_{\sss{ij_{\sss{1}}}}, \cdots, \hmb_{\sss{ij_{\sss{| \Nset_{\sss{i}}(\sigma(t))|}}}}) \mapsto \Omi{i}^{\sss{h}}(t,\hmb) \in \Rn[3]$ defined as
\begin{align}
	\Omi{i}^{\sss{h}}(t, \hmb_{\sss{i}})
	:=
	\sum\limits_{j \in \Nset_{\sss{i}}(\sigma(t))} 
	w_{\sss{ij}}(\theta(\bar{\nmb}_{\sss{i}},\hmb_{\sss{ij}}))
	\sk{\bar{\nmb}_{\sss{i}}}
	\hmb_{\sss{ij}},
	\label{eq:DistributedControlLawUnitVector}
\end{align}
where $w_{\sss{ij}}: [0, \pi] \mapsto \Rn[{}]_{\sss{\ge 0}} $ is a continuous function satisfying
\begin{align}
	w_{\sss{ij}}(\theta) > 0, \forall \theta \in (0,\pi],
	\label{eq:UnitVectorPositivenessGFunction}
\end{align}
and corresponding to a weight on the feedback law~\eqref{eq:DistributedControlLawUnitVector} agent $i$ assigns to the angular displacement between itself and its neighbor $j$.
Denote $\Omi{i}^{\sss{cl}}$ as the composition of the output feedback control law~\eqref{eq:DistributedControlLawUnitVector} with the output function~\eqref{eq:UnitVectorAgentMeasurement}, i.e., $\Omi{i}^{\sss{cl}} : \Rn[]_{\sss{+}} \times \SO[3]^{\sss{N}} \ni (t, \Rmat) \mapsto \Omi{i}^{\sss{cl}}(t,\Rmat) := \Omi{i}^{\sss{h}}(t, \hmb_{\sss{i}}(t,\Rmat)) \in \Rn[3]$. 
It follows that
\begin{align}
	&
	\Omi{i}^{\sss{cl}}(t, \Rmat)
	:= \Omi{i}^{\sss{h}}(t, \hmb_{\sss{i}}(t,\Rmat))
	\\
	&
	=
	\sum\nolimits_{j \in \Nset_{\sss{i}}(\sigma(t))} 
	w_{\sss{ij}}(\theta(\bar{\nmb}_{\sss{i}},\Rmati{i}\tp\Rmati{j}\bar{\nmb}_{\sss{j}}))
	\sk{\bar{\nmb}_{\sss{i}}}\Rmati{i}\tp\Rmati{j}\bar{\nmb}_{\sss{j}}
	\\
	&
	=
	\Rmati{i}\tp
	\sum\nolimits_{j \in \Nset_{\sss{i}}(\sigma(t))} 
	w_{\sss{ij}}(\theta(\nmbi{i},\nmbi{j}))
	\sk{\nmbi{i}}\nmbi{j},	
	\label{eq:DistributedControlLawUnitVectorState}
\end{align}
where we have used of~\eqref{eq:R complete} and~\eqref{eq:n complete};
and of the fact that $\theta(\nmb_{\sss{i}},\nmbi{j}) = \theta(\Rmati{i}\tp \nmb_{\sss{i}},\Rmati{i}\tp\nmbi{j})$ for any $\Rmati{i} \in \SO[3]$ and any $\nmb_{\sss{i}},\nmbi{j} \in \Sn{2}$.
We emphasize that the control law~\eqref{eq:DistributedControlLawUnitVectorState} is based on the output feedback control law, and thus depends only on the relative orientation measurements (see~\eqref{eq:UnitVectorAgentMeasurement} and~\eqref{eq:UnitVectorAgentMeasurement2}); 
moreover,~\eqref{eq:DistributedControlLawUnitVectorState} is orthogonal to $\bar{\nmb}_{\sss{i}}$, which implies that full angular velocity control is not necessary, i.e. we only need to control the angular velocity along the two directions orthogonal to $\bar{\nmb}_{\sss{i}}$ ($\Omi{i}^{\sss{h}}: \Rn[]_{\sss{+}} \times (\Sn{2})^{\sss{| \Nset_{\sss{i}}(\sigma(t))|}} \mapsto T_{\sss{\bar{\nmb}_{\sss{i}}}} \Sn{2}  \subset \Rn[3]$).

Denote $\tilde{w}_{\sss{ij}}(\nmbi{i},\nmbi{j}) = w_{\sss{ij}}(\theta(\nmb_{\sss{i}},\nmbi{j}))$, which satisfies~\eqref{eq:gTildeCase} for any $\alpha \in [0,\pi]$ and $\bar{\numb} \in \Sn{2}$, due to~\eqref{eq:UnitVectorPositivenessGFunction} and Definition~\ref{defn:ThetaUnitVector}.
With the above in mind, the kinematics~\eqref{eq:UnitVectorKinematics}, when composed with the proposed law~\eqref{eq:DistributedControlLawUnitVectorState} yield~\eqref{eq:GeneralForm2}, i.e.,
\begin{align}
	&
	\fmb_{\sss{n}}(\Rmat_{\sss{i}},\bm{\omega}_{\sss{i}}^{\sss{cl}}(t,(\Rmat_{\sss{1}},\cdots,\Rmat_{\sss{N}})),\bar{\nmb}_{\sss{i}})
	|_{\Rmat_{\sss{i}} \bar{\nmb}_{\sss{i}} = \nmb_{\sss{i}} \forall i \in \Nset}
	\\
	{\sss{\eqref{eq:UnitVectorKinematics}}}
	&
	=
	-
	\sk{\nmbi{i}}
	\Rmati{i}
	\bm{\omega}_{\sss{i}}^{\sss{cl}}(t,(\Rmat_{\sss{1}},\cdots,\Rmat_{\sss{N}}))
	|_{\Rmat_{\sss{i}} \bar{\nmb}_{\sss{i}} = \nmb_{\sss{i}} \forall i \in \Nset}
	\\
	{\sss{\eqref{eq:DistributedControlLawUnitVectorState}}}
	&
	=
	-
	\sum\nolimits_{j \in \Nset_{\sss{i}}(\sigma(t))} 
	\tilde{w}_{\sss{ij}}(\nmbi{i},\nmbi{j})
	\sk{\nmbi{i}}\sk{\nmbi{i}}\nmbi{j}
	\\
	&
	=
	\sum\nolimits_{j \in \Nset_{\sss{i}}(\sigma(t))} 
	\tilde{w}_{\sss{ij}}(\nmbi{i},\nmbi{j})
	\OP{\nmbi{i}}\nmbi{j}
	\\
	{\sss{\eqref{eq:GeneralForm2}}}
	&
	=:
	\tilde{\fmb}_{\sss{i,\sigma(t)}}(\nmb)|_{\nmb = (\nmb_{\sss{1}},\cdots,\nmb_{\sss{N}}) \in (\Sn{2})^{\sss{N}}}
	.
\end{align}
We have thus casted this problem in the form~\eqref{eq:GeneralForm}-\eqref{eq:GeneralForm2} with $\bm{\nu} \equiv \nmb \in (\Sn{2})^{\sss{N}}$.

As stated in Section~\ref{sec:Synchronization}, consensus in $\Rn[n]$ can also be casted as a synchronization problem in $\mathcal{S}^{\sss{n}}$, for any $n \in \mathbb{N}$. This case is presented in Appendix~\ref{app:ConsensusRn}.			

	\section{Analysis}
	\label{sec:Analysis}
	In this section, we analyze the solutions of~\eqref{eq:GeneralForm}-\eqref{eq:GeneralForm2}, and show that given a wide set of initial conditions, asymptotic synchronization is guaranteed.
Specifically, asymptotic synchronization is guaranteed if all unit vectors are initially contained in an open $\alpha^{\sss{\star}}$-cone, i.e. if $\exists \bar{\bm{\nu}} \in \Sn{n} : \bm{\nu}(0) \in \ConeOpen(\alpha^{\sss{\star}},\bar{\bm{\nu}})^{\sss{N}}$, where $\alpha^{\sss{\star}} = \frac{\pi}{2}$ for synchronization in $\Sn{2}$ and consensus in $\Rn[n]$, and $\alpha^{\sss{\star}} = \frac{\pi}{4}$ for synchronization in $\SO[3]$.
%
%
%
\begin{rem}
	If $\bm{\nu}(0) \in \ConeOpen(\alpha,\bar{\bm{\nu}})^{\sss{N}}$ for some $\alpha \in [0,\alpha^{\sss{\star}})$ and some $\bar{\bm{\nu}} \in \Sn{n}$, then Proposition~\ref{prop:ConeGeneration} in Appendix~\ref{app:AuxiliaryResults} guarantees that there exist $n+1$ linearly independent unit vectors $\{\bar{\numb}_{\sss{k}} \in \Sn{n}\}_{\sss{k \in \{1,\cdots,n+1\}}}$ such $\bm{\nu}(0) \in \ConeOpen(\alpha^{\sss{\star}},\bar{\bm{\nu}}_{\sss{k}})^{\sss{N}} \forall k \in \{1,\cdots,n+1\}$. 
	Thus, if $\bm{\nu}(0)$ is contained in an open $\alpha$ cone, then there exist other bigger cones that contain $\bm{\nu}(0)$; as such, the choice of $\bar{\bm{\nu}}$ is not unique.
\end{rem}
Next, we introduce a coordinate transformation that we exploit in order to cast the dynamics~\eqref{eq:GeneralForm2} into a form that satisfies the conditions of Theorem~\ref{thm:MainTheorem}. 
In particular, given some $\bar{\bm{\nu}} \in \Sn{2}$ we consider the projection of the cone $\ConeOpen(\frac{\pi}{2},\bar{\bm{\nu}})$ to the plane in $\Rn[n+1]$ orthogonal to $\bar{\bm{\nu}}$ and containing zero, and then map this plane isometrically to $\Rn[n]$.
%
%
\begin{defn}
	\label{defn:NormalError}
	Let $\bar{\bm{\nu}} \in \Sn{n}$ and $Q_{\sss{\bar{\bm{\nu}}}} \in \Rn[(n+1)\times n]$, such that $\bar{\bm{\nu}}$ and the columns of $Q_{\sss{\bar{\bm{\nu}}}}$ form an orthonormal basis of $\Rn[n+1]$,
	%
	and consider the diffeomorphism $\hmb_{\sss{\bar{\bm{\nu}}}}: \ConeOpen(\frac{\pi}{2},\bar{\bm{\nu}}) \mapsto  \mathcal{B}(1)$ (see Notation for the definition of $\mathcal{B}$), defined as 
	\begin{align}
		\hmb_{\sss{\bar{\bm{\nu}}}}(\numbi{i}) = Q_{\sss{\bar{\bm{\nu}}}}\tp \numbi{i}.
		\label{eq:Diffeomorphism}
	\end{align}
	Its derivative is given by $d\hmb_{\sss{\bar{\bm{\nu}}}} = Q_{\sss{\bar{\bm{\nu}}}}\tp$, and its inverse $\hmb_{\sss{\bar{\bm{\nu}}}}^{\sss{-1}}: \mathcal{B}(1) \mapsto \ConeOpen(\frac{\pi}{2},\bar{\bm{\nu}})$ is given by $\hmb_{\sss{\bar{\bm{\nu}}}}^{\sss{-1}}(\xmbi{i}) = \sqrt{1 - \|\xmbi{\sss{i}}\|^{\sss{2}}}\bar{\numb} +  Q_{\sss{\bar{\bm{\nu}}}} \xmbi{i}$.
\end{defn}
Figure~\ref{fig:Cone30WithError} illustrates the mapping $\hmb_{\sss{\bar{\bm{\nu}}}}$ as introduced in Definition~\ref{defn:NormalError}, for $n = 2$ and $N = 3$. 
\begin{figure}
	\centering
	\includegraphics[clip=true,trim=2.4cm 0cm 0cm 0cm,width=0.4\textwidth]
	{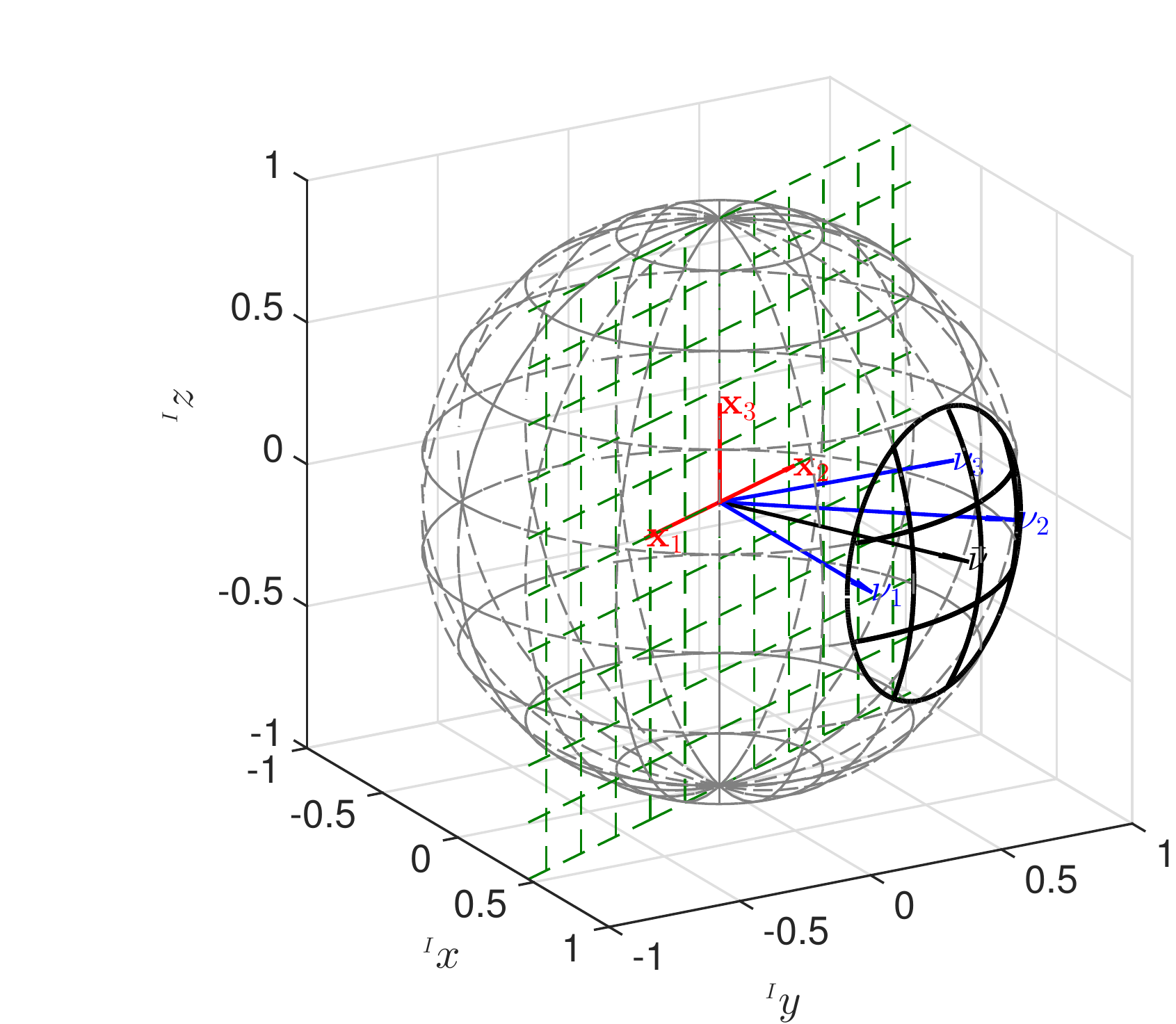}
	\caption{Illustration of~\eqref{eq:Diffeomorphism} as in Definition~\eqref{defn:NormalError}:  $\xmbi{i} = \hmb_{\sss{\bar{\bm{\nu}}}}(\numbi{i}) \overset{\sss{\eqref{eq:Diffeomorphism}}}{=} Q_{\sss{\bar{\bm{\nu}}}}\tp \numbi{i}$, and green lines are those spanned by the columns of $Q_{\sss{\bar{\bm{\nu}}}}$ (which form a plane).}
	\label{fig:Cone30WithError}
\end{figure}
\begin{prop}	
	\label{prop:ReconstructionNi}
	Consider $\bar{\bm{\nu}} \in \Sn{n}$ and $\numbi{i},\numbi{j} \in \ConeOpen(\frac{\pi}{2},\bar{\bm{\nu}})$.
	Then, $\innerproduct{\bar{\numb}}{\numb} = \sqrt{1 - \|\hmb_{\sss{\bar{\bm{\nu}}}}(\numb)\|^{\sss{2}}} > 0$ and the following implications hold: $\|\hmb_{\sss{\bar{\bm{\nu}}}}(\numbi{\sss{i}})\| > \|\hmb_{\sss{\bar{\bm{\nu}}}}(\numbi{\sss{j}})\|  \Leftrightarrow 0< \innerproduct{\bar{\numb}}{\numbi{i}} < \innerproduct{\bar{\numb}}{\numbi{j}}$, 
	and $\|\hmb_{\sss{\bar{\bm{\nu}}}}(\numbi{\sss{i}})\| \ge \|\hmb_{\sss{\bar{\bm{\nu}}}}(\numbi{\sss{j}})\| \Leftrightarrow 0< \innerproduct{\bar{\numb}}{\numbi{i}} \le \innerproduct{\bar{\numb}}{\numbi{j}}$.
\end{prop}

\begin{proof}
	Since $\hmb_{\sss{\bar{\bm{\nu}}}}(\numb) = Q_{\sss{\bar{\bm{\nu}}}} \numb$, then $\|\hmb_{\sss{\bar{\bm{\nu}}}}(\numb)\|^{\sss{2}} = \innerproduct{Q_{\sss{\bar{\bm{\nu}}}}\tp\numb}
	{Q_{\sss{\bar{\bm{\nu}}}}\tp \numb} = \innerproduct{\numb} {\OP{\bar{\bm{\nu}}} \numb}  = 1 - (\innerproduct{\bar{\numb}}{\numb})^{\sss{2}}$ (notice that 	$Q_{\sss{\bar{\bm{\nu}}}} Q_{\sss{\bar{\bm{\nu}}}}\tp = \OP{\bar{\bm{\nu}}}$). 
	Since $\innerproduct{\bar{\numb}}{\numb}> 0$, for any $\numb \in \ConeOpen(\frac{\pi}{2},\bar{\bm{\nu}})$, it follows that $\innerproduct{\bar{\numb}} {\numb} = \sqrt{1 - \|\hmb_{\sss{\bar{\bm{\nu}}}}(\numb)\|^{\sss{2}}}$.
	The implications in the Proposition follow since $\innerproduct{\bar{\numb}}{\numb} = \sqrt{1 - \|\hmb_{\sss{\bar{\bm{\nu}}}}(\numb)\|^{\sss{2}}}$ is decreasing with $\|\hmb_{\sss{\bar{\bm{\nu}}}}(\numb)\|$, and since $\numbi{i},\numbi{j} \in \ConeOpen(\frac{\pi}{2},\bar{\bm{\nu}})$.
\end{proof}

Consider now the solution $\bm{\nu}: \Rn[]_{\sss{\ge 0}} \mapsto (\Sn{n})^{\sss{N}}$ of~\eqref{eq:GeneralForm}-\eqref{eq:GeneralForm2} with $\bm{\nu}(0) \in \ConeOpen(\frac{\pi}{2},\bar{\bm{\nu}})^{\sss{N}}$ for some $\bar{\bm{\nu}} \in \Sn{n}$, which, as will be shown in Theorem \ref{thm:Synchronization}, remains in $\ConeOpen(\frac{\pi}{2},\bar{\bm{\nu}})^{\sss{N}}$ for all $t \ge 0$;
and define $\xmb: \Rn[]_{\sss{\ge 0}} \mapsto (\Rn[n])^{\sss{N}}$ as $\xmb = \hmb_{\sss{\bar{\bm{\nu}}}}^{\sss{N}} \circ \bm{\nu}$. Then, based on the transformation introduced in Definition~\ref{defn:NormalError}, it follows that $\dot{\xmb}(t) = \fmb_{\sss{\sigma(t)}}(\xmb(t))$, where  
\begin{align}
	\fmb_{\sss{\sigma(t)}}(\xmb)
	=
	d\hmb_{\sss{\bar{\bm{\nu}}}}^{\sss{N}}(\bm{\nu}) 
	\tilde{\fmb}_{\sss{\sigma(t)}}(\bm{\nu})
	|_{
		\bm{\nu} 
		= 
		(\hmb_{\sss{\bar{\bm{\nu}}}}^{\sss{-1}})^{\sss{N}}(\xmb)
	},
	\label{eq:GeneralDynamicsFull}
\end{align}
and $\xmbiDot{i} = \hmb_{\sss{\bar{\bm{\nu}}}} \circ \bm{\nu}_{\sss{i}}$ evolves according to $\xmbiDot{i}(t) = \fmb_{\sss{i,\sigma(t)}}(\xmb(t))$, where
\begin{align}
	\fmb_{\sss{i,\sigma(t)}}(\xmb)
	& = 
	d\hmb_{\sss{\bar{\bm{\nu}}}}(\bm{\nu}_{\sss{i}}) \tilde{\fmb}_{\sss{\sigma(t)}}(\bm{\nu})
	|_{
		\bm{\nu} 
		= 
		(\hmb_{\sss{\bar{\bm{\nu}}}}^{\sss{-1}})^{\sss{N}}(\xmb)
	}
	\label{eq:GeneralDynamicsComponent}
	\\
	{\sss{\text{\eqref{eq:GeneralForm2} and Def~\ref{defn:NormalError}}}}
	& =
	Q_{\sss{\bar{\bm{\nu}}}}\tp 
	\sum_{j \in \Nset_{\sss{i}}(\sigma(t))} 
	\tilde{w}_{\sss{ij}}(\numbi{i},\numbi{j})
	\OP{\numbi{i}}
	\numbi{j}	
	|_{
		\bm{\nu} 
		= 
		(\hmb_{\sss{\bar{\bm{\nu}}}}^{\sss{-1}})^{\sss{N}}(\xmb)
	}
	.
\end{align}
It follows from~\eqref{eq:GeneralDynamicsComponent} that, for $\numb  \in \ConeOpen(\frac{\pi}{2},\bar{\bm{\nu}})^{\sss{N}}$ with $\bar{\bm{\nu}} \in \Sn{n}$, $\xmb = \hmb_{\sss{\bar{\bm{\nu}}}}^{\sss{N}}(\numb) \in \mathcal{B}(1)^{\sss{N}}$ and any $p \in \mathcal{P}$,
\begin{align}
	& 
	\innerproduct{\xmbi{i}} {\fmb_{\sss{i,p}}(\xmb)}
	\overset{\sss{\eqref{eq:GeneralDynamicsComponent}}}{=}
	\langle
		Q_{\sss{\bar{\bm{\nu}}}}\tp\numbi{i}
		,
		Q_{\sss{\bar{\bm{\nu}}}}\tp
		\sum_{j \in \Nset_{\sss{i}}(p)} 
		\tilde{w}_{\sss{ij}}(\numbi{i},\numbi{j})
		\OP{\numbi{i}}
		\numbi{j}
	\rangle
	\\
	& =
	\langle
		\numbi{i},
		\OP{\bar{\bm{\nu}}}
		\sum_{j \in \Nset_{\sss{i}}(p)} 
		\tilde{w}_{\sss{ij}}(\numbi{i},\numbi{j})
		\OP{\numbi{i}}
		\numbi{j}	
	\rangle
	\\
	& =
	\langle
		\numbi{i} - \innerproduct{\numbi{i}}{\bar{\numb}}\bar{\numb},
		\sum_{j \in \Nset_{\sss{i}}(p)} 
		\tilde{w}_{\sss{ij}}(\numbi{i},\numbi{j})
		\OP{\numbi{i}}
		\numbi{j}
	\rangle
	\\
	& =
	- \innerproduct{\bar{\numb}}{\numbi{i}}
	\sum_{j \in \Nset_{\sss{i}}(p)} 
	\tilde{w}_{\sss{ij}}(\numbi{i},\numbi{j})
	\langle
		\bar{\numb},
		\OP{\numbi{i}}\numbi{j}	
	\rangle
	\label{eq:NegativeSemiDefinite}	
\end{align}
%
%
%
The following result provides certain properties that are exploited in determining the sign of~\eqref{eq:NegativeSemiDefinite}.
\begin{prop}
	\label{prop:UnitVectorsInequality}
	Consider three unit vectors $\numbi{1}$, $\numbi{2},\bar{\numb} \in \Sn{n}$, satisfying $0 < \innerproduct{\bar{\numb}}{\numbi{1}} \le \innerproduct{\bar{\numb}}{\numbi{2}}$. 
	Then $\emph{(a)}$ $\innerproduct{\bar{\numb}}{\OP{\numbi{1}}\numbi{2}} = 0$ iff $\numbi{1} = \numbi{2} $, and $\emph{(b)}$ $ \innerproduct{\bar{\numb}}{\OP{\numbi{1}}\numbi{2}}> 0$ iff $\numbi{2} \ne \numbi{1}$.
\end{prop}
The proof is found in Appendix~\ref{app:AuxiliaryResults}.
We show next, by combining Propositions~\ref{prop:ReconstructionNi} and~\ref{prop:UnitVectorsInequality} and exploiting~\eqref{eq:NegativeSemiDefinite}, that the conditions of Theorem~\ref{thm:MainTheorem} are satisfied for the dynamics~\eqref{eq:GeneralDynamicsFull}-\eqref{eq:GeneralDynamicsComponent}.
\begin{prop}
	\label{eq:ErrorDynamics}
	Consider the vector field as defined in~\eqref{eq:GeneralDynamicsFull}-\eqref{eq:GeneralDynamicsComponent} for a certain $\bar{\bm{\numb}} \in \Sn{n}$ and assume that the switching signal $\sigma: \Rpositive \mapsto \mathcal{P}$ encodes only connected network graphs.
	Then the vector field~\eqref{eq:GeneralDynamicsFull}-\eqref{eq:GeneralDynamicsComponent} satisfies the conditions of Theorem~\ref{thm:MainTheorem} for $r = 1$.
\end{prop}

\begin{proof}
	In order to verify that the conditions of Theorem~\ref{thm:MainTheorem} are satisfied by the vector field in~\eqref{eq:GeneralDynamicsFull}-\eqref{eq:GeneralDynamicsComponent}, we exploit \eqref{eq:NegativeSemiDefinite} and the fact that for each $\xmb = (\xmbi{1},\cdots,\xmbi{N}) \in \mathcal{B}(1)^{\sss{N}}$ there exists a (unique) $\numb = (\numbi{1},\cdots,\numbi{N})  \in \ConeOpen(\alpha^{\sss{\star}},\bar{\bm{\nu}})^{\sss{N}}$  such that $\xmb = \hmb_{\sss{\bar{\bm{\nu}}}}^{\sss{N}}(\numb)$. We proceed with the verification  of Condition 1) of Theorem~\ref{thm:MainTheorem} and pick $\xmb \in \mathcal{B}(1)^{\sss{N}}$, where  $\xmb = \hmb_{\sss{\bar{\bm{\nu}}}}^{\sss{N}}(\numb)$ for certain $\numb \in \ConeOpen(\alpha^{\sss{\star}},\bar{\bm{\nu}})^{\sss{N}}$.
	Notice, that since $\numb \in \ConeOpen(\alpha^{\sss{\star}},\bar{\bm{\nu}})^{\sss{N}}$, it holds by definition that $\innerproduct{\bar{\numb}}{\numbi{i}} > \cos(\alpha^{\sss{\star}}) \ge 0 \forall i \in \Nset$.
	Therefore, since Condition 1)a) depends exclusively on the sign of~\eqref{eq:NegativeSemiDefinite}, we can ignore the effect of the positive term $\innerproduct{\bar{\numb}}{\numbi{i}}$.
	
	In order to show Condition 1)a), pick any  $p \in \mathcal{P}$ and notice, that due to~\eqref{eq:gTildeCase} and continuity of  $\tilde{w}_{\sss{ij}}$ it holds that $\tilde{w}_{\sss{ij}}(\numbi{i},\numbi{j}) \ge 0 $ for any $\numbi{i},\numbi{j}\in \ConeOpen(\alpha^{\sss{\star}},\bar{\bm{\nu}})$ and $i,j \in \Nset$. 
	In addition, by recalling that $\mathcal{H}(\xmb = (\xmb_{\sss{1}},\cdots,\xmb_{\sss{N}})) = \arg \max_{\sss{i \in \Nset}} \|\xmbi{i}\|$ and thus, that  $\|\xmbi{i}\| \ge \|\xmbi{j}\|$ for all  $i \in \mathcal{H}(\xmb)$ and $j \in \Nset$, it follows from Proposition~\ref{prop:ReconstructionNi} that $\innerproduct{\bar{\numb}}{\numbi{i}} \le \innerproduct{\bar{\numb}}{\numbi{j}}$  for all  $j \in \Nset$. 
	From the latter and the result of  Proposition~\ref{prop:UnitVectorsInequality}, we get that $\innerproduct{\bar{\numb}} {\OP{\numbi{i}}\numbi{j}}\ge 0$ for all $j \in \Nset(p)$. 
	Thus, we conclude from \eqref{eq:NegativeSemiDefinite}  that for any $p \in \mathcal{P}$ and $\xmb \in \mathcal{B}(1)^{\sss{N}}$ it holds $\innerproduct{\xmbi{i}}{\fmb_{\sss{i,p}}(\xmb)}\le 0$, which by virtue of \eqref{eq:Deltaxmax} implies that Condition 1)a) is satisfied.
	
	For the verification of Condition 1)b), we additionally assume that $\xmb \not\in \mathcal{C} := \{ (\xmb_{\sss{1}},\cdots,\xmb_{\sss{N}}) \in (\Rn[n])^{\sss{N}}: \xmbi{1} = \cdots = \xmbi{N}\}$. We will show that for each $p \in \mathcal{P}$ there exists $k \in \mathcal{H}(\xmb)$ such that $\innerproduct{\xmbi{k}}{\fmb_{\sss{k,p}}(\xmb)} < 0$. Indeed, suppose on the contrary that there exists $p \in \mathcal{P}$ such that
	\begin{align}
		\innerproduct{\xmbi{i}}{\fmb_{\sss{i,p}}(\xmb)} =  0 \, \forall i \in \mathcal{H}(\xmb),
		\label{eq:ContradictionAssumption}
	\end{align}
	and recall that the agents' network is not synchronized, since $\xmb \not\in \mathcal{C}$.
	Consider then an $l \in \mathcal{H}(\xmb)$, for which $\innerproduct{\xmbi{l}}{\fmb_{\sss{l,p}}(\xmb)} = 0$ according to assumption~\eqref{eq:ContradictionAssumption}.
	Notice also, that due to~\eqref{eq:gTildeCase}, Propositions~\ref{prop:ReconstructionNi} and~\ref{prop:UnitVectorsInequality} and \eqref{eq:NegativeSemiDefinite}, it can be shown (as in the proof of Condition 1)a) above) that $\innerproduct{\xmbi{l}}{\fmb_{\sss{l,p}}(\xmb)} = 0$ is satisfied only if all neighbors of agent $l$ are synchronized with agent $l$, i.e., only if $\numbi{j} = \numbi{l} \Leftrightarrow \xmbi{j} = \xmbi{l}$ for all $j \in \Nset_{\sss{l}}(p)$.
	This implies that all $j \in \Nset_{\sss{l}}(p)$ are contained in $\mathcal{H}(\xmb)$, i.e., $ \Nset_{\sss{l}}(p) \cup \{l\} \subseteq \mathcal{H}(\xmb)$.
	As such, by assumption~\eqref{eq:ContradictionAssumption}, $\innerproduct{\xmbi{j}}{\fmb_{\sss{j,p}}(\xmb)} =  0$ for all $j \in \Nset_{\sss{l}}(p)$, which means that the previous rationale is applicable for all $j \in \Nset_{\sss{l}}(p)$, thus leading to the conclusion that all neighbors of all neighbors of agent $l$ are necessarily synchronized with each other. 
	Since the graph encoded by $p \in \mathcal{P}$ is connected, the previous rationale, applied $N-1$ times, leads to the conclusion that all agents are synchronized.
	Since $\xmb \not\in \mathcal{C}$, a contradiction has been reached, and therefore, for each $p \in \mathcal{P}$, there exists a $k \in \mathcal{H}(\xmb) $ for which $\innerproduct{\xmbi{k}}{\fmb_{\sss{k,p}}(\xmb)} < 0$, and therefore condition~1)b) of Theorem~\ref{thm:MainTheorem} is satisfied.
	
	Finally, let $\xmb \in \mathcal{C}$. 
	Since $\xmbi{1} = \cdots = \xmbi{N} \Leftrightarrow \numbi{1} = \cdots = \numbi{N}$, it follows from Proposition~\ref{prop:UnitVectorsInequality} that $\innerproduct{\xmbi{i}}{\fmb_{\sss{i,p}}(\xmb)} = \zvec$ for all $p \in \mathcal{P}$ and $i \in \Nset$.
	Thus, the second condition of Theorem~\ref{thm:MainTheorem} is also satisfied.
\end{proof}

\begin{thm}
	\label{thm:Synchronization}
	%
	%
	Consider the solution $\bm{\nu}: \Rn[]_{\sss{\ge 0}} \mapsto (\Sn{n})^{\sss{N}}$ of~\eqref{eq:GeneralForm} with $\bm{\nu}(0) \in \ConeOpen(\alpha^{\sss{\star}},\bar{\bm{\nu}})^{\sss{N}}$ for some $\bar{\bm{\nu}} \in \Sn{n}$.
	Then, for a network graph connected at all times, \emph{i)}  $\bm{\nu}(t) \in \ConeClosed(\alpha,\bar{\bm{\nu}})^{\sss{N}}$ for all $ t \ge 0$, where $\alpha = \arccos(\max_{\sss{i \in \Nset}} \innerproduct{\bar{\numb}}{\numbi{i}}(0)) \in [0,\alpha^{\sss{\star}})$; 
	\emph{ii)} $\bm{\nu}$ synchronizes asymptotically and $\lim_{\sss{t \rightarrow \infty}} \bar{\numb} \numbi{i}(t)$ exists for all $i \in \Nset$; and \emph{iii)} all unit vectors converge to a constant unit vector,  i.e. $ \exists \numb^{\sss{\star}} \in \Sn{n}:  \lim_{\sss{t \rightarrow \infty}} \numb(t) \in \mathcal{C}(0,\numb^{\sss{\star}})^{\sss{N}}$.
\end{thm}

\begin{proof}
	Consider a solution $\bm{\nu}$ of~\eqref{eq:GeneralForm}, and $\xmb = \hmb_{\sss{\bar{\bm{\nu}}}}^{\sss{N}} \circ \numb$.
	%
	%
	Since $\alpha^{\sss{\star}}$ is either $\frac{\pi}{2}$ or $\frac{\pi}{4}$, then $\bm{\nu}(0)$ is within the domain of $\hmb_{\sss{\bar{\bm{\nu}}}}^{\sss{N}}$ and moreover $\bar{\mathcal{B}}(r_{\sss{0}})^{\sss{N}} \subset \mathcal{B}(1)^{\sss{N}}$ where $r_{\sss{0}} \overset{\sss{\eqref{eq:Diffeomorphism}}}{=} \max_{\sss{ i \in \Nset}} \|Q_{\sss{\bar{\numb}}}\tp \numbi{i}(0)\| \overset{\sss{\text{Prop}~\ref{prop:ReconstructionNi}}}{=}  \sqrt{1 -\max_{\sss{ i \in \Nset}} (\innerproduct{\bar{\numb}}{\numbi{i}}(0))^{\sss{2}}} < 1$ (where the latter inequality follows from the fact that $\bm{\nu}(0) \in \ConeOpen(\alpha^{\sss{\star}},\bar{\bm{\nu}})^{\sss{N}}$).
	From Proposition~\ref{eq:ErrorDynamics}, the dynamics~\eqref{eq:GeneralDynamicsFull} satisfy Theorem's~\ref{thm:MainTheorem} conditions and therefore the set $\bar{\mathcal{B}}(r_{\sss{0}})^{\sss{N}}$ is positively invariant for trajectories of $\dot{\xmb}(t) = \fmb_{\sss{\sigma(t)}}(\xmb(t))$.
	This, in turn, implies that the set $\Hmb_{\sss{\bar{\bm{\nu}}}}^{\sss{-1}}(\bar{\mathcal{B}}(r_{\sss{0}})^{\sss{N}}) = \ConeClosed(\alpha,\bar{\bm{\nu}})^{\sss{N}}$, where $\alpha = \arccos(\max_{\sss{i \in \Nset}} \innerproduct{\bar{\numb}}{\numbi{i}}(0)) \in [0,\alpha^{\sss{\star}})$, is positively invariant for trajectories of $\dot{\numb}(t) = \tilde{\fmb}_{\sss{\sigma(t)}}(\numb(t))$;
	i.e, all unit vectors are forever contained in the closed $\alpha$-cone they start on.
	This suffices to conclude part \emph{i)} in the Theorem.
	%
	%
	
	Let us now focus on part \emph{ii)} of the Theorem.
	From Proposition~\ref{eq:ErrorDynamics}, the dynamics~\eqref{eq:GeneralDynamicsFull} satisfy Theorem's~\ref{thm:MainTheorem} conditions.
	It follows from Theorem~\ref{thm:MainTheorem} that $\lim_{\sss{t \rightarrow \infty}} \xmbi{i}(t) - \xmbi{j}(t) = 0$ for all $i,j \in \Nset$, which implies that $\lim_{\sss{t \rightarrow \infty}} \numbi{i}(t) - \numbi{j}(t) = 0$, for all $i,j \in \Nset$ (see Proposition~\ref{prop:ReconstructionNi}).
	Moreover, it follows that the Lyapunov function in Theorem~\ref{thm:MainTheorem} converges to a constant, i.e., $\lim_{\sss{t \rightarrow \infty}} V(\xmb(t)) = \lim_{\sss{t \rightarrow \infty}} \max_{\sss{i \in \Nset}} \frac{1}{2}\|\xmbi{i}(t)\|^{\sss{2}} = \lim_{\sss{t \rightarrow \infty}} \frac{1}{2}\|\xmbi{1}(t)\|^{\sss{2}} = V^{\sss{\infty}}$, for some constant $0 \le V^{\sss{\infty}} \le V(0) <\frac{1}{2}$.
	From Proposition~\ref{prop:ReconstructionNi}, it follows that $\lim_{\sss{t \rightarrow \infty}} \bar{\numb} \numbi{i}(t) = \lim_{\sss{t \rightarrow \infty}} \sqrt{1 - \|\xmbi{i}(t)\|^{\sss{2}}} = \sqrt{1 - 2V^{\sss{\infty}}}$.	
	
	We now prove part \emph{iii)} of the Theorem.
	Since $\bm{\nu}(0) \in \ConeOpen(\alpha^{\sss{\star}},\bar{\bm{\nu}})^{\sss{N}}$ for some $\bar{\bm{\nu}} \in \Sn{n}$, Proposition~\ref{prop:ConeGeneration} guarantees that there exist $n+1$ linearly independent unit vectors $\{\bar{\numb}_{\sss{1}},\cdots,\bar{\numb}_{\sss{n+1}}\}$ such that $\bm{\nu}(\cdot) \in \ConeOpen(\alpha^{\sss{\star}},\bar{\bm{\nu}}_{\sss{k}})^{\sss{N}}$ for all $k\in \{1,\cdots,n+1\} $.
	From part \emph{ii)} of this Theorem, it follows that, for each $k \in \{1,\cdots,n+1\}$, there exists a constant $V_{\sss{k}}^{\sss{\infty}} < \frac{1}{2}$ such that  $\lim_{\sss{t \rightarrow \infty}}  \bar{\numb}_{\sss{k}}\tp\numbi{1}(t) = \sqrt{1 - 2 V_{\sss{k}}^{\sss{\infty}}}$.
	Thus, it follows that $\lim_{\sss{t \rightarrow \infty }} A \numbi{1}(t) = b \Leftrightarrow \lim_{\sss{t \rightarrow \infty }} \numbi{1}(t) = A^{\sss{-1}} b $, where $A\tp = [\bar{\numb}_{\sss{1}} \, \cdots \bar{\numb}_{\sss{n+1}}]$ is non-singular, since $\bar{\numb}_{\sss{1}},\cdots,\bar{\numb}_{\sss{n+1}}$ are linearly independent, and $b\tp = [\sqrt{1 - 2 V_{\sss{1}}^{\sss{\infty}}} \, \cdots \, \sqrt{1 - 2 V_{\sss{n+1}}^{\sss{\infty}}}]$. 
	Since synchronization is asymptotically reached, $\lim_{\sss{t \rightarrow \infty }} \numbi{i}(t) = A^{\sss{-1}} b $ for all $i \in \Nset$.
	%
\end{proof}
	
	\section{Simulations}
	\label{sec:Simulations}
	In this section, we present simulations that illustrate some of the results in the previous sections.

All simulations are provided for a network of six agents, i.e., $\Nset = \{1,\cdots,6\}$, whose network topology is presented in Fig.~\ref{fig:Agents6Network}. 
The neighbor sets for all agents change in time: for agent 1, $\Nset_{\sss{1}}$ alternates between $\{2\}$ and $\{2,4\}$;  for agent 2, $\Nset_{\sss{2}}$ alternates between  $\{3\}$ and $\{3,6\}$;  for agent 3, $\Nset_{\sss{3}}$ alternates between $\{4\}$ and $\{4,5\}$;  for agent 4, $\Nset_{\sss{4}}$ alternates between $\{5\}$ and $ \{5,1\}$;  for agent 5, $\Nset_{\sss{5}}$ alternates between $\{6\}$ and $\{6,3\}$;  for agent 6, $\Nset_{\sss{6}}$ alternates between $\{1\}$ and $\{1,2\}$. 
For these time-varying neighbor sets, the network graph is connected at all times.

The switching time instants for the neighbor set of each agent $i\in \Nset$ are those from the sequences $\mathcal{T}^{\sss{i}} = \{\frac{1}{i} + k i\}_{\sss{k \in \mathbb{N}}}$, which are shown in the time axes in Fig~\ref{fig:Simulations}%
\if\paperextended1 (and how to obtain the switching time instants for the complete network graph is found in Algorithm~\ref{alg:ComputeT} in Appendix~\ref{app:AppendixAverageDwellTime})\fi.
For the same initial conditions, we perform two simulations, each one with different weight functions: 
for agents whose $i\in \Nset$ is even, $g_{\sss{ij}}(\theta) = j$ for both simulations;
for agents whose $i\in \Nset$ is odd, we consider $g_{\sss{ij}}(\theta) = j (2 - \cos(\theta))$ for simulation~$1$ and $g_{\sss{ij}}(\theta) = j (1 - \cos(\theta))$ for simulation~$2$.
Notice that $g_{\sss{ij}}(0) = j >0$ for the former case, and $g_{\sss{ij}}(0) = 0$ for the latter case, which means odd agents penalize small errors differently between the two simulations.

In Figs.~\ref{subfig:Trajectories_UnitVector_fast}-\ref{subfig:Distances_UnitVector_slow}, six unit vectors are randomly initialized in an open $\frac{\pi}{2}$-cone around $(1,0,0) \in \Sn{2}$. 
In Figs.~\ref{subfig:Trajectories_UnitVector_fast} and~\ref{subfig:Trajectories_UnitVector_slow}, the trajectories of the unit vectors on the unit sphere are shown, and a visual inspection indicates convergence to a synchronized network.
In Figs.~\ref{subfig:Norms_UnitVector_fast} and~\ref{subfig:Norms_UnitVector_slow}, the function $V(\xmb) = \max_{\sss{i \in \Nset}} \frac{1}{2} \|\xmbi{i}\|^2$ -- used in Theorem~\ref{thm:Synchronization} --  is provided, and we can verify that despite being non-smooth, it is almost always decreasing; notice that $V(\xmb(\cdot))$ converges to a constant which quantifies the asymptotic angular distance between all unit vectors and $(1,0,0)$.
In Figs.~\ref{subfig:Distances_UnitVector_fast} and~\ref{subfig:Distances_UnitVector_slow}, the angular distance, i.e., $\theta(\cdot,\cdot)$ as in Definition~\ref{defn:ThetaUnitVector}, between some agents is presented, and it indicates convergence to a synchronized network.
Notice that in simulation~$1$ convergence is quicker when compared to simulation~$2$. 
This is a consequence of choosing, for simulation~$2$,  weight functions that are zero when two unit vectors are synchronized, i.e., $g_{\sss{ij}}(0) = 0$ for $i$ odd. 
This means that odd agents do not penalize the error between themselves and their neighbors \emph{as much as} when they are \emph{close}, and thus leading to a slow convergence to a synchronized network.
In turn, even agents $i \in \Nset$ tend to converge to some odd agent, when $|\Nset_{\sss{i}}(\cdot)| = 1$ (if $|\Nset_{\sss{i}}(\cdot)| = 1$ then $\Nset_{\sss{i}}(\cdot)$ is a set composed of one odd number); and tend to converge to somewhere in between their two neighbors when $|\Nset_{\sss{i}}(\cdot)| = 2$. 
This explains the oscillatory behavior for simulation~$2$ in Figs.\ref{subfig:Norms_UnitVector_slow} and~\ref{subfig:Distances_UnitVector_slow}.

In Figs.~\ref{subfig:Trajectories_RotMatrix_fast}-\ref{subfig:Distances_RotMatrix_slow}, six rotation matrices were randomly initialized such that $\theta(\Idmat, \Rmati{i}) \le \frac{\pi}{2}$ for all $i \in \Nset$. 
In Figs.~\ref{subfig:Trajectories_RotMatrix_fast} and~\ref{subfig:Trajectories_RotMatrix_slow}, the trajectories of the rotation matrices are shown on a sphere of $\pi$ radius\footnote{For each rotation matrix $\Rmati{i}$, we plot $\theta_{\sss{i}} \nmbi{i}$ where $\theta_{\sss{i}} = \theta(\Idmat,\Rmati{i}) \in [0,\pi]$ and $\nmbi{i} = \frac{1}{2\sin(\theta_{\sss{i}})} \invsk{\Rmati{i} - \Rmati{i}\tp} \in \mathcal{S}^{\sss{2}}$.}, and a visual inspection indicates convergence to a synchronized network.
In Figs.~\ref{subfig:Norms_RotMatrix_fast} and~\ref{subfig:Norms_RotMatrix_slow}, the function $V(\xmb) = \max_{\sss{i \in \Nset}} \frac{1}{2} \|\xmbi{i}\|^2$ -- used in Theorem~\ref{thm:Synchronization} --  is provided, and we can verify that despite being non-smooth, it is almost always decreasing; notice that $V(\xmb(\cdot))$ converges to a constant which quantifies the asymptotic angular distance between all rotation matrices and $\Idmat$ (the rotation matrix that  all rotation matrices start \emph{close} to).
In Figs.~\ref{subfig:Distances_RotMatrix_fast} and~\ref{subfig:Distances_RotMatrix_slow}, the angular distance, i.e., $\theta(\cdot,\cdot)$ as in Definition~\ref{defn:ThetaRotMatrix}, between some agents is presented, and it indicates convergence to a synchronized network.
Notice that in simulation~$1$ convergence is quicker when compared to simulation~$2$. 
The explanation for this behavior is the same as that provided before, and it is a consequence of choosing, for simulation~$2$,  weight functions that are zero when two rotation matrices are synchronized, i.e., $g_{\sss{ij}}(0) = 0$ for $i$ odd. 
The oscillatory behavior for simulation~$2$ in Figs.\ref{subfig:Norms_RotMatrix_slow} and~\ref{subfig:Distances_RotMatrix_slow} is also explained by the same reasoning described before.
 
\begin{figure*}
	\centering
	\begin{subfigure}[b]{0.3\textwidth}
		\includegraphics[clip=true,trim=1.7cm 0cm 0cm 0cm,width=\textwidth]
		{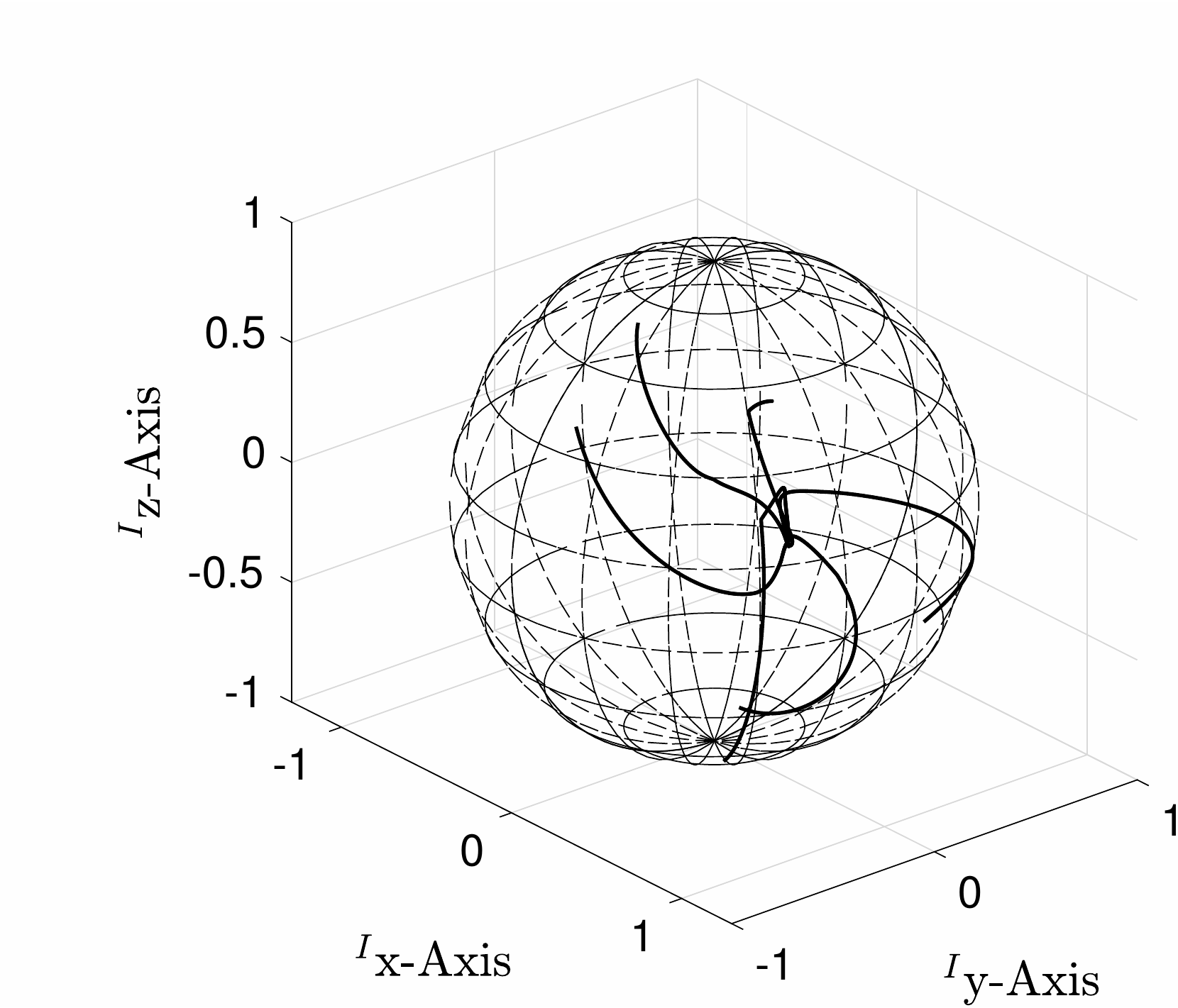}
		\caption{Simulation~$1$: Trajectories of unit vectors in unit sphere}
		\label{subfig:Trajectories_UnitVector_fast}	
	\end{subfigure}
	\begin{subfigure}[b]{0.3\textwidth}
		\includegraphics[clip=true,trim=0cm 0cm 0cm 0cm,width=\textwidth]
		{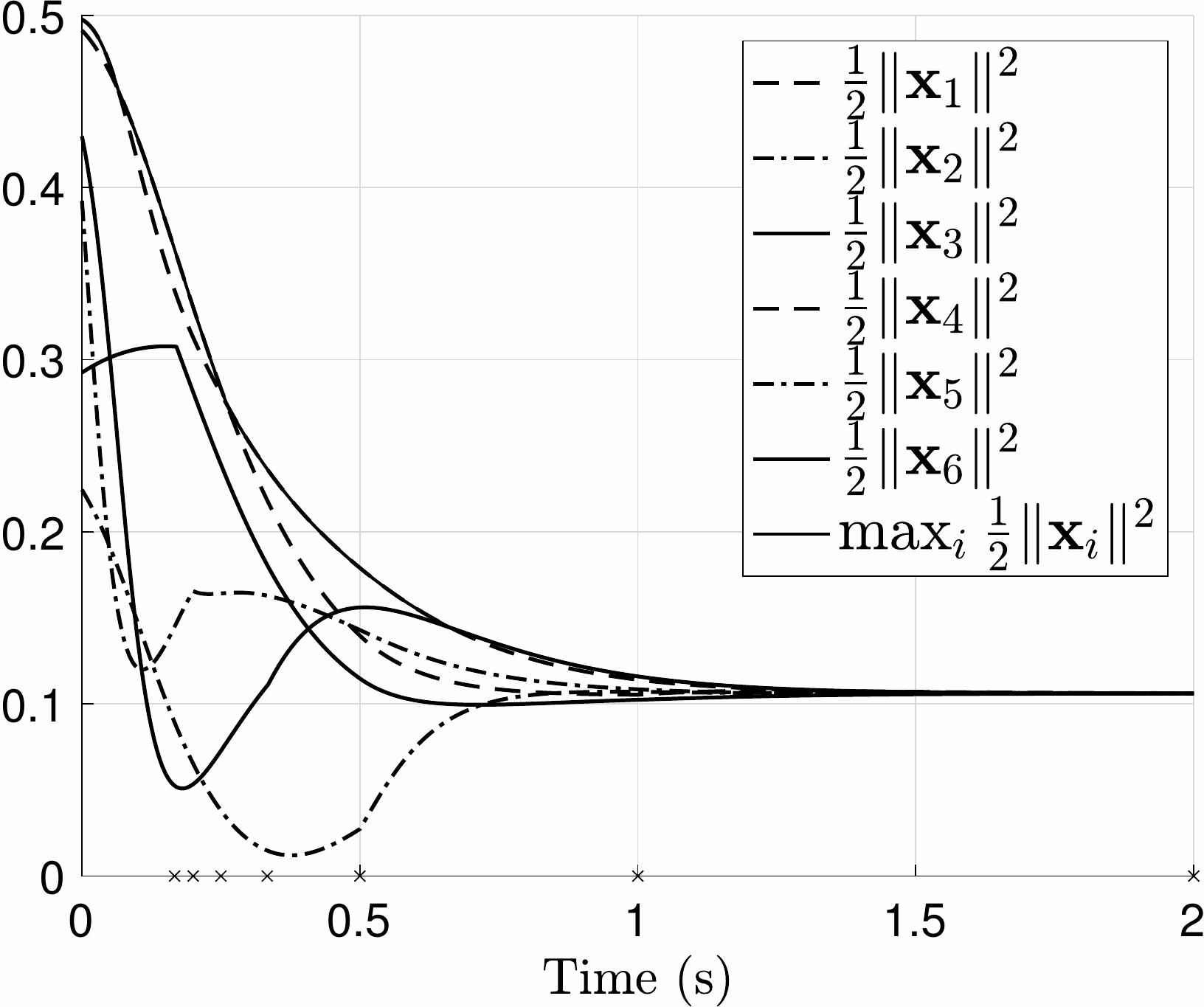}	
		\caption{Simulation~$1$: time evolution of $\frac{1}{2}\|\xmbi{i}\|^{\sss{2}}$, with $\xmbi{i}$ as in Definition~\ref{defn:NormalError}}
		\label{subfig:Norms_UnitVector_fast}
	\end{subfigure}
	\begin{subfigure}[b]{0.3\textwidth}
		\includegraphics[clip=true,trim=0cm 0cm 0cm 0cm,width=\textwidth]
		{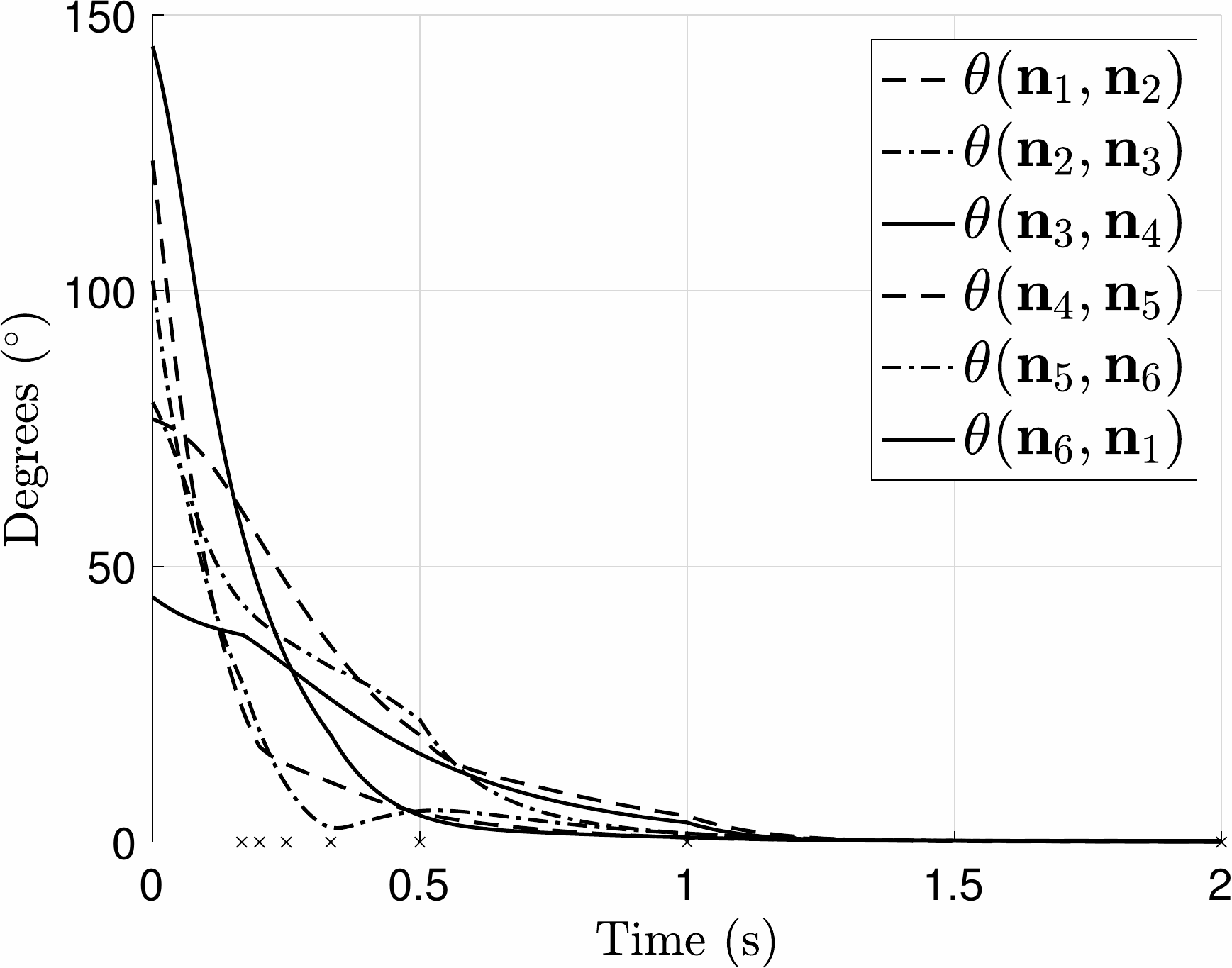}
		\caption{Simulation~$1$: angular distance between some pairs of unit vectors}
		\label{subfig:Distances_UnitVector_fast}
	\end{subfigure}
	\begin{subfigure}[b]{0.3\textwidth}
		\includegraphics[clip=true,trim=1.7cm 0cm 0cm 0cm,width=\textwidth]
		{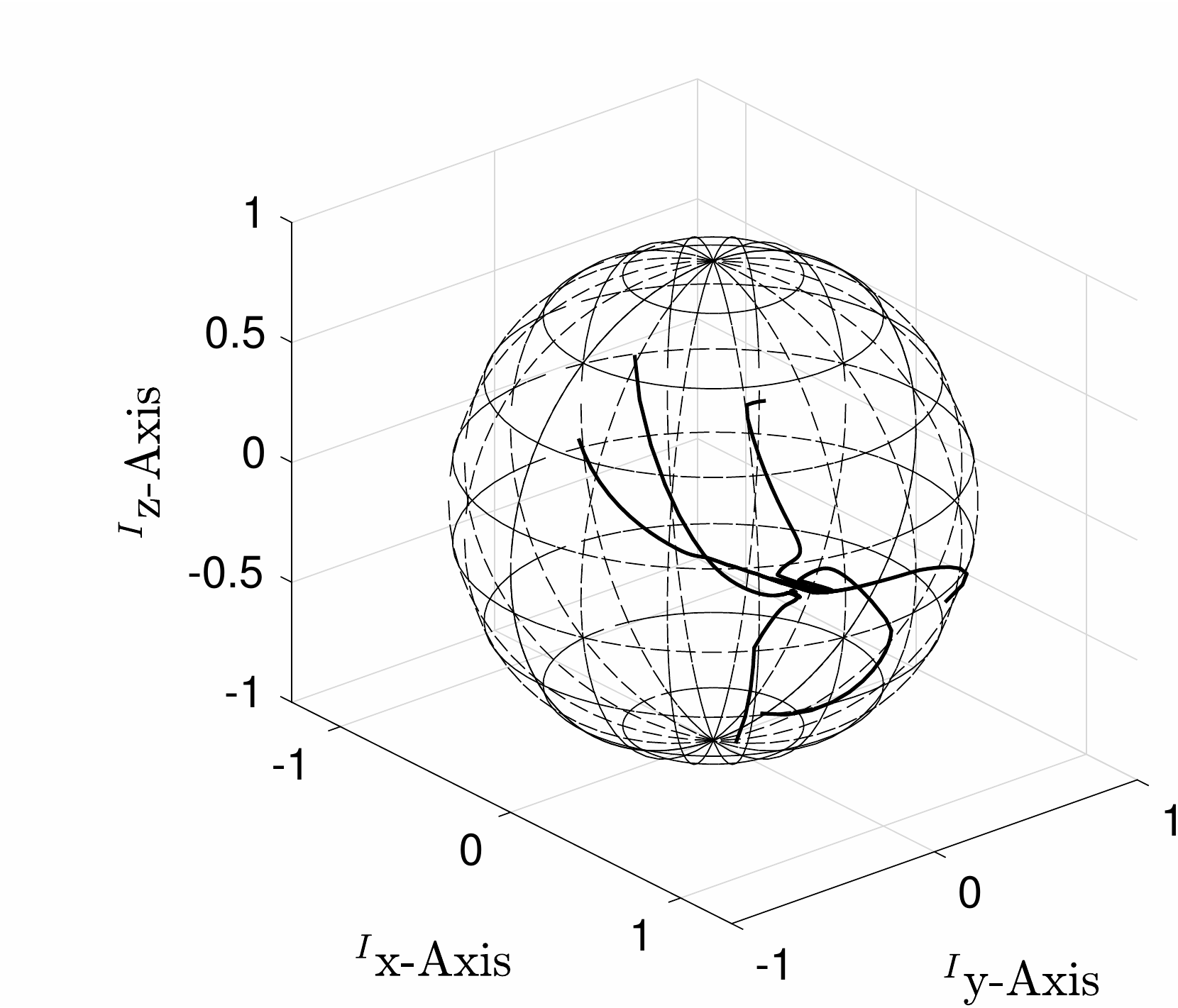}
		\caption{Simulation~$2$: Trajectories of unit vectors in unit sphere}
		\label{subfig:Trajectories_UnitVector_slow}		
	\end{subfigure}
	\begin{subfigure}[b]{0.3\textwidth}
		\includegraphics[clip=true,trim=0cm 0cm 0cm 0cm,width=\textwidth]
		{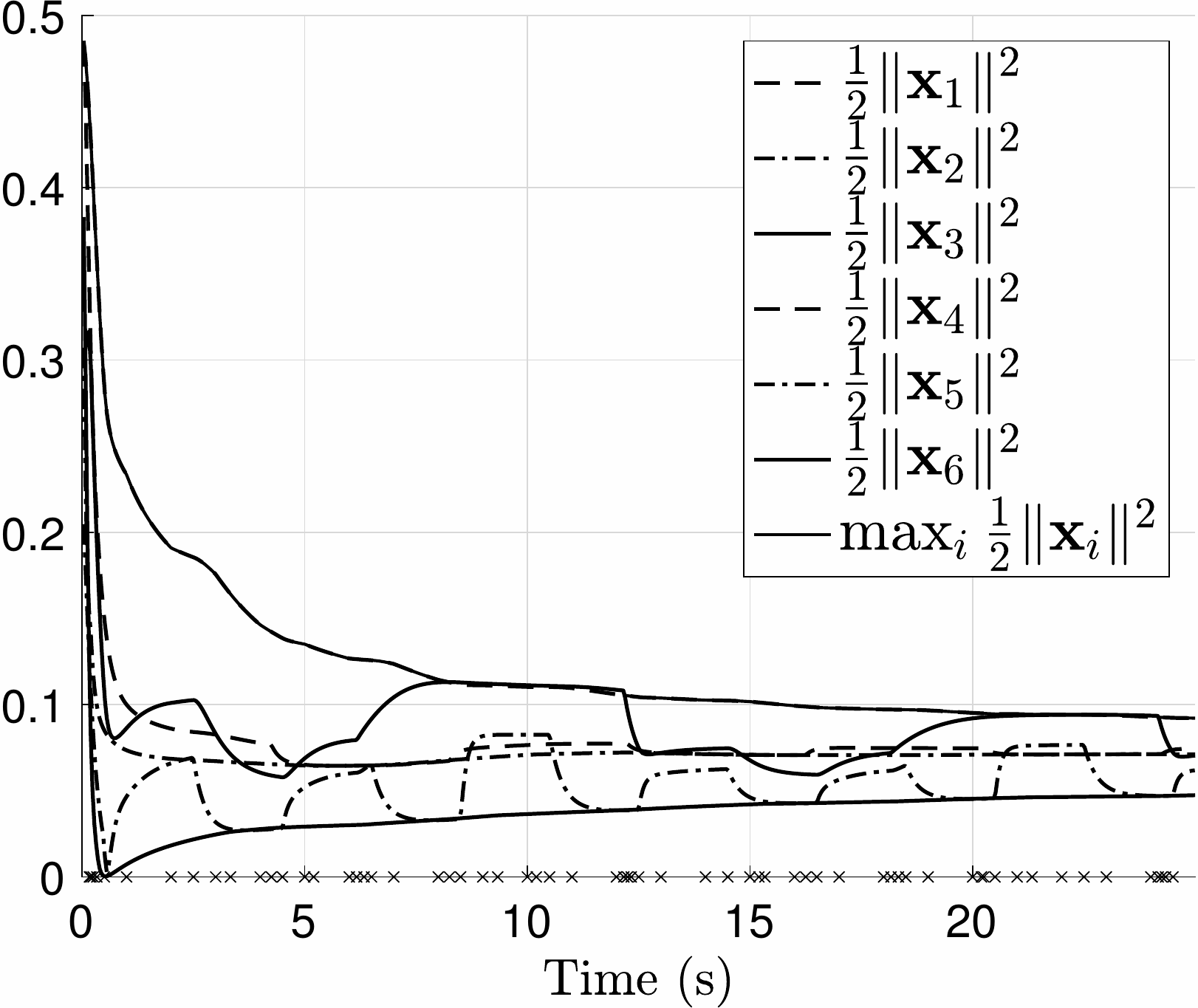}
		\caption{Simulation~$2$: time evolution of $\frac{1}{2}\|\xmbi{i}\|^{\sss{2}}$, with $\xmbi{i}$ as in Definition~\ref{defn:NormalError}}
		\label{subfig:Norms_UnitVector_slow}
	\end{subfigure}
	\begin{subfigure}[b]{0.3\textwidth}
		\includegraphics[clip=true,trim=0cm 0cm 0cm 0cm,width=\textwidth]
		{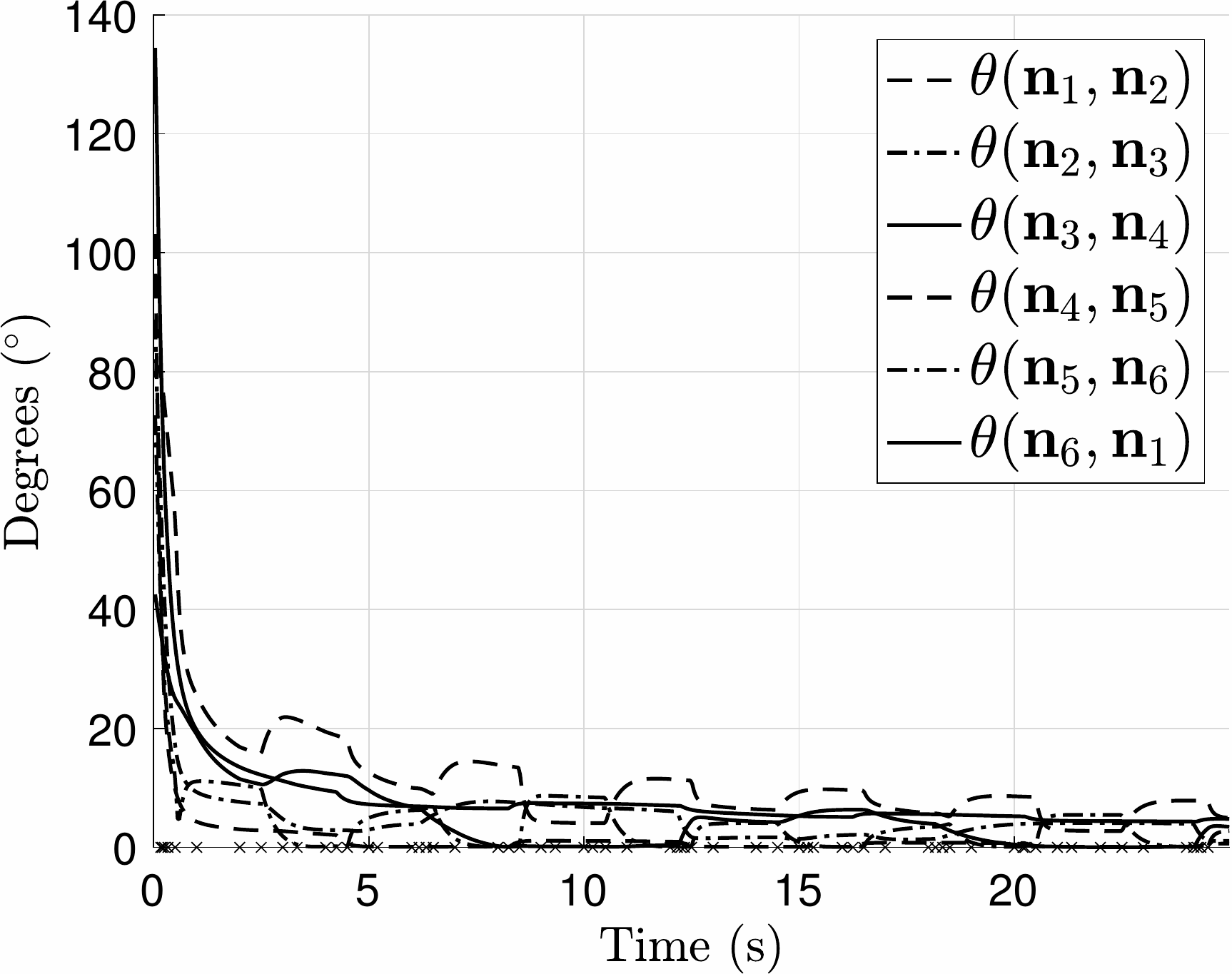}
		\caption{Simulation~$2$: angular distance between some pairs of unit vectors}
		\label{subfig:Distances_UnitVector_slow}
	\end{subfigure}		
	\begin{subfigure}[b]{0.3\textwidth}
		\includegraphics[clip=true,trim=2.5cm 0cm 0cm 2.5cm,width=\textwidth]
		{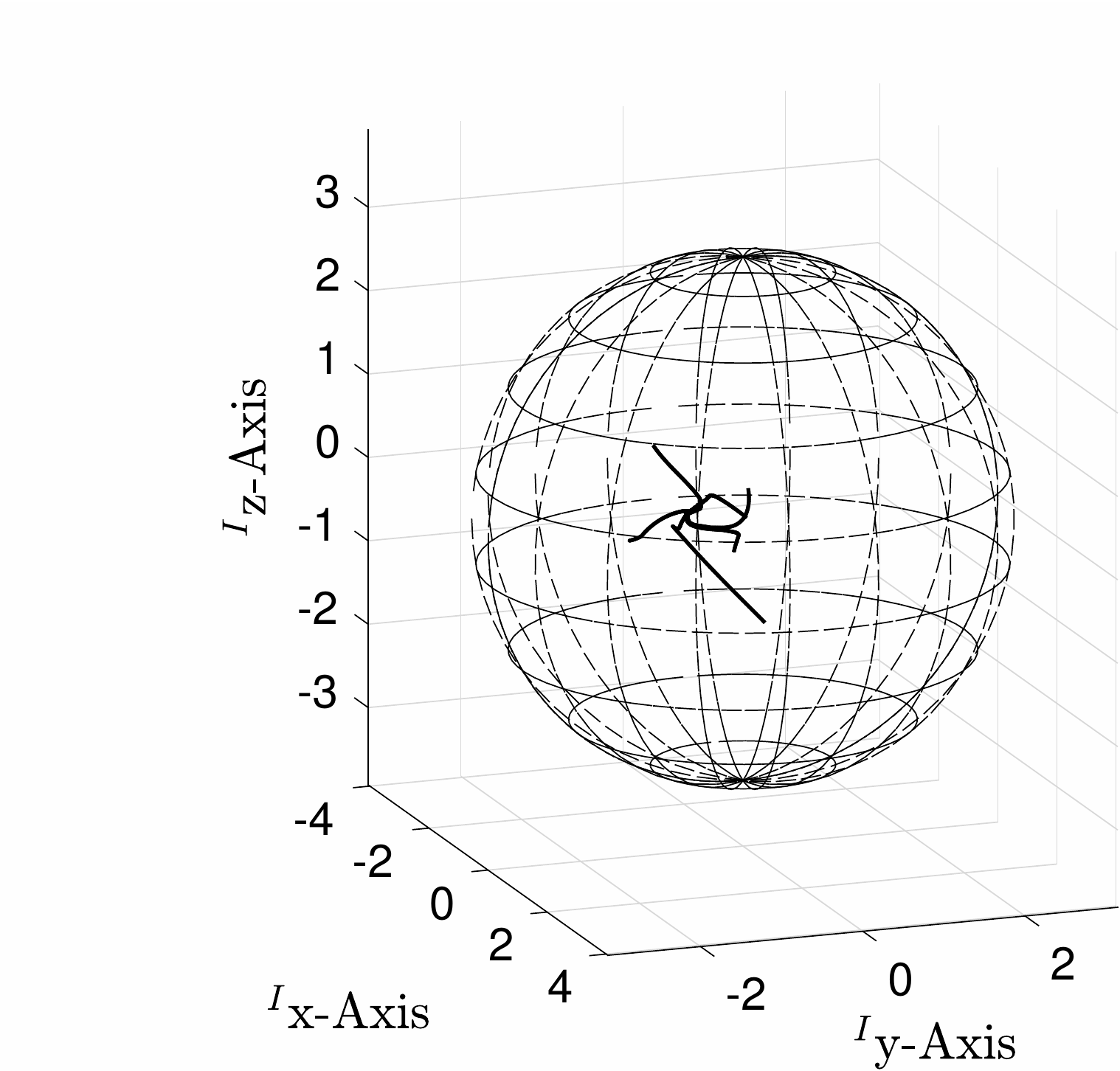}
		\caption{Simulation~$1$: Trajectories of rotation matrices in sphere of $\pi$ radius}
		\label{subfig:Trajectories_RotMatrix_fast}
	\end{subfigure}		
	\begin{subfigure}[b]{0.3\textwidth}
		\includegraphics[clip=true,trim=0cm 0cm 0cm 0cm,width=\textwidth]
		{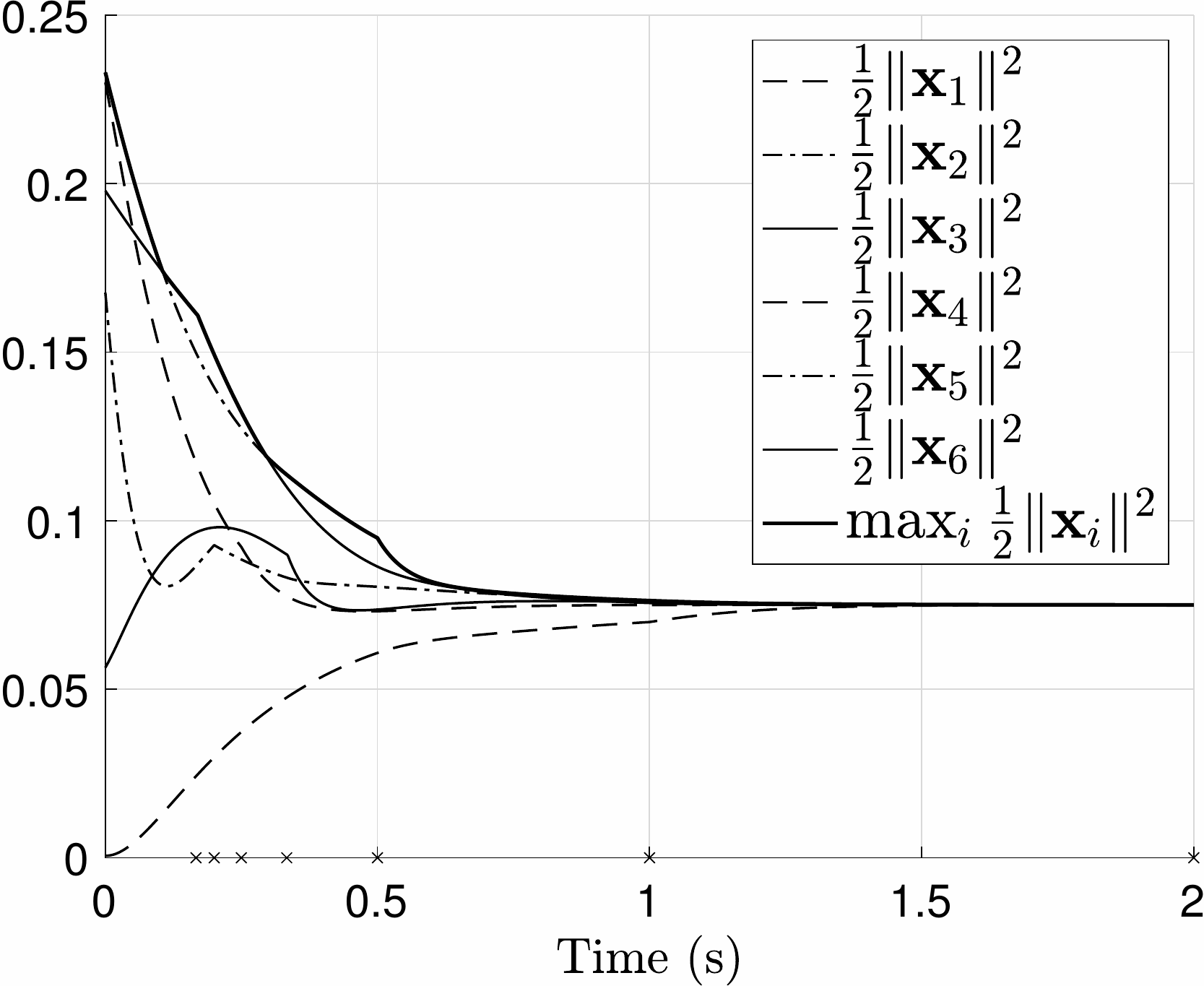}
		\caption{Simulation~$1$: time evolution of $\frac{1}{2}\|\xmbi{i}\|^{\sss{2}}$, with $\xmbi{i}$ as in Definition~\ref{defn:NormalError}}
		\label{subfig:Norms_RotMatrix_fast}
	\end{subfigure}	
	\begin{subfigure}[b]{0.3\textwidth}
		\includegraphics[clip=true,trim=0cm 0cm 0cm 0cm,width=\textwidth]
		{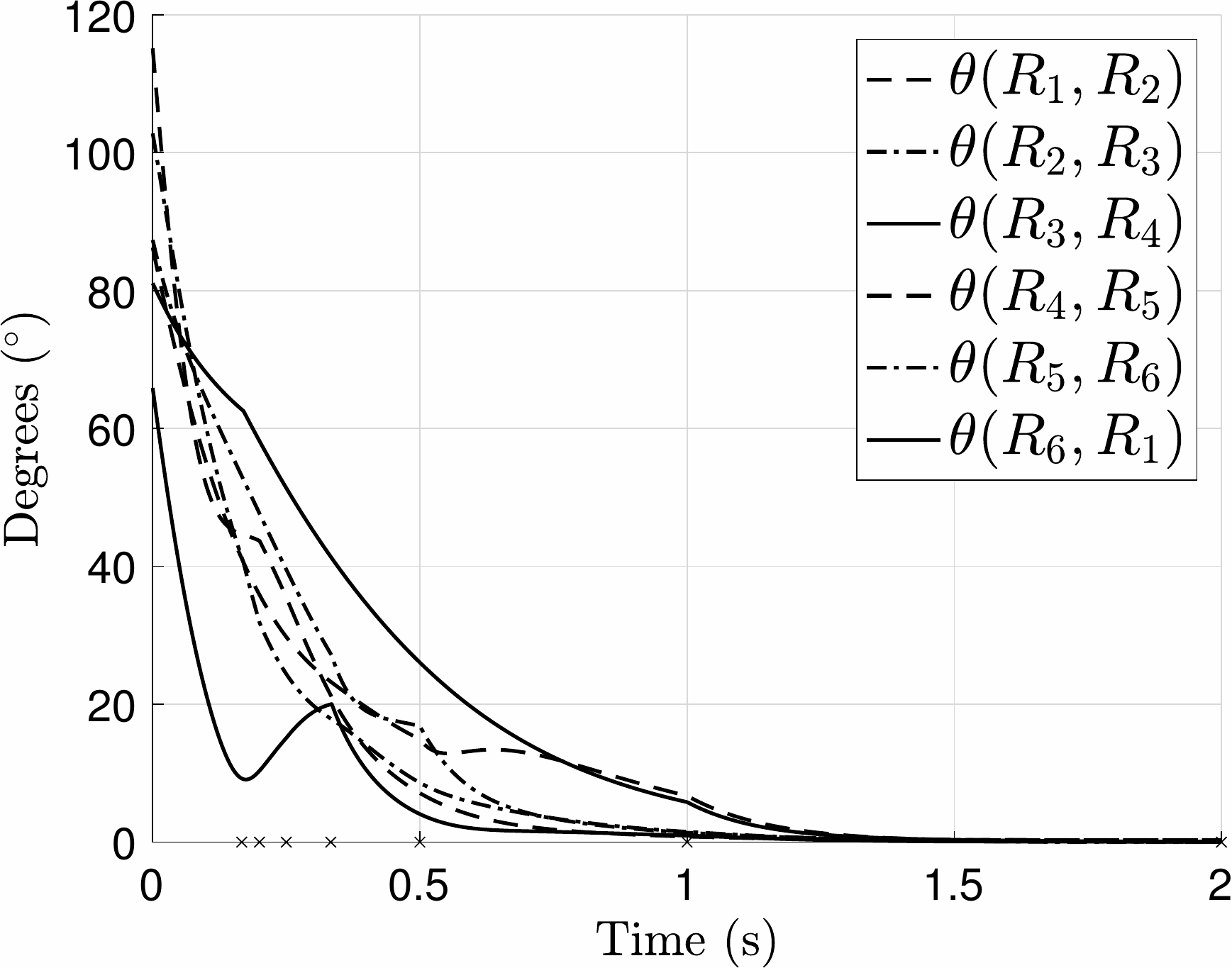}
		\caption{Simulation~$1$: angular distance between some pairs of rotation matrices}
		\label{subfig:Distances_RotMatrix_fast}
	\end{subfigure}	
	\begin{subfigure}[b]{0.3\textwidth}
		\includegraphics[clip=true,trim=2.5cm 0cm 0cm 2.5cm,width=\textwidth]
		{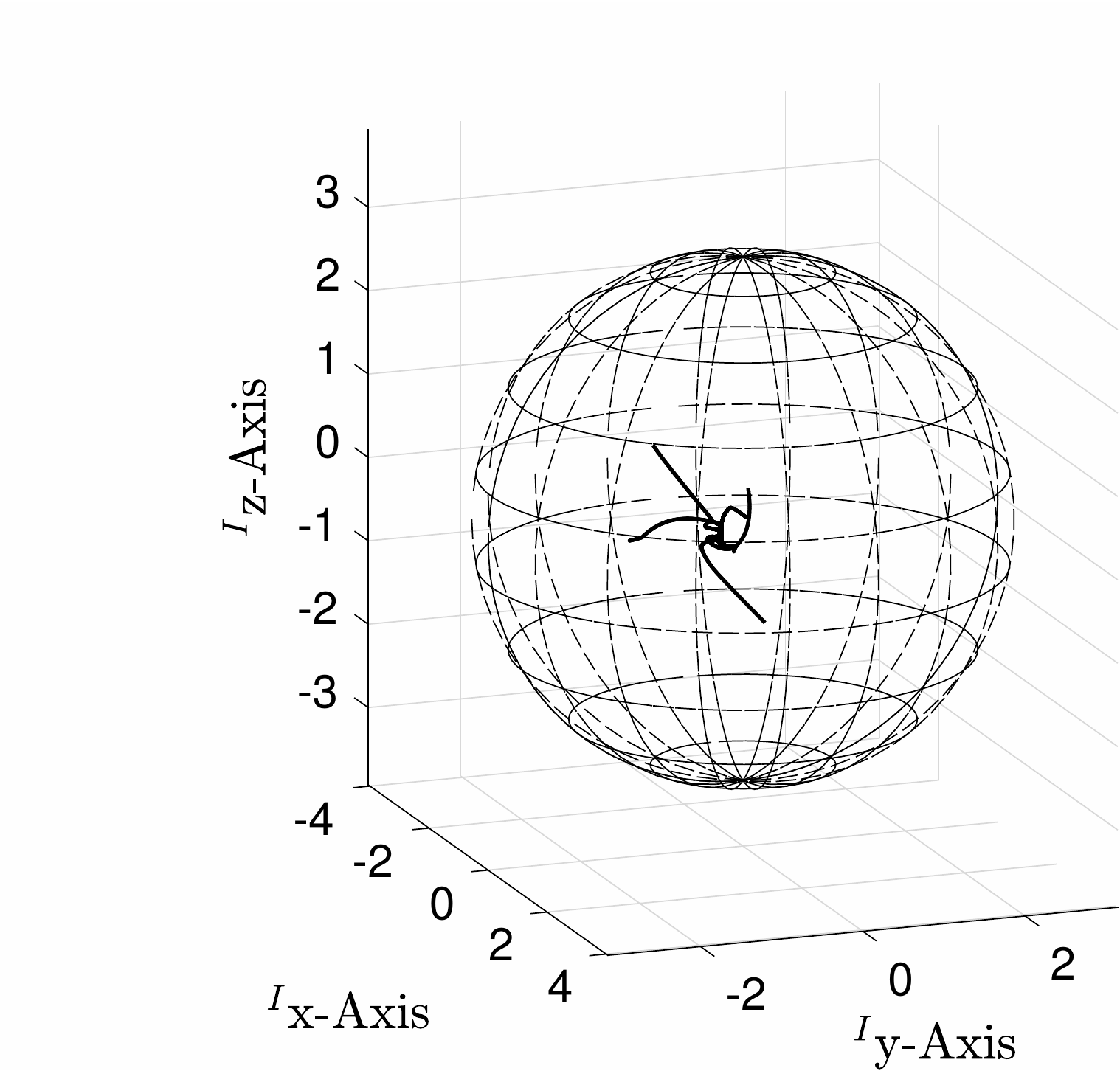}
		\caption{Simulation~$2$: Trajectories of rotation matrices in sphere of $\pi$ radius}
		\label{subfig:Trajectories_RotMatrix_slow}
	\end{subfigure}		
	\begin{subfigure}[b]{0.3\textwidth}
		\includegraphics[clip=true,trim=0cm 0cm 0cm 0cm,width=\textwidth]
		{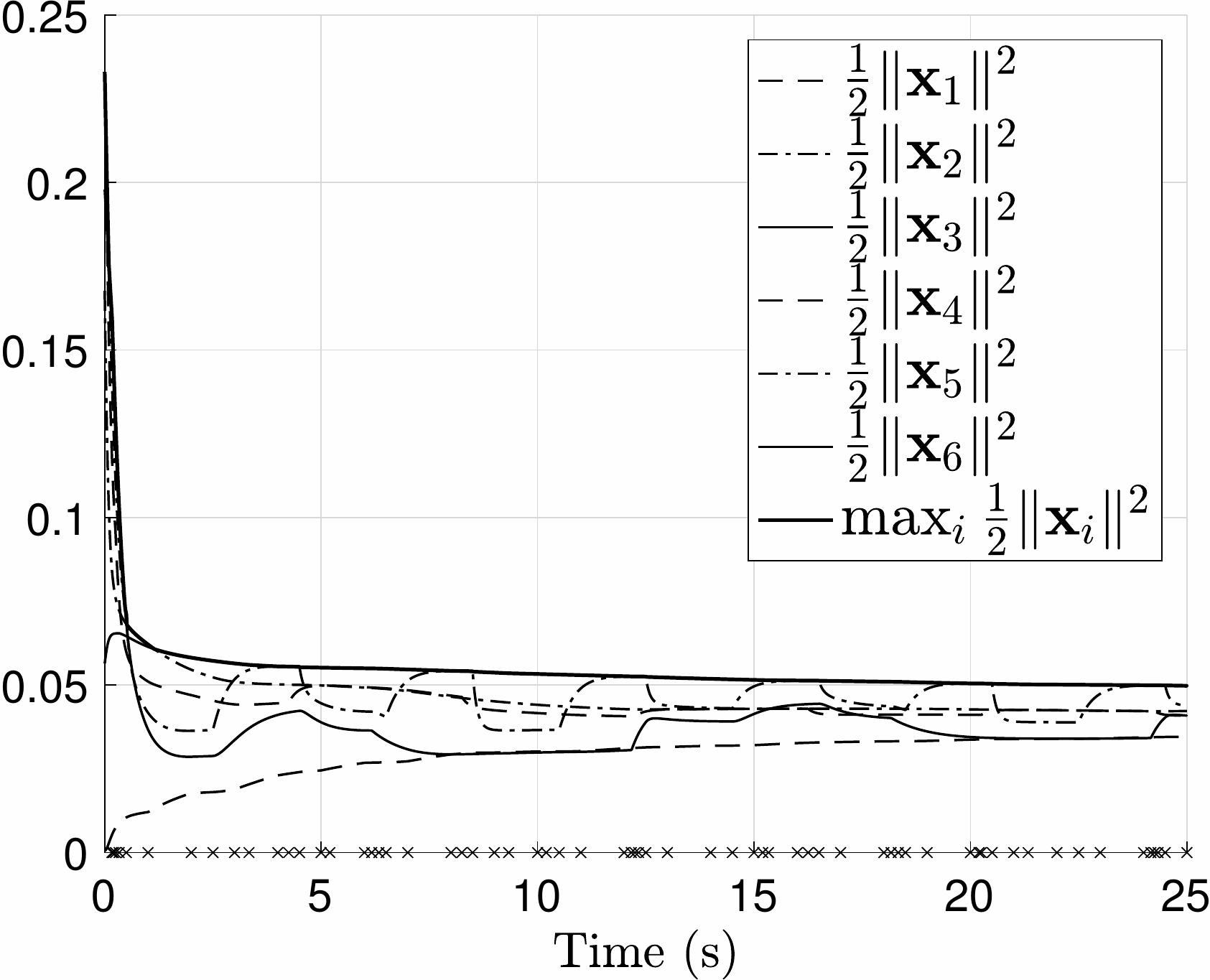}
		\caption{Simulation~$2$: time evolution of $\frac{1}{2}\|\xmbi{i}\|^{\sss{2}}$, with $\xmbi{i}$ as in Definition~\ref{defn:NormalError}}
		\label{subfig:Norms_RotMatrix_slow}
	\end{subfigure}	
	\begin{subfigure}[b]{0.3\textwidth}
		\includegraphics[clip=true,trim=0cm 0cm 0cm 0cm,width=\textwidth]
		{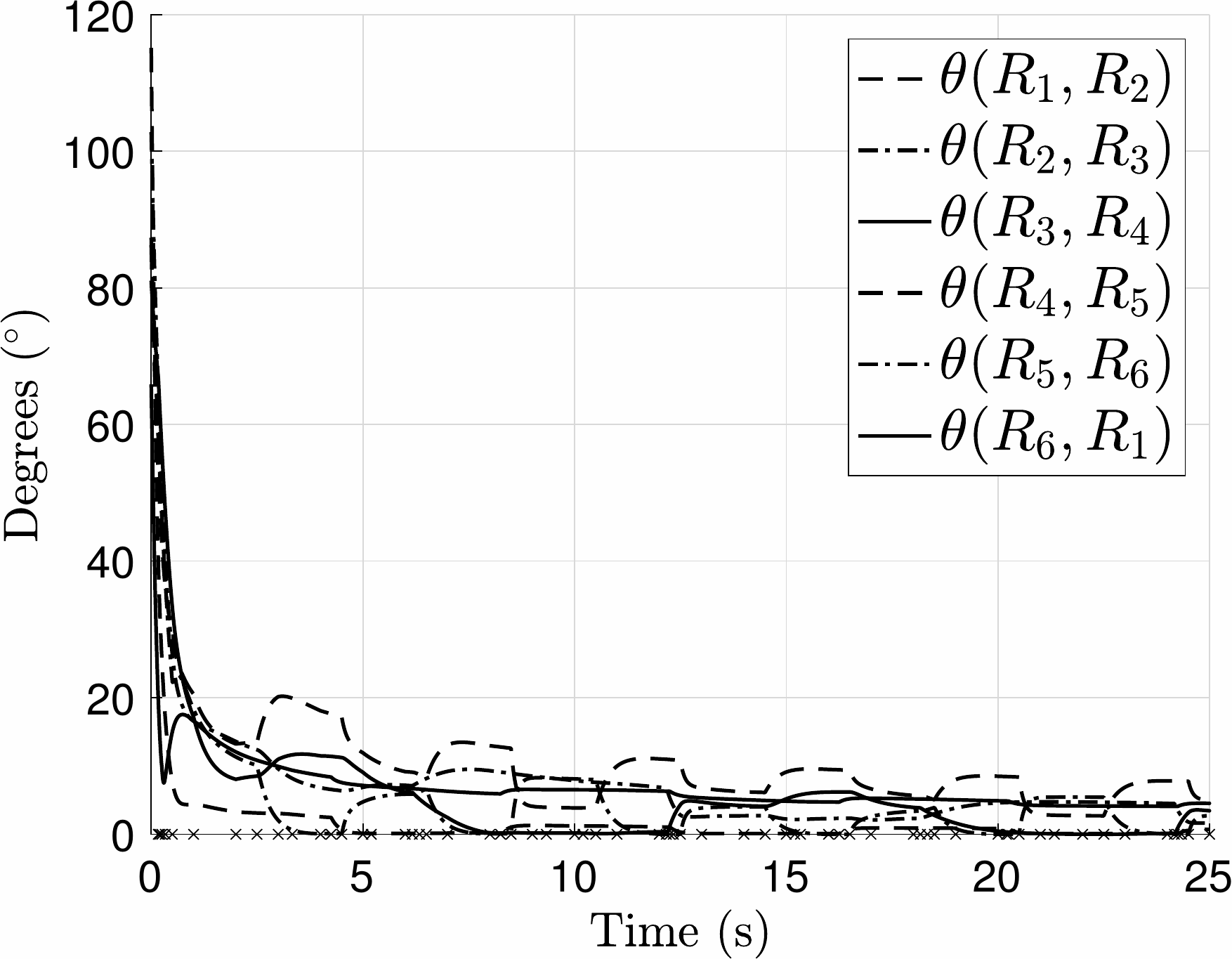}
		\caption{Simulation~$2$: angular distance between some pairs of rotation matrices}
		\label{subfig:Distances_RotMatrix_slow}
	\end{subfigure}	
	\caption{Simulations $1$ and $2$.}
	\label{fig:Simulations}
\end{figure*}

\begin{figure}
			\centering
			\begin{tikzpicture}[-,>=stealth',shorten >=1pt,auto,node distance=3cm,
			                    thick,main node/.style={circle,draw,font=\sffamily\Large\bfseries},scale=0.35,every node/.style={scale=0.45}]
			
			  \node[main node] (1) {$1$};
			  \node[main node] (2) [above right of=1] {$2$};
			  \node[main node] (3) [right of=2] {$3$};
			  \node[main node] (4) [below right of=3] {$4$}; 
			  \node[main node] (6) [below right  of=1] {$6$};			  
			  \node[main node] (5) [right of=6] {$5$};

			  \path[every node/.style={font=\sffamily\small}]
			    (1) edge [->] node {} (2)
			    (1) edge [->,dashed, transform canvas={yshift=1pt}] node {} (4)
			    (2) edge [->] node {} (3)
				(2) edge [->,dashed, transform canvas={xshift=1pt}] node {} (6)			    
			    (3) edge [->] node {} (4)
				(3) edge [->,dashed, transform canvas={xshift=1pt}] node {} (5)			    
			    (4) edge [->] node {} (5)
			    (4) edge [->,dashed, transform canvas={yshift=-1pt}] node {} (1)			    
			    (5) edge [->] node {} (6)
				(5) edge [->,dashed, transform canvas={xshift=-1pt}] node {} (3)			    			    
			    (6) edge [->] node {} (1)   	    
				(6) edge [->,dashed, transform canvas={xshift=- 1pt}] node {} (2)			    			    
			    ;
			\end{tikzpicture}
					
	\caption{Time-varying digraph with 6 agents.}
	\label{fig:Agents6Network}
\end{figure}	
	
	\section{Conclusions}
	In this paper, we studied attitude synchronization in $\Sn{2}$ and in $\SO[3]$, for a group of agents under connected network switching graphs.
We proposed switching output feedback control laws for each agent's angular velocity, which are decentralized and do not require a common orientation frame among agents.
Our main contribution lied in transforming those two problems into a common framework, where all agents dynamics are transformed into unit vectors' dynamics on a sphere of appropriate dimension.
Convergence to a synchronized network was guaranteed for a wide range of initial conditions.
Directions for future work include extending all results to agents controlled at the torque level, rather than the angular velocity level.

	\appendices
	\section{Proof of Theorem~\ref{thm:MainTheorem}}
		\label{app:AppendixProofConvergenceToC}
		\begin{proof}[of Theorem~\ref{thm:MainTheorem}]
	In what that follows, we invoke~\cite[Corollary~4.7]{goebel2008invariance}. 
	We emphasize that the latter corollary requires persistent dwell time signals and Lyapunov functions to be continuously differentiable.
	However, an average dwell time signal is necessarily a persistent dwell time signal~\cite{hespanha2004uniform}.
	On the other hand, the latter corollary can be extended to continuous Lyapunov functions by replacing~\cite[Theorem~1]{goebel2008invariance} with ~\cite[Corollary~4.4 b)]{sanfelice2007invariance} (and making used of the Clarke generalized derivative).
	Finally, \cite{goebel2008invariance} considers multiple Lyapunov functions, one for each switching signal mode, whilst in this proof we restrict ourselves to a common Lyapunov function.
	For brevity, in what follows, we denote 
	\begin{align}
		\xmb \mapsto  f(\xmb)
		:\Leftrightarrow
		(\Rn[n])^{\sss{N}} \ni \xmb := (\xmbi{1},\cdots,\xmbi{N}) 
		\mapsto 
		f(\xmb).
	\end{align}
	
	Consider then any $p \in \mathcal{P}$ and the vector field~\eqref{eq:VectorField}, and denote 
	\begin{align}
		\xmb
		\mapsto 
		v_{\sss{i,p}}(\xmb) := 
		\innerproduct{\xmbi{i}}{\fmb_{\sss{i,p}}(\xmb)}
		\in \mathbb{R},
		\forall i \in \Nset
		\label{eq:velocity}
	\end{align}	
	which are all continuous functions.
	Consider then the continuous function
	\begin{align}
		\xmb
		\mapsto 
		V(\xmb) := \max_{\sss{i\in\Nset}} \frac{1}{2}  \|\xmb_{\sss{i}}\|^{\sss{2}} 
		\in \mathbb{R}_{\sss{0}}^{\sss{+}},
		\label{eq:lyapunov function}
	\end{align}
	whose generalized gradient (in the sense of Clarke) is given by (see~\eqref{eq:SetH} and denote $\text{co}(S)$ as the convex hull of a finite point set $S \subset \Rn[m]$, for any $m \in \mathbb{N}$)
	\begin{align}
		\xmb
		\mapsto 
		\partial V(\xmb) = \text{co}(\{\embi{i} \otimes \xmbi{i}\}_{\sss{i \in \mathcal{H}(\xmb)}})
		\subseteq \Rn[nN],
		\label{eq:LyapunovGradient}
	\end{align}
	and where we emphasize that $\cup_{v \le \frac{1}{2}r^{\sss{2}}} V^{\sss{-1}}(v) = (\bar{\mathcal{B}}(r))^{\sss{N}}$ for any $r \in \Rn[]_{\sss{+}}$ (see Notation for definition of $\mathcal{B}$).
	The generalized directional derivative of $V$ along~\eqref{eq:VectorField}, for a mode $p \in \mathcal{P}$, is then given by
	\begin{align}
		\xmb
		\mapsto
		V_{\sss{p}}^{\sss{\circ}}(\xmb) 
		:=
		& 
		\max_{dV \in \partial V(\xmb) } \innerproduct{dV}{\fmb_{\sss{p}}(\xmb)}
		\\
		{\sss{\eqref{eq:LyapunovGradient}}}
		=
		&
		\max 
		\text{co}(
			\{
				\innerproduct{\xmbi{i}}{\fmb_{\sss{i,p}}(\xmb)} 
			\}_{\sss{i \in \mathcal{H}(\xmb)}}
		).
		\label{eq:generalized derivative}
	\end{align}
	Recall~\eqref{eq:velocity}, and notice that the Theorem's condition \emph{1a)} reads as $v_{\sss{i,p}}(\xmb)|_{\forall i \in H(\xmb)} \le 0$ and condition \emph{1b)} reads as  $ \exists k \in H(\xmb):  v_{\sss{i,p}}(\xmb) |_{i \in H(\xmb)}< 0$, for any $p \in \mathcal{P}$.
	As such, the generalized derivative in~\eqref{eq:generalized derivative} can be expressed equivalently as
	\begin{align}
		\xmb
		\mapsto
		V_{\sss{p}}^{\sss{\circ}}(\xmb) 
		=
		\begin{cases}
			v_{\sss{i,p}}(\xmb)|_{\sss{\text{any } i \in H(\xmb)}} < 0 &  \xmb \in A_{\sss{p}} \\
			\max 
			\text{co}(
				\{
					v_{\sss{i,p}}(\xmb)
				\}_{\sss{i \in \mathcal{H}(\xmb)}}
			)
			\le 0
			&
			\xmb \in B_{\sss{p}}
			\\
			0 & \xmb \in \mathcal{C}
		\end{cases},
		\label{eq:LyapunovGeneralizedDerivative}
	\end{align}	
	where
	\begin{align}
		\hspace{-0.4cm}
		&
		A_{\sss{p}} := \{\xmb \in \Rn[nN] : \forall k,l \in H(\xmb), v_{\sss{k,p}}(\xmb) = v_{\sss{l,p}}(\xmb) <0 \},
		\label{eq:sets C1}
		\\
		\hspace{-0.4cm}
		&
		B_{\sss{p}} := \{\xmb \in \Rn[nN] : \exists k,l \in H(\xmb), v_{\sss{k,p}}(\xmb) \ne v_{\sss{l,p}}(\xmb)\}.
		\label{eq:sets C2}
	\end{align}
	
	The function $V$ in~\eqref{eq:lyapunov function} is lower bounded and its generalized derivative along~\eqref{eq:Solution} is non-positive.
	This implies that $\lim_{t \rightarrow \infty} V(\xmb(t)) =: V^{\sss{\infty}} \in [0,V(\xmb(0))]$;
	and that $\bar{\mathcal{B}}(r_{\sss{0}})^{\sss{N}}$, with  $r_{\sss{0}} := \max_{\sss{i \in \Nset}} \| \xmbi{i}(0)\|$, is positively invariant (since $V(\xmb) \le V(\xmb(0)) \Leftrightarrow \xmb \in \cup_{\sss{r \le V(\xmb(0))}} V^{\sss{-1}}(r) = \bar{\mathcal{B}}(r_{\sss{0}})^{\sss{N}} $).
	
	Moreover, it follows from~\eqref{eq:LyapunovGeneralizedDerivative}, that
	\begin{align}
		(V_{\sss{p}}^{\sss{\circ}})^{\sss{-1}}(0)
		\subseteq
		\mathcal{C}
		\cup
		B_{\sss{p}},
	\end{align}
	and, moreover, (see Proposition~\ref{prop:closure of vdot zero} in this appendix)
	\begin{align}
		\overline{
			(V_{\sss{p}}^{\sss{\circ}})^{\sss{-1}}(0)
		}
		\subseteq
		\mathcal{C}
		\cup
		B_{\sss{p}}.
		\label{eq:closure vdot zero}
	\end{align}		
	
	We now wish to compute the largest invariant subset (in the sense of~\cite[Corollary~4.4]{goebel2008invariance}) of $V^{\sss{-1}}(V^{\sss{\infty}}) \cap \overline{(V_{\sss{p}}^{\sss{\circ}})^{\sss{-1}}(0)}$.
	For that purpose, consider a solution 
	\begin{align}
		\Rn[] \supset I \ni t \mapsto \xmb(t) \in V^{\sss{-1}}(V^{\sss{\infty}}) \cap \overline{(V_{\sss{p}}^{\sss{\circ}})^{\sss{-1}}(0)}.
		\label{eq:invariant solution proof}
	\end{align}
	Composing~\eqref{eq:lyapunov function} with ~\eqref{eq:invariant solution proof} yields a constant function $I \ni t \mapsto  V(\xmb(t)) = V^{\sss{\infty}}$, whose derivative is well defined, namely 
	\begin{align}
		\hspace{-0.2cm}
		I \ni t \mapsto  \dot{V}(\xmb(t)) = v_{\sss{i,p}}(\xmb(t))|_{\sss{\forall i \in H(\xmb(t))}} = 0.
		\label{eq:lyapunov derivative on invariant}
	\end{align}
	In fact, \eqref{eq:lyapunov derivative on invariant} implies that $ v_{\sss{k,p}}(\xmb) =  v_{\sss{l,p}}(\xmb)  $ for all $k,l \in H(\xmb)$, which is not satisfied for any $\xmb \in B_{\sss{p}}$ as defined in~\eqref{eq:sets C2}.
	This implies that $\xmb$ in~\eqref{eq:invariant solution proof} does not belong to $B_{\sss{p}}$.
	It then follows that the largest invariant subset of $V^{\sss{-1}}(V^{\sss{\infty}}) \cap \overline{(V_{\sss{p}}^{\sss{\circ}})^{\sss{-1}}(0)} \subseteq V^{\sss{-1}}(V^{\sss{\infty}}) \cap (\mathcal{C} \cup B_{\sss{p}}) $ is, in fact, a subset of $V^{\sss{-1}}(V^{\sss{\infty}}) \cap \mathcal{C}$, which is independent of the mode $p \in \mathcal{P}$.
	
	In brief, we used a Lyapunov function common to all modes, and verified that the largest invariant set for each mode is independent of the mode. 
	Based on the previous observations, we can invoke~\cite[Corollary 4.7]{goebel2008invariance}, from which it follows that~\eqref{eq:Solution} converges to $V^{\sss{-1}}(V^{\sss{\infty}}) \cap \mathcal{C} \subset \bar{\mathcal{B}}(r_{\sss{0}})^{\sss{N}} \cap \mathcal{C}$.
\end{proof}

\begin{prop}[Closure of the set where $V_{\sss{p}}^{\sss{\circ}}$ vanishes]\label{prop:closure of vdot zero}
	In order to verify that~\eqref{eq:closure vdot zero} holds, consider a convergent sequence in $(V_{\sss{p}}^{\sss{\circ}})^{\sss{-1}}(0)$, namely
	\begin{align}
		&
		\{ \xmb^{\sss{m}} \in (V_{\sss{p}}^{\sss{\circ}})^{\sss{-1}}(0) \}_{\sss{m \in \mathbb{N}}},
		\label{eq:sequence in vdot 0}
		\\
		&
		\xmb^{\sss{\infty}} := \lim_{\sss{m \rightarrow \infty}} \xmb^{\sss{m}} \in A_{\sss{p}} \cup B_{\sss{p}} \cup \mathcal{C}.
	\end{align}
	Let us prove~\eqref{eq:closure vdot zero} by assuming that $\xmb^{\sss{\infty}} \in A_{\sss{p}}$, which will lead to a contraction.
	For that purpose, denote $B_{\sss{\epsilon}} (\xmb^{\sss{\infty}}) := \{ \xmb\in \Rn[nN]: \| \xmb - \xmb^{\sss{\infty}}\| < \epsilon \}$ as an open ball of size $\epsilon >0$ around $\xmb^{\sss{\infty}}$.
	
	\emph{a)} Since the sequence~\eqref{eq:sequence in vdot 0} belongs to $(V_{\sss{p}}^{\sss{\circ}})^{\sss{-1}}(0)$, it follows that for any $m \in \mathbb{N}$, there exist $i \in H(\xmb^{\sss{m}}) $ s.t. $v_{\sss{i,p}}(\xmb^{\sss{m}}) = 0$.
	\emph{b)} Since the sequence~\eqref{eq:sequence in vdot 0} is convergent, it follows that for any $ \epsilon >0$ there exits $M \in \mathbb{N}$ s.t., for all $m \ge M$, $\xmb^{\sss{m}} \in B_{\sss{\epsilon}} (\xmb^{\sss{\infty}})$.
	\emph{c)} By definition of $H$ in~\eqref{eq:SetH}, it follows that for any $\xmb \in \Rn[nN]$ there exists $\epsilon_{\sss{1}} >0 $ s.t., for all $ \ymb \in B_{\sss{\epsilon_{\sss{1}}}}(\xmb)$, $H(\ymb) \subseteq H(\xmb)$.
	\emph{d)} Since $v_{\sss{1,p}}$, \ldots, $v_{\sss{N,p}}$ in~\eqref{eq:velocity} are continuous, it follows that for any $\xmb \in A_{\sss{p}}$ (see~\eqref{eq:sets C1}) there exists $\epsilon_{\sss{2}} \in (0,\epsilon_{\sss{1}})$ s.t., for all $\ymb \in B_{\sss{\epsilon_{\sss{1}}}}(\xmb) $, $v_{\sss{i,p}}(\ymb) < 0 \forall i \in H(\ymb) \overset{\sss{\text{\emph{c)}}}}{\subseteq} H(\xmb)$.
	\emph{e)} Combining \emph{a)} with \emph{d)}, it follows that for $\epsilon = \epsilon_{\sss{2}}$ there exists $ M \in \mathbb{N}$ s.t., for all $m \ge M$, $\xmb^{\sss{m}} \in B_{\sss{\epsilon_{\sss{2}}}} (\xmb^{\sss{\infty}}) \overset{\sss{\text{\emph{d)}}}}{\Rightarrow} v_{\sss{i,p}}(\xmb^{\sss{m}}) < 0 \forall i \in H(\xmb^{\sss{m}}) \subseteq H(\xmb^{\sss{\infty}})$.
	However, \emph{e)} contradicts \emph{a)}, which implies that $\xmb^{\sss{\infty}} \not\in A_{\sss{p}}$.
\end{prop}	

	\section{Proof of Proposition~\ref{thm:MainTheorem}}
		\label{app:AppendixAverageDwellTime}
		A switching signal $\sigma$ has an average dwell-time $\tau_{\sss{D}} > 0$ and a chatter bound $N_{\sss{0}} \in \mathbb{N}$ if the number of switching times of $\sigma$ in any open finite interval $(t_{\sss{1}},t_{\sss{2}}) \subset \Rn[]$ is upper bounded by $\#_{\sss{s}}(t_{\sss{1}},t_{\sss{2}}) = N_{\sss{0}} + \frac{t_{\sss{2}} - t_{\sss{1}}}{\tau_{\sss{D}}}$~\cite{mancilla2006extension}.

\begin{proof}[of Proposition~\ref{thm:MainTheorem}]
	For brevity, denote $\tau_{\sss{D}}^{\sss{\min}} = \min_{\sss{i \in \Nset}} \tau_{\sss{D}}^{\sss{i}}$, $N_{\sss{0}}^{\sss{\max}} = \max_{\sss{i \in \Nset}} N_{\sss{0}}^{\sss{i}}$, $\mathcal{T}$ as the set of switching time instants of $\sigma$ and $\mathcal{T}^{\sss{i}}$ as the set of switching time instants of $\Nset_{\sss{i}} \circ \sigma $ for each $i \in \Nset$.
	%
	%
	If $N_{\sss{0}} = N N_{\sss{0}}^{\sss{\max}}$ and $\tau_{\sss{D}} = \frac{1}{N}\tau_{\sss{D}}^{\sss{\min}}$, then $\#_{\sss{s}}(t_{\sss{1}},t_{\sss{2}}) = N \left(N_{\sss{0}}^{\sss{\max}} + \frac{t_{\sss{2}} - t_{\sss{1}}}{\tau_{\sss{D}}^{\sss{\min}}}\right)$.
	Next, we show that this function indeed upper bounds the number of switches of $\sigma$.
	Consider then any open interval $(t_{\sss{1}},t_{\sss{2}}) \subset \Rpositive$ where
	\begin{align}
		0 < t_{\sss{2}} - t_{\sss{1}} 
		\le 
		\frac{1}{N} \tau_{\sss{D}}^{\sss{\min}}
		\le
		\tau_{\sss{D}}^{\sss{i}}
		,
		\forall i \in \Nset .
		\label{eq:ConditionTimeInterval}
	\end{align} 	
	Consider a switch instant $t$ in that interval, i.e $t \in \mathcal{T}$ and $ t \in (t_{\sss{1}},t_{\sss{2}})$. 
	It follows that  $t \in \mathcal{T}^{\sss{j}}$ for some $j \in \Nset$.
	Under Assumption~\ref{assump:SwitchesAgents}, the next $N_{\sss{0}}^{\sss{j}}$ switches from $\mathcal{T}^{\sss{j}}$ may come only after a period $\tau_{\sss{D}}^{\sss{j}} \ge \tau_{\sss{D}}^{\sss{\min}} \ge \frac{1}{N} \tau_{\sss{D}}^{\sss{\min}}$;
	and, since~\eqref{eq:ConditionTimeInterval} holds, it then follows that agent $j$ can only switch $N_{\sss{0}}^{\sss{j}}$ times in $(t_{\sss{1}},t_{\sss{2}})$.
	Since this is valid for any $j \in \Nset$, it follows that a maximum of $N N_{\sss{0}}^{\sss{\max}}$ switches are possible in $(t_{\sss{1}},t_{\sss{2}})$.
	It also follows from~\eqref{eq:ConditionTimeInterval} that $N N_{\sss{0}}^{\sss{\max}} \le \#_{\sss{s}}(t_{\sss{1}},t_{\sss{2}})$, and therefore $\#_{\sss{s}}(t_{\sss{1}},t_{\sss{2}})$ upper bounds the number of switches in the interval $(t_{\sss{1}},t_{\sss{2}})$.
	
	Consider now any open interval  $(t_{\sss{1}},t_{\sss{2}}) \subset \Rpositive$ where
	\begin{align}
		\frac{k-1}{N} \tau_{\sss{D}}^{\sss{\min}} < t_{\sss{2}} - t_{\sss{1}} \le \frac{k}{N} \tau_{\sss{D}}^{\sss{\min}},
		\label{eq:ConditionTimeInterval2}
	\end{align} 
	for some $k \in \{2,3,\cdots\}$.
	The interval $(t_{\sss{1}},t_{\sss{2}})$ may be broken in $k$ intervals of equal length, i.e., given $\Delta t = \frac{t_{\sss{2}} - t_{\sss{1}}}{k}$, $(t_{\sss{1}},t_{\sss{2}}) \backslash \cup_{\sss{l \in \{2,\cdots,k\}}} \{t_{\sss{1}} + (l'-1) \Delta t\} = \cup_{\sss{l \in \{1,\cdots, k \}}} (t_{\sss{1}} + (l-1)\Delta t, t_{\sss{1}} + l\Delta t )$.
	For any $l \in \{1,\cdots,k\}$, the interval $(t_{\sss{1}} + (l-1)\Delta t,t_{\sss{1}} + l\Delta t)$ has length $\Delta t \le \frac{1}{N} \tau_{\sss{D}}^{\sss{\min}}$ where the inequality follows since~\eqref{eq:ConditionTimeInterval2} holds.
	Thus, we invoke the same reasoning as before to conclude that in each of those intervals only a maximum of $N N_{\sss{0}}^{\sss{\max}}$ switches can occur, and, as such, only a total of $k N N_{\sss{0}}^{\sss{\max}}$ switches can occur in $(t_{\sss{1}},t_{\sss{2}})$.
	It also follows from~\eqref{eq:ConditionTimeInterval2} that $k N N_{\sss{0}}^{\sss{\max}} \le \#_{\sss{s}}(t_{\sss{1}},t_{\sss{2}})$, and therefore $\#_{\sss{s}}(t_{\sss{1}},t_{\sss{2}})$ upper bounds the number of switches in the interval $(t_{\sss{1}},t_{\sss{2}})$.
\end{proof}
\if\paperextended1
\begin{figure}
	\centering
	\includegraphics[clip=true,trim=0cm 0cm 0cm 0cm,width=0.4\textwidth]
	{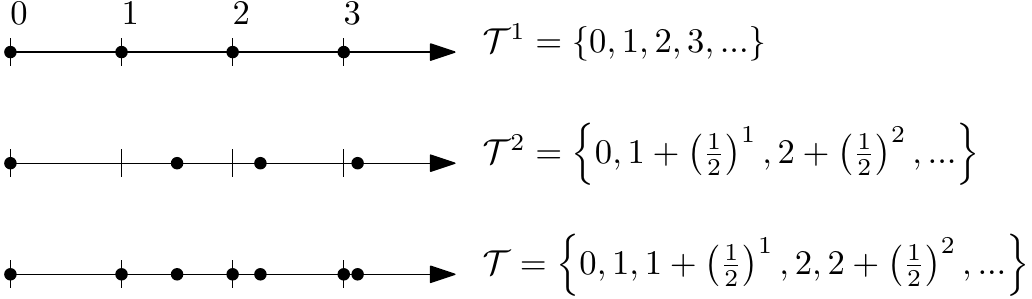}	
	\caption{Network with 2 agents, where neighbor set of agent 1 switches at time instants $\mathcal{T}^{\sss{1}}$, with dwell time $\tau_{\sss{D}}^{\sss{1}} = 1$ (and chattering bound $N_{\sss{0}}^{\sss{1}} = 1$); where neighbor set of agent 2 switches at time instants $\mathcal{T}^{\sss{2}}$, with dwell time $\tau_{\sss{D}}^{\sss{1}} = 0.75$ (and chattering bound $N_{\sss{0}}^{\sss{2}} = 1$); and the network dynamics switches at time instants $\mathcal{T}$ with no dwell time, but with average dwell time $\tau_{\sss{D}} \ge  \frac{1}{N} \min\{1,0.75\} = \frac{0.75}{2}$ and chattering bound $N_{\sss{0}} = N \max\{1,1\} =2$.}
	\label{fig:NoDwellTime}
\end{figure}
In Fig.~\ref{fig:NoDwellTime}, we illustrate a scenario where the neighbors' set of two agents switch with different dwell times, while the switching signal of the complete network switches with average dwell time.
Given the sets $\mathcal{T}^{\sss{i}}$ of switching time instants of $\Nset_{\sss{i}}(\sigma(\cdot))$ for each $i \in \Nset$, the set of switching time instants for the complete network (as done in Fig.~\ref{fig:NoDwellTime}) may be computed recursively following Algorithm~\ref{alg:ComputeT}, where we assume $t_{\sss{0}}^{\sss{i}} = 0$ for all $i \in \Nset$.
\begin{algorithm}\label{alg:ComputeT}
 \textbf{Given:} $\mathcal{T}^{\sss{i}} = \{t_{\sss{k}}^{\sss{i}}\}_{\sss{k \in \mathbb{N}}}$ for each $i \in \Nset$ \\
 \textbf{Initialization:} $\mathcal{T} = \emptyset$ \\
 $t_{\sss{k}} \leftarrow 0$\\
 $\mathcal{T} \leftarrow \mathcal{T} \cup \{t_{\sss{k}}\}$ \\
 $\mathcal{T}^{\sss{i}} \leftarrow \mathcal{T}^{\sss{i}}\backslash\{t_{\sss{k}}\}$\\
 \While{$\mathcal{T}^{\sss{i}} \ne \emptyset$ for some $i \in \Nset$ }{
  $V = \arg_{\sss{(j,t_{\sss{k}}^{\sss{j}})}} \min_{\sss{j \in \Nset}} \min_{\sss{t_{\sss{k}}^{\sss{j}} \in \mathcal{T}^{\sss{j}}}} (t_{\sss{k}}^{\sss{j}} - t_{\sss{k}})$\\
  $t_{\sss{k}} \leftarrow t_{\sss{k}}^{\s{j}}$ for any $(j,t_{\sss{k}}^{\sss{j}}) \in V$ \\
   $\mathcal{T} \leftarrow \mathcal{T} \cup \{t_{\sss{k}}\}$\\ 
  $\mathcal{T}^{\sss{j}} \leftarrow \mathcal{T}^{\sss{j}}\backslash\{t_{\sss{k}}\}$ for all $(j,t_{\sss{k}}^{\sss{j}}) \in V$
 }
 \caption{Constructing the set of switching time instants $\mathcal{T}$ for the complete network, given the sets of switching time instants of $\Nset_{\sss{i}}(\sigma(\cdot))$ of each agent $i \in \Nset$. (Remark: at each step, the operation $\min_{\sss{t_{\sss{k}}^{\sss{j}} \in \mathcal{T}^{\sss{j}}}} (t_{\sss{k}}^{\sss{j}} - t_{\sss{k}})$ is well defined since $\mathcal{T}^{\sss{j}}$ is an increasing sequence and its first element is always greater than $t_{\sss{k}}$.)}
\end{algorithm}
\fi	
	
	\section{Auxiliary results}
		\label{app:AuxiliaryResults}
		\begin{prop}
	\label{prop:ConeGeneration}
	Let $\numb\in \mathcal{S}^{\sss{n}}$ and $\alpha \in [0,\frac{\pi}{2})$.
	There exist $n+1$ linearly independent unit vectors $\numbi{1}, \cdots, \numbi{n+1} \in \mathcal{S}^{\sss{n}}$ such that $\ConeOpen(\alpha,\numb) \subset \ConeOpen(\alpha + \delta,\numbi{i}) \subset \ConeOpen(\frac{\pi}{2},\numbi{i})$ for all $i \in \{1,\cdots,n+1\}$ and for some $\delta \in (0,\frac{\pi}{2} - \alpha)$.
	%
\end{prop}

\begin{proof}
	It is trivial to verify that $\ConeOpen(\alpha,\numb) \subset \ConeOpen(\alpha + \delta,\numbi{i})$ for any $\delta \in (0,\frac{\pi}{2} - \alpha)$.
	
	Recall that, by definition,  $\innerproduct{\numb}{\nmb} > \cos(\alpha) > 0$ for all $\nmb \in \ConeOpen(\alpha,\numb)$.
	Consider now the unit vector $\bm{\eta}_{\sss{i}} = \cos(\delta)\numb + \sin(\delta) \numb^{\sss{\perp}}_{\sss{i}} \in \Sn{n}$, for some $\delta \in (0,\frac{\pi}{2} - \alpha)$ and  where $\numb^{\sss{\perp}}_{\sss{i}} \in \mathcal{S}^{\sss{n}}$ is an unit vector orthogonal to $\numb$. 
	Since $\delta \in (0,\frac{\pi}{2} - \alpha)$, it follows that $\cos(\delta)>0$.
	Then $\innerproduct{\nmb}{\bm{\eta}_{\sss{i}}} = \cos(\delta)\innerproduct{\nmb}{\numb} + \sin(\delta)  \innerproduct{\nmb}{\numb^{\sss{\perp}}_{\sss{i}}} > \cos(\delta)\cos(\alpha) - \sin(\delta) \sin(\alpha) = \cos(\alpha + \delta) > \cos(\alpha)$, for all $\nmb\in \ConeOpen(\alpha,\numb)$. 
	Since there are $n$ unit vectors $\numb^{\sss{\perp}}_{\sss{1}}, \cdots, \numb^{\sss{\perp}}_{\sss{n}}$ orthogonal to $\numb$, it follows that we can find $n$ linearly independent unit vectors $\bm{\eta}_{\sss{1}}, \cdots,  \bm{\eta}_{\sss{n}}$ such that $\ConeOpen(\alpha,\numb) \subset \ConeOpen(\alpha + \delta,\bm{\eta}_{\sss{i}})$ for all $i \in \{1,\cdots,n\}$.
	Moreover, $\{\numb,\bm{\eta}_{\sss{1}}, \cdots,  \bm{\eta}_{\sss{n}}\}$ are $n+1$ linearly independent unit vectors.
\end{proof}
Figure~\ref{fig:MultipleCones} illustrates the result in Proposition~\ref{prop:ConeGeneration} for $n = 2$.
From Proposition~\ref{prop:ConeGeneration} it follows that if a group of unit vectors (in $\Rn[n+1]$) is contained in a closed $\alpha$-cone for  some $\alpha \in [0,\frac{\pi}{2})$, then we can find $n+1$ larger (by $\delta$) closed cones that contain the same group of unit vectors;
i.e.,  given $\numb = (\numbi{1},\cdots,\numbi{N}) \in \ConeClosed(\alpha,\bar{\numb})^{\sss{N}}$ for some $\bar{\numb} \in \Sn{n}$, there exist  $n+1$ linearly independent unit vectors, $\{\bar{\numb}_{\sss{1}}, \cdots,  \bar{\numb}_{\sss{n+1}}\}$,  such that $\numb \in \ConeClosed(\alpha +\delta , \bar{\numb}_{\sss{i}})^{\sss{N}}$ for all $i \in \{1,\cdots,n+1\}$ and for some $\delta \in (0,\frac{\pi}{2} - \alpha)$.

\begin{figure}
	\centering
	\includegraphics[clip=true,trim=0cm 0cm 0cm 0cm,width=0.45\textwidth]
	{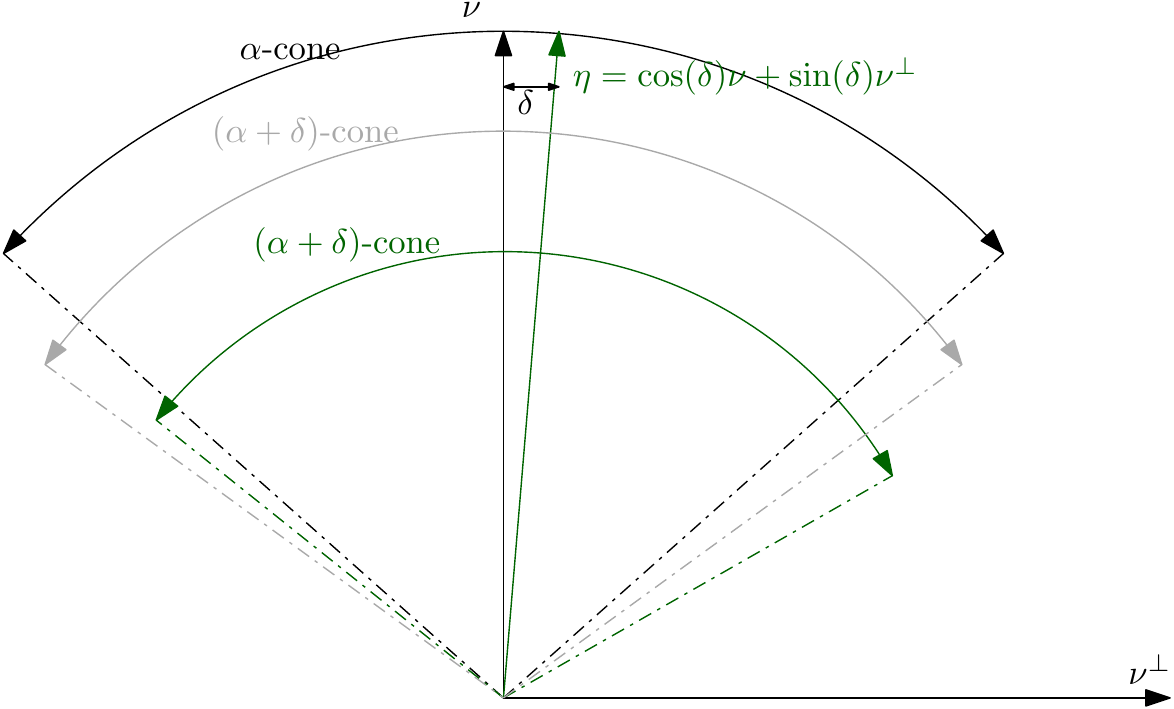}
	\caption{Illustration of result in Proposition~\ref{prop:ConeGeneration} for $n = 2$: closed $(\alpha+\delta)$-cones formed by $\numb$ and $\bm{\eta}$ contain closed $\alpha$-cone formed by $\numb$; also, $\numb$ and $\bm{\eta}$ are linearly independent.}
	\label{fig:MultipleCones}
\end{figure}

\begin{proof}[of Proposition~\ref{prop:UnitVectorsInequality}]
	\emph{Sufficiency}:
	Regarding $\emph{(a)}$, if $\numbi{2} = \numbi{1}$, then $\innerproduct{\numb}{\OP{\numbi{1}}\numbi{2}}  = \innerproduct{\bar{\numb}}{\OP{\numbi{1}}\numbi{1}}= 0 $.
	Regarding $\emph{(b)}$, if $\numbi{1} \ne \numbi{2}$ then $\innerproduct{\numbi{1}}{\numbi{2}} <1$.
	Consider then the following two cases: \emph{(1)} $\innerproduct{\numb}{\numbi{1}} = \innerproduct{\numb}{\numbi{2}}$ and \emph{(2)} $\innerproduct{\numb}{\numbi{1}} \ne \innerproduct{\numb}{\numbi{2}}$.
	For case \emph{(1)}, it follows that $\innerproduct{\numb}{\OP{\numbi{1}}\numbi{2}} = \innerproduct{\numb}{\numbi{2}} - \innerproduct{\numb}{\numbi{1}} \numbi{1}\tp\numbi{2} = \innerproduct{\numb}{\numbi{1}} (1 -\numbi{1}\tp\numbi{2}) > 0$, where the inequality applies since $\numbi{1}\tp\numbi{2}  < 1$ and since, by assumption $0<\innerproduct{\numb}{\numbi{1}}$.
	For case \emph{(2)},  the Proposition's assumption becomes 
	$0 < \innerproduct{\numb}{\numbi{1}} < \innerproduct{\numb}{\numbi{2}}$.
	Then, since $\innerproduct{\numbi{1}}{\numbi{2}} <1$, it follows that  $ 0 < \innerproduct{\numb}{\numbi{1}} < \innerproduct{\numb}{\numbi{1}} \innerproduct{\numbi{1}}{\numbi{2}}  \Rightarrow \innerproduct{\numb}{\OP{\numbi{1}}\numbi{2}} > 0 $. 
	%
	%
	\emph{Necessity}: Regarding $\emph{(a)}$, if we assume on the contrary that $\numbi{1} \ne \numbi{2}$, then it follows from before that $\innerproduct{\numb}{\OP{\numbi{1}}\numbi{2}}  =  \innerproduct{\numb}{\numbi{2}} - \innerproduct{\numb}{\numbi{1}} \numbi{1}\tp\numbi{2} > 0$, which implies that $\numbi{1} = \numbi{2}$.
	Similarly, regarding $\emph{(b)}$, if we assume on the contrary that $\numbi{1} = \numbi{2}$, then it follows from before that $\innerproduct{\numb}{\OP{\numbi{1}}\numbi{2}} = \innerproduct{\bar{\numb}}{\OP{\numbi{1}}\numbi{1}}  = 0$, which implies that $\numbi{1} \ne \numbi{2}$.
\end{proof}
	
	\section{Quaternions}
		\label{app:QuaternionsAppendix}
		This section provides some auxiliary results that are useful in Section~\ref{sec:SynchronizationInSO3}.
Recall the map~\eqref{eq:map rotation to quaternion}.
Given $\Rmat \in \tilde{\mathbb{SO}}(3)$, consider then $\qmb = \mrtoq(\Rmat)$, and for $i \in \{1,2,3,4\}$, denote the $i^{\text{th}}$ component of the quaternion as $\qmb_{\sss{i}} = \mrtoqi(\Rmat)$ (denote $d\mrtoqi$ as the derivative of the function $\mrtoqi$).
Its kinematics are given by
\begin{align}
	\dot{\qmb}_{\sss{i}}(\qmb,\bm{\omega})
	\equiv
	&
	\fmb_{\sss{qi}}(\qmb,\bm{\omega})
	\\
	:=
	& 
	\text{tr}(d\mrtoqi(\Rmat) \fmb_{\sss{\Rmat}}(\Rmat,\bm{\omega})\tp)
	|_{\Rmat = \mqtor(\qmb)}
	\\
	{\sss{\eqref{eq:RotationMatrixKinematics}}}
	=
	&
	- \text{tr}(d\mrtoqi(\Rmat) \sk{\bm{\omega}} \Rmat\tp)
	|_{\Rmat = \mqtor(\qmb)}
	\\
	=
	&
	- \text{tr}(\Rmat\tp d\mrtoqi(\Rmat) \sk{\bm{\omega}} )
	|_{\Rmat = \mqtor(\qmb)}
	\\
	=
	&
	\innerproduct{
		\invsk{X - X\tp}
	}
	{\bm{\omega}}
	|_{X = \Rmat\tp d\mrtoqi(\Rmat)}
	|_{\Rmat = \mqtor(\qmb)}
\end{align}
where we have made use of the facts that $\text{tr}(AB)=\text{tr}(BA)$ and $\text{tr}(A\sk{\amb})= \innerproduct{\amb}{\invsk{A - A\tp}}$ for any $A,B \in \Rn[3\times 3]$ and $\amb\in \Rn[3]$.
Denote $\tilde{\mathbb{S}}^{\sss{3}}$ as the image of the map $\mrtoqi$.
Collectively, the kinematics of the quaternion
\begin{align}
	\dot{\qmb}
	\equiv
	\fmb_{\sss{q}}:
	\tilde{\mathbb{S}}^{\sss{3}} \times \Rn[3]
	\ni
	(\qmb,\bm{\omega})
	\mapsto
	\fmb_{\sss{q}}(\qmb,\bm{\omega})
	\in 
	T_{\sss{\qmb}} \tilde{\mathbb{S}}^{\sss{3}} 
\end{align}
are given by
\begin{align}
	\dot{\qmb}(\qmb,\bm{\omega})
	\equiv
	&
	\fmb_{\sss{q}}(\qmb,\bm{\omega})
	:=
	(\fmb_{\sss{q1}}(\qmb,\bm{\omega}),\cdots,\fmb_{\sss{q4}}(\qmb,\bm{\omega}))
	\\
	=
	&
	\frac{1}{2} Q(\qmb) 
	[\zvec_{\sss{3\times 1}} \, \Idmat_{\sss{3}}]\tp \bm{\omega}
\end{align}
with $Q$ as defined in~\eqref{eq:q matrix quaternion}, and which can be extended from $\tilde{\mathbb{S}}^{\sss{3}} $ to $\Sn{3}$.

It can be verified that, given any $\qmb_{\sss{1}},\qmb_{\sss{2}} \in \Sn{3}$, and $\Rmat_{\sss{1}} = \mqtor(\qmb_{\sss{1}})$ and $\Rmat_{\sss{2}} = \mqtor(\qmb_{\sss{2}})$, it follows that
\begin{align}
	\Rmat_{\sss{1}}\tp\Rmat_{\sss{2}}
	=
	\mqtor(Q(\qmb_{\sss{1}})\tp\qmb_{\sss{2}}).
	\label{eq:product quaternion}
\end{align}
It can also be easily verified that, given $\qmb\in \Sn{3}$, and $\Rmat= \mqtor(\qmb)$, it follows that
\begin{align}
	&
	\theta(\Idmat_{\sss{3}},\Rmat)|_{\Rmat = \mqtor(\qmb) } = \arccos(2 \innerproduct{\qmb}{\emb_{\sss{1}}}^{\sss{2}} - 1),
	\label{eq:theta quaternion}
	\\
	&
	\invsk{\Rmat - \Rmat\tp}|_{\Rmat = \mqtor(\qmb) } 
	=
	2
	\innerproduct{\qmb}{\emb_{\sss{1}}}
	[\zvec_{\sss{3\times 1}} \, \Idmat] 
	Q(\qmb) \emb_{\sss{1}}.
	\label{eq:invsk quaternion}
\end{align}
Combining the latter with~\eqref{eq:product quaternion}, it follows that
\begin{align}
	&
	\theta(
		\Idmat_{\sss{3}},
		\Rmat_{\sss{1}}\tp\Rmat_{\sss{2}}
	)
	|_{
		\Rmat_{\sss{1}}\tp\Rmat_{\sss{2}} 
		= 
		\mqtor(Q(\qmb_{\sss{1}})\tp\qmb_{\sss{2}}) 
	}
	=
	\\
	& 
	= 
	\arccos(
	2 \innerproduct{\emb_{\sss{1}}}{Q(\qmb_{\sss{1}})\tp\qmb_{\sss{2}}}^{\sss{2}} - 1
	)
	\\
	& 
	= 
	\arccos(
		2 \innerproduct{\qmb_{\sss{1}}}{\qmb_{\sss{2}}}^{\sss{2}} - 1
	),
	\label{eq:theta quaternion product}
\end{align}
and that (note, from~\eqref{eq:q matrix quaternion}, that $Q(\qmb)\tp \emb_{\sss{1}} = \qmb$)
\begin{align}
	&
	\invsk{
		\Rmat_{\sss{1}}\tp\Rmat_{\sss{2}} - (\Rmat_{\sss{1}}\tp\Rmat_{\sss{2}})\tp
	}
	|_{
		\Rmat_{\sss{1}}\tp\Rmat_{\sss{2}} 
		= 
		\mqtor(Q(\qmb_{\sss{1}})\tp\qmb_{\sss{2}}) 
	} 
	=
	\\
	&
	=
	2
	\innerproduct{\emb_{\sss{1}}}{Q(\qmb_{\sss{1}})\tp\qmb_{\sss{2}}} 
	[\zvec_{\sss{3\times 1}} \, \Idmat]
	Q(Q(\qmb_{\sss{1}})\tp\qmb_{\sss{2}})\tp \emb_{\sss{1}}.
	\\
	&
	=
	2\innerproduct{\qmb_{\sss{1}}}{\qmb_{\sss{2}}} 
	[\zvec_{\sss{3\times 1}} \, \Idmat]
	Q(\qmb_{\sss{1}})\tp \qmb_{\sss{2}}.
	\label{eq:invsk quaternion product}
\end{align}
	
	\section{Consensus in $\Rn[n]$ casted as synchronization in $\Sn{n}$}
		\label{app:ConsensusRn}
		In this section, we consider a group of $N$ agents operating in $\Rn[n]$, for some $n \in \mathbb{N}$.
For each $i \in \Nset$, $\umb_{\sss{i}}: \Rn[]_{\sss{\ge 0}} \mapsto \Rn[n]$ is the body-framed linear velocity of agent~$i$, which can be actuated.
Consider then $\xmb = (\xmbi{1},\cdots,\xmbi{N}) : \Rn[]_{\sss{\ge 0}} \mapsto (\Rn[n])^{\sss{N}}$, where each position $\xmbi{i}: \Rn[]_{\sss{\ge 0}} \ni t \mapsto \xmbi{i}(t) \in \Rn[n]$ evolves according to
\begin{align}
	\dot{\xmb}_{\sss{i}}(t) = \Rmat_{\sss{i}} \umb_{\sss{i}}(t) =: \fmb_{\sss{i}}(\umb_{\sss{i}}(t)), 
	\label{eq:PositionKinematics}
\end{align}
where $\fmb_{\sss{i}} : \Rn[n] \ni \umb\mapsto \fmb_{\sss{i}}(\umb) := \Rmati{i} \umb \in \Rn[n]$ and where $\Rmati{i} \in \SO[3]$ represents the body orientation frame of agent $i$ w.r.t. some unknown inertial orientation frame.
In physical terms, $\Rmat_{\sss{i}} \umb_{\sss{i}}$  corresponds to the inertial-framed linear velocity of agent $i$, and we assume that agent $i$ is unaware of its orientation w.r.t. the inertial orientation frame.

If, at a time instant $t \in \Rn[]_{\sss{\ge 0}}$, agent $i \in \Nset$ is aware of the relative position between itself and another agent $j$, then $j \in \Nset_{\sss{i}}(\sigma(t))$, where $\sigma(t)$ encodes the network graph at time $t$.
At each time instant $t \in \Rpositive$ and for $\xmb \in  (\Rn[n])^{\sss{N}}$, each agent $i \in \Nset$ measures $\hmb_{\sss{i}}(t,\xmb)$, where $\hmb_{\sss{i}}(t,\cdot): (\Rn[n])^{\sss{N}} \ni \xmb \mapsto \hmb_{\sss{i}}(t,\xmb) \in (\Rn[n])^{\sss{|\Nset_{\sss{i}}(\sigma(t))|}}$ is defined as
\begin{align}
	\Scale[0.9]{
		\hmb_{\sss{i}}(t,\xmb)
		:= 
		(\hmb_{\sss{ij_{\sss{1}}}}(\xmb), \cdots, \hmb_{\sss{ij_{\sss{| \Nset_{\sss{i}}(\sigma(t))|}}}}(\xmb)) 
		\in 
		(\Rn[n])^{\sss{|\Nset_{\sss{i}}(\sigma(t))|}}	
		,
	}
	\label{eq:RnAgentMeasurement}
\end{align}
where $\{ j_{\sss{1}}, \cdots, j_{\sss{| \Nset_{\sss{i}}(\sigma(t))|}} \} \equiv \Nset_{\sss{i}}(\sigma(t))$ and where $\hmb_{\sss{ij}}: (\Rn[n])^{\sss{N}} \ni \xmb :=(\xmbi{1},\cdots,\xmbi{N}) \mapsto \hmb_{\sss{ij}}(\xmb) \in \Rn[n]$ is defined as
\begin{align}
	\hmb_{\sss{ij}}(\xmb) 
	:= 
	\Rmati{i}\tp(\xmbi{j} - \xmbi{i}) \in \Rn[n],
	\label{eq:RnAgentMeasurement2}
\end{align}
for each $j \in \Nset_{\sss{i}}(\sigma(t))$.
Thus, at each time instant $t \in \Rpositive$, an agent $i \in \Nset$ makes $|\Nset_{\sss{i}}(\sigma(t))|$ measurements, and each measurement corresponds to a distance vector between agent~$i$ and one of its neighbors, projected onto agent's $i$ orientation frame.
Notice that~\eqref{eq:RnAgentMeasurement2} does not require a common reference frame among agents, i.e., agents do not need to agree on a common origin and orientation frame.

\begin{prob}
	\label{prob:ProblemRnDynamics}
	For each time instant $t \in \Rpositive$ and for each $i \in \Nset$, design time-varying decentralized feedback laws $\umb_{\sss{i}}^{\sss{h}}(t,\cdot): (\Rn[n])^{\sss{| \Nset_{\sss{i}}(\sigma(t))|}} \mapsto \Rn[n] $, such that asymptotic consensus of $\xmb = (\xmbi{1},\cdots,\xmbi{N}) : \Rn[]_{\sss{\ge 0}} \mapsto (\Rn[n])^{\sss{N}}$ is accomplished, where $\dot{\xmb}_{\sss{i}}(t) = \fmb_{\sss{i}}(\umb_{\sss{i}}^{\sss{h}}(\hmb_{\sss{i}}(t,\xmb(t))))$ for every $i \in \Nset$ and with $\hmb_{\sss{i}}(t,\cdot)$ as defined in~\eqref{eq:RnAgentMeasurement}.
\end{prob}
For each $t \in \Rpositive$ and each agent $i \in \Nset$, consider the control law  $\umb_{\sss{i}}^{\sss{h}}(t,\cdot) : \Rn[n| \Nset_{\sss{i}}(\sigma(t))|] \ni \hmb_{\sss{i}} := (\hmb_{\sss{ij_{\sss{1}}}}, \cdots, \hmb_{\sss{ij_{\sss{| \Nset_{\sss{i}}(\sigma(t))|}}}}) \mapsto \umb_{\sss{i}}^{\sss{h}}(t,\hmb_{\sss{i}}) \in \Rn[n]$ defined as
\begin{align}
	\umb_{\sss{i}}^{\sss{h}}(t,\hmb_{\sss{i}})
	:=
	-
	\sum_{j \in \Nset_{\sss{i}}(\sigma(t))} 
	w_{\sss{ij}}(\|\hmb_{\sss{ij}}\|)
	\hmb_{\sss{ij}},
	\label{eq:DistributedControlLawPositionConsensusRn}
\end{align}
where $w_{\sss{ij}}: \Rn[]_{\sss{\ge 0}} \mapsto \Rn[]_{\sss{\ge 0}}$ is a weight function agent $i$ assigns to the position error between itself and its neighbor $j$.
This weight may be used, for example, to bound the actuation: indeed, if $\Rn[]_{\sss{\ge 0}} \ni x \mapsto w_{\sss{ij}}(x) := \frac{\sigma}{\sqrt{\sigma^2 + x^2}}$ for all $j \in \Nset$ and some $\sigma > 0$, then $\|\umb_{\sss{i}}^{\sss{h}}(\cdot,\cdot)\| \le N \sigma$ (since $\sup_{\sss{x \ge 0}} w_{\sss{ij}}(x) x = \sigma$).
Denote $\umb_{\sss{i}}^{\sss{cl}}$ as the composition of the control law~\eqref{eq:DistributedControlLawPositionConsensusRn} with the measurement function~\eqref{eq:RnAgentMeasurement}, i.e.,  $\umb_{\sss{i}}^{\sss{cl}} : \Rpositive \times (\Rn[n])^{\sss{N}} \ni (t,\xmb) \mapsto \umb_{\sss{i}}^{\sss{cl}}(t, \xmb) := \umb_{\sss{i}}^{\sss{h}}(t,\hmb_{\sss{i}}(t,\xmb)) \in \Rn[n]$.
It follows that
\begin{align}
	\umb_{\sss{i}}^{\sss{cl}}(t, \xmb)
	= &
	-
	\Rmati{i}\tp
	\sum\limits_{j \in \Nset_{\sss{i}}(\sigma(t))} 
	w_{\sss{ij}}(\|\xmbi{i} - \xmbi{j}\|)
	(\xmbi{i} - \xmbi{j}).
	\label{eq:DistributedControlLawPositionConsensusRnState}
\end{align}
By composing~\eqref{eq:PositionKinematics} with~\eqref{eq:DistributedControlLawPositionConsensusRnState}, it follows that
\begin{align}
	\xmbiDot{i}(t)
	& = 
	\sum_{\sss{j \in \Nset_{\sss{i}}(\sigma(t))}}
	w_{\sss{ij}}(\|\xmbi{j}(t) -\xmbi{i}(t)\|)
	(\xmbi{j}(t) -\xmbi{i}(t))
	\label{eq:xComponentsDynamics}
\end{align}
which is not in the form~\eqref{eq:GeneralForm2}.
In order to write the closed loop dynamics~\eqref{eq:xComponentsDynamics} as in~\eqref{eq:GeneralForm2}, we perform a transformation which is discussed next.

In order to analyze consensus in $\Rn[n]$ under the same framework as synchronization in $\Sn{2}$ and $\SO[3]$, we now perform a change of variables that serves only the purpose of analysis.
Consider then the unit vector $\embi{n +1} =(\zvec_{\sss{n}}, 1) \in \Sn{n} \in \Rn[n+1]$ and the matrix $P = [\Idmat_{\sss{n}} \, \zvec_{\sss{n}}]\tp \in \Rn[(n+1) \times n]$.
Consider also the mapping $\hmb : \Rn[n] \ni \xmb \mapsto \hmb(\xmb) \in \ConeOpen(\frac{\pi}{2},\embi{n +1}) \subset \Sn{n} \subset \Rn[n+1]$, defined as
\begin{align}
	\hmb(\xmb) := \frac{P \xmb + \embi{n +1}}{\|P \xmb + \embi{n +1}\|}, 
	\label{eq:UnitVectorPositionConsensus}
\end{align}
where $\|P \xmb + \embi{n+1}\| = \sqrt{1 +  \|\xmb\|^{\sss{2}}} \ge 1 > 0$ for all $\xmb \in \Rn[n]$.
This transformation is illustrated in Fig.\ref{fig:ConsenusRn}.
\begin{figure}
	\centering
	\includegraphics[clip=true,trim=0cm 0cm 0cm 0cm,width=0.4\textwidth]
	{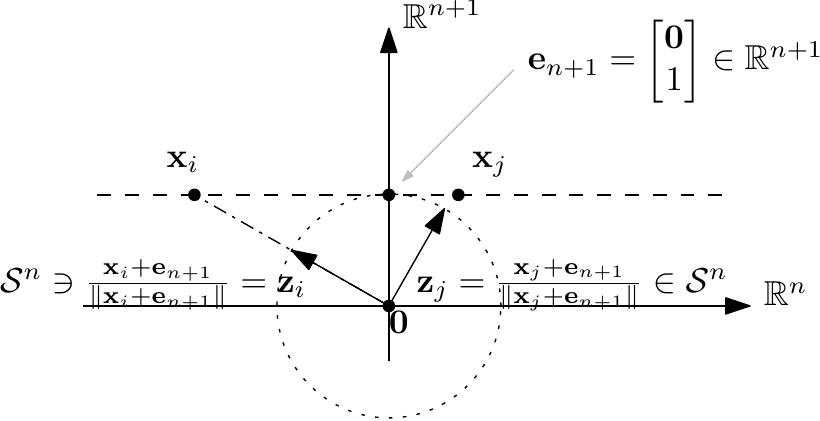}
	\caption{Casting consensus in $\Rn[n]$ as synchronization in $\Sn{n}$.}
	\label{fig:ConsenusRn}
\end{figure}
Notice that $\hmb$ is, in fact, a diffeomorphism between $\Rn[n]$ and $\ConeOpen(\frac{\pi}{2},\embi{n +1})$, with $\hmb^{\sss{-1}} : \ConeOpen(\frac{\pi}{2},\embi{n +1}) \ni \zmb \mapsto \hmb^{\sss{-1}}(\zmb)  \in \Rn[n]$ given by %
\begin{align}
	\hmb^{\sss{-1}}(\zmb) 
	= 
	P\tp
	\left(
		\frac{1}{\innerproduct{\zmb}{\embi{n+1}}} \zmbi{i} -\embi{n+1}
	\right).
	\label{eq:UnitVectorPositionConsensusInverse}
\end{align}
If follows from~\eqref{eq:UnitVectorPositionConsensus} and~\eqref{eq:UnitVectorPositionConsensusInverse} that, for any $\zmbi{i}, \zmbi{j} \in \ConeOpen(\frac{\pi}{2},\embi{n +1})$,
\begin{align}
	&
	\frac{P \xmbi{j} - P \xmbi{i}}{\|P \xmbi{i} + \embi{n+1}\|}
	|_{
		\sss{
				\xmbi{i} = \hmb^{\sss{-1}}(\zmbi{i}), \xmbi{j} = \hmb^{\sss{-1}}(\zmbi{j})
			}
	}
	\\
	& 
	\Scale[0.95]{
		=
		\left(
		\frac{\|P \xmbi{j} + \embi{n+1}\|}{\|P \xmbi{i} + \embi{n+1}\|}
		\frac{P \xmbi{j}}{\|P \xmbi{j} + \embi{n+1}\|}
		-
		\frac{P \xmbi{i}}{\|P \xmbi{i} + \embi{n+1}\|}
		\right)
		|_{
			\substack{
				\sss{
					\xmbi{i} = \hmb^{\sss{-1}}(\zmbi{i})
				}
				\\
				\sss{
					\xmbi{j} = \hmb^{\sss{-1}}(\zmbi{j})
				}
			}
		}
	}
	\\
	& =
	\zmbi{j}
	\frac{
		\innerproduct{\zmbi{i}}{\embi{n+1}}
	}
	{
		\innerproduct{\zmbi{j}}{\embi{n+1}}
	} 
	- \zmbi{i}.
	\label{eq:UnitVectorPositionConsensusAux1}
\end{align}

Let $\xmb = (\xmbi{1},\cdots,\xmbi{N}) \in (\Rn[n])^{\sss{N}}$ and $\zmb = (\zmbi{1},\cdots,\zmbi{N}) \in \ConeOpen(\frac{\pi}{2},\embi{n +1})^{\sss{N}}$ where $\zmb = \hmb^{\sss{N}}(\xmb) \Leftrightarrow \xmb = (\hmb^{\sss{-1}})^{\sss{N}}(\zmb)$.
It holds that, for any $i \in \Nset$, (note that $d\hmb(\xmbi{i}) = \frac{1}
{\|P \xmbi{i} + \embi{n+1}\|} \OP{\hmb(\xmbi{i})} P $)
\begin{align}
	&
	d \hmb(\xmbi{i})
	\fmb_{\sss{i}}(\umb_{\sss{i}}^{\sss{cl}}(t,\xmb))|_{\xmb = (\hmb^{\sss{-1}})^{\sss{N}}(\zmb)}
	=
	\\ 
	{\sss{\eqref{eq:PositionKinematics}}}
	&
	=
	\OP{\hmb(\xmbi{i})}
	\frac{P \Rmat_{\sss{i}} \umb_{\sss{i}}^{\sss{cl}}(t,\xmb) }
	{\|P \xmbi{i} + \embi{n+1}\|}
	|_{\xmb = (\hmb^{\sss{-1}})^{\sss{N}}(\zmb)}
	\\
	{\sss{\eqref{eq:xComponentsDynamics}}}
	&
	=
	\Scale[0.8]{
		\sum\limits_{\sss{j \in \Nset_{\sss{i}}(\sigma(t))}}
		\OP{\hmb(\xmbi{i})}
		w_{\sss{ij}}
		\left(
		\left\|
		P\xmbi{j} - P\xmbi{i}
		\right\|
		\right)
		\frac{P\xmbi{j} - P\xmbi{i}}
		{\|P\xmbi{i} + \embi{n+1}\|}
		|_{\xmb = (\hmb^{\sss{-1}})^{\sss{N}}(\zmb)}
	}
	\\
	{\sss{\eqref{eq:UnitVectorPositionConsensusAux1}}}
	&
	=
	\Scale[0.8]{
		\sum\limits_{\sss{j \in \Nset_{\sss{i}}(\sigma(t))}}
		\OP{\zmbi{i}}
		w_{\sss{ij}}
		\left(
		\left\|
			\frac{\zmbi{j}}{ \innerproduct{\embi{n+1}}{\zmbi{j}} } - \frac{\zmbi{i}}{ \innerproduct{\embi{n+1}}{\zmbi{i}} }
		\right\|
		\right)	
		\left(
			\frac{\innerproduct{\embi{n+1}}{\zmbi{i}}}{\innerproduct{\embi{n+1}}{\zmbi{j}}} \zmbi{j} - \zmbi{i}
		\right)
	}
	\\
	&
	=
	\Scale[0.8]{
		\sum\limits_{\sss{j \in \Nset_{\sss{i}}(\sigma(t))}}
		\frac{\innerproduct{\embi{n+1}}{\zmbi{i}}}{\innerproduct{\embi{n+1}}{\zmbi{j}}}	
		w_{\sss{ij}}
		\left(
			\left\|
				\frac{\zmbi{j}}{ \innerproduct{\embi{n+1}}{\zmbi{j}} } - \frac{\zmbi{i}}{ \innerproduct{\embi{n+1}}{\zmbi{i}} }
			\right\|
		\right)
		\OP{\zmbi{i}}\zmbi{j}
	}
	\\
	& 
	=:
	\sum\limits_{\sss{j \in \Nset_{\sss{i}}(\sigma(t))}}
	\tilde{w}_{\sss{ij}}(\zmbi{i},\zmbi{j})	
	\OP{\zmbi{i}}	
	\zmbi{j}
	\\
	{\sss{\eqref{eq:GeneralForm2}}}
	&
	=:
	\tilde{\fmb}_{\sss{i,\sigma(t)}}(\zmb)
	\label{eq:DerivativeTransformationConsensusRn} 
\end{align}
where, in the one to last step, we defined 
$
	\tilde{w}_{\sss{ij}}(\zmbi{i},\zmbi{j}) := 
	\frac{\innerproduct{\embi{n+1}}{\zmbi{i}}}{\innerproduct{\embi{n+1}}{\zmbi{j}}}	
	w_{\sss{ij}}
	\left(
	\left\|
	\frac{\zmbi{j}}{ \innerproduct{\embi{n+1}}{\zmbi{j}} } - \frac{\zmbi{i}}{ \innerproduct{\embi{n+1}}{\zmbi{i}} }
	\right\|
	\right)
$%
, which satisfies~\eqref{eq:gTildeCase} for $\alpha = \frac{\pi}{2}$ and for $\bar{\numb} = \embi{n+1} \in \Sn{n}$.
We have thus casted this problem in the form~\eqref{eq:GeneralForm2} with $\bm{\nu} \equiv \zmb \in (\Sn{n})^{\sss{N}}$.

\begin{rem}
	Unlike synchronization in $\Sn{2}$ and $\SO[3]$, the unit vectors in this section are by construction contained in a $\frac{\pi}{2}$-cone formed by the unit vector $\embi{n+1}$ (see the co-domain of $\hmb$ in~\eqref{eq:UnitVectorPositionConsensus}).
	Also, note that $\hmb$ in~\eqref{eq:UnitVectorPositionConsensus} may be defined with other vectors other than $\emb_{\sss{n+1}}$, i.e., $\Rn[n] \ni \xmb \mapsto \hmb(\xmb) = \frac{P\xmb + \tilde{\emb}}{\|P\xmb + \tilde{\emb}\|}$ with some $\tilde{\emb} \in \Sn{n}$ satisfying $\innerproduct{\embi{n+1}}{\tilde{\emb}} \ne 0$ also works as an alternative transformation.
\end{rem}
	
	\bibliographystyle{IEEEtran}
	\bibliography{bibliography}

\begin{thebibliography}{10}
\providecommand{\url}[1]{#1}
\csname url@samestyle\endcsname
\providecommand{\newblock}{\relax}
\providecommand{\bibinfo}[2]{#2}
\providecommand{\BIBentrySTDinterwordspacing}{\spaceskip=0pt\relax}
\providecommand{\BIBentryALTinterwordstretchfactor}{4}
\providecommand{\BIBentryALTinterwordspacing}{\spaceskip=\fontdimen2\font plus
\BIBentryALTinterwordstretchfactor\fontdimen3\font minus
  \fontdimen4\font\relax}
\providecommand{\BIBforeignlanguage}[2]{{%
\expandafter\ifx\csname l@#1\endcsname\relax
\typeout{** WARNING: IEEEtran.bst: No hyphenation pattern has been}%
\typeout{** loaded for the language `#1'. Using the pattern for}%
\typeout{** the default language instead.}%
\else
\language=\csname l@#1\endcsname
\fi
#2}}
\providecommand{\BIBdecl}{\relax}
\BIBdecl

\bibitem{lawton2002synchronized}
J.~R. Lawton and R.~W. Beard, ``Synchronized multiple spacecraft rotations,''
  \emph{Automatica}, vol.~38, no.~8, pp. 1359--1364, 2002.

\bibitem{abdessameud2009attitude}
A.~Abdessameud and A.~Tayebi, ``Attitude synchronization of a group of
  spacecraft without velocity measurements,'' \emph{IEEE Transactions on
  Automatic Control}, vol.~54, no.~11, pp. 2642--2648, 2009.

\bibitem{leonard2007collective}
N.~Leonard, D.~Paley, F.~Lekien, R.~Sepulchre, D.~Fratantoni, and R.~Davis,
  ``Collective motion, sensor networks, and ocean sampling,'' \emph{Proceedings
  of the IEEE}, vol.~95, no.~1, pp. 48--74, Jan 2007.

\bibitem{hatanaka2012passivity}
T.~Hatanaka, Y.~Igarashi, M.~Fujita, and M.~Spong, ``Passivity-based pose
  synchronization in three dimensions,'' \emph{IEEE Transactions on Automatic
  Control}, vol.~57, no.~2, pp. 360--375, 2012.

\bibitem{1656474}
W.~Ren, ``Distributed attitude consensus among multiple networked spacecraft,''
  in \emph{American Control Conference}.\hskip 1em plus 0.5em minus 0.4em\relax
  IEEE, June 2006, pp. 1760--1765.

\bibitem{sarlette2009autonomous}
A.~Sarlette, R.~Sepulchre, and N.~E. Leonard, ``Autonomous rigid body attitude
  synchronization,'' \emph{Automatica}, vol.~45, no.~2, pp. 572--577, 2009.

\bibitem{dimarogonas2009leader}
D.~V. Dimarogonas, P.~Tsiotras, and K.~Kyriakopoulos, ``Leader--follower
  cooperative attitude control of multiple rigid bodies,'' \emph{Systems \&
  Control Letters}, vol.~58, no.~6, pp. 429--435, 2009.

\bibitem{bai2008rigid}
H.~Bai, M.~Arcak, and J.~T. Wen, ``Rigid body attitude coordination without
  inertial frame information,'' \emph{Automatica}, vol.~44, no.~12, pp.
  3170--3175, 2008.

\bibitem{tron2012intrinsic}
R.~Tron, B.~Afsari, and R.~Vidal, ``Intrinsic consensus on {SO}(3) with
  almost-global convergence.'' in \emph{Conference on Decision and Control},
  2012, pp. 2052--2058.

\bibitem{chung2013phase}
S.~Chung, S.~Bandyopadhyay, I.~Chang, and F.~Hadaegh, ``Phase synchronization
  control of complex networks of lagrangian systems on adaptive digraphs,''
  \emph{Automatica}, vol.~49, no.~5, pp. 1148--1161, 2013.

\bibitem{bondhus2005leader}
A.~K. Bondhus, K.~Y. Pettersen, and J.~T. Gravdahl, ``Leader/follower
  synchronization of satellite attitude without angular velocity
  measurements,'' in \emph{Conference on Decision and Control and European
  Control Conference}.\hskip 1em plus 0.5em minus 0.4em\relax IEEE, 2005, pp.
  7270--7277.

\bibitem{krogstad2006coordinated}
T.~Krogstad and J.~Gravdahl, ``Coordinated attitude control of satellites in
  formation,'' in \emph{Group Coordination and Cooperative Control}.\hskip 1em
  plus 0.5em minus 0.4em\relax Springer, 2006, pp. 153--170.

\bibitem{oh2014formation}
K.~Oh and H.~Ahn, ``Formation control and network localization via orientation
  alignment,'' \emph{IEEE Transactions on Automatic Control}, vol.~59, no.~2,
  pp. 540--545, 2014.

\bibitem{olfati2006swarms}
R.~Olfati-Saber, ``Swarms on sphere: A programmable swarm with synchronous
  behaviors like oscillator networks,'' in \emph{Conference on Decision and
  Control}.\hskip 1em plus 0.5em minus 0.4em\relax IEEE, 2006, pp. 5060--5066.

\bibitem{moshtagh2007distributed}
N.~Moshtagh and A.~Jadbabaie, ``Distributed geodesic control laws for flocking
  of nonholonomic agents,'' \emph{Transactions on Automatic Control}, vol.~52,
  no.~4, pp. 681--686, 2007.

\bibitem{paley2009stabilization}
D.~A. Paley, ``Stabilization of collective motion on a sphere,''
  \emph{Automatica}, vol.~45, no.~1, pp. 212--216, 2009.

\bibitem{li2014unified}
W.~Li and M.~W. Spong, ``Unified cooperative control of multiple agents on a
  sphere for different spherical patterns,'' \emph{Transactions on Automatic
  Control}, vol.~59, no.~5, pp. 1283--1289, 2014.

\bibitem{sarlette2008global}
A.~Sarlette, S.~E. Tuna, V.~Blondel, and R.~Sepulchre, ``Global synchronization
  on the circle,'' in \emph{Proceedings of the 17th IFAC world congress}, 2008,
  pp. 9045--9050.

\bibitem{sarlette2009synchronization}
A.~Sarlette and R.~Sepulchre, ``Synchronization on the circle,'' \emph{arXiv
  preprint arXiv:0901.2408}, 2009.

\bibitem{dorfler2014synchronization}
F.~D{\"o}rfler and F.~Bullo, ``Synchronization in complex networks of phase
  oscillators: A survey,'' \emph{Automatica}, vol.~50, no.~6, pp. 1539--1564,
  2014.

\bibitem{moreau2004stability}
L.~Moreau, ``Stability of continuous-time distributed consensus algorithms,''
  in \emph{Conference on Decision and Control}, vol.~4.\hskip 1em plus 0.5em
  minus 0.4em\relax IEEE, 2004, pp. 3998--4003.

\bibitem{moreau2005stability}
------, ``Stability of multiagent systems with time-dependent communication
  links,'' \emph{Transactions on Automatic Control}, vol.~50, no.~2, pp.
  169--182, 2005.

\bibitem{zhang2011general}
H.~Zhang, C.~Zhai, and Z.~Chen, ``A general alignment repulsion algorithm for
  flocking of multi-agent systems,'' \emph{IEEE Transactions on Automatic
  Control}, vol.~56, no.~2, pp. 430--435, 2011.

\bibitem{pereira2015Synchronization}
\BIBentryALTinterwordspacing
P.~O. Pereira and D.~V. Dimarogonas, ``Family of controllers for attitude
  synchronization on the sphere,'' \emph{Automatica}, vol.~75, pp. 271 -- 281,
  2017. [Online]. Available:
  \url{http://www.sciencedirect.com/science/article/pii/S0005109816303788}
\BIBentrySTDinterwordspacing

\bibitem{thunberg2014distributed}
J.~Thunberg, W.~Song, E.~Montijano, Y.~Hong, and X.~Hu, ``Distributed attitude
  synchronization control of multi-agent systems with switching topologies,''
  \emph{Automatica}, vol.~50, no.~3, pp. 832--840, 2014.

\bibitem{igarashi2009passivity}
Y.~Igarashi, T.~Hatanaka, M.~Fujita, and M.~W. Spong, ``Passivity-based
  attitude synchronization in $\mathcal{SE}(3)$,'' \emph{IEEE Transactions on
  Control Systems Technology}, vol.~17, no.~5, pp. 1119--1134, 2009.

\bibitem{sepulchre2011consensus}
R.~Sepulchre, ``Consensus on nonlinear spaces,'' \emph{Annual reviews in
  control}, vol.~35, no.~1, pp. 56--64, 2011.

\bibitem{sarlette2009consensus}
A.~Sarlette and R.~Sepulchre, ``Consensus optimization on manifolds,''
  \emph{SIAM Journal on Control and Optimization}, vol.~48, no.~1, pp. 56--76,
  2009.

\bibitem{lin2007state}
Z.~Lin, B.~Francis, and M.~Maggiore, ``State agreement for continuous-time
  coupled nonlinear systems,'' \emph{SIAM Journal on Control and Optimization},
  vol.~46, no.~1, pp. 288--307, 2007.

\bibitem{hespanha2004uniform}
J.~P. Hespanha, ``Uniform stability of switched linear systems: Extensions of
  lasalle's invariance principle,'' \emph{IEEE Transactions on Automatic
  Control}, vol.~49, no.~4, pp. 470--482, 2004.

\bibitem{bacciotti2005invariance}
A.~Bacciotti and L.~Mazzi, ``An invariance principle for nonlinear switched
  systems,'' \emph{Systems \& Control Letters}, vol.~54, no.~11, pp.
  1109--1119, 2005.

\bibitem{mancilla2006extension}
J.~L. Mancilla-Aguilar and R.~A. Garc{\'\i}a, ``An extension of lasalle's
  invariance principle for switched systems,'' \emph{Systems \& Control
  Letters}, vol.~55, no.~5, pp. 376--384, 2006.

\bibitem{fischer2013lasalle}
N.~Fischer, R.~Kamalapurkar, and W.~E. Dixon, ``Lasalle-yoshizawa corollaries
  for nonsmooth systems,'' \emph{IEEE Transactions on Automatic Control},
  vol.~9, no.~58, pp. 2333--2338, 2013.

\bibitem{field1984spacecraft}
M.~Field and D.~Pence, ``Spacecraft attitude, rotations and quaternions,''
  \emph{UMAP Modules}, vol.~5, no.~2, p. 130, 1984.

\bibitem{pereira2016CDCSynchronization}
\BIBentryALTinterwordspacing
P.~O. Pereira, D.~Boskos, and D.~V. Dimarogonas, ``A common framework for
  attitude synchronization of unit vectors in networks with switching
  topology,'' in \emph{2016 IEEE 55th Conference on Decision and Control
  (CDC)}, Dec 2016, pp. 3530--3536. [Online]. Available:
  \url{http://ieeexplore.ieee.org/document/7798799/}
\BIBentrySTDinterwordspacing

\bibitem{goebel2008invariance}
R.~Goebel, R.~G. Sanfelice, and A.~R. Teel, ``Invariance principles for
  switching systems via hybrid systems techniques,'' \emph{Systems \& Control
  Letters}, vol.~57, no.~12, pp. 980--986, 2008.

\bibitem{sanfelice2007invariance}
R.~G. Sanfelice, R.~Goebel, and A.~R. Teel, ``Invariance principles for hybrid
  systems with connections to detectability and asymptotic stability,''
  \emph{IEEE Transactions on Automatic Control}, vol.~52, no.~12, pp.
  2282--2297, 2007.

\end{thebibliography}
	\clearpage	
\end{document}